\documentclass[11pt]{article}
\usepackage{nicematrix}
\usepackage{comment}
\usepackage{amsfonts}
\usepackage{amssymb}
\usepackage{amsthm}%,thmtools, thm-restate}
\usepackage{amsmath}
\usepackage{graphicx}
\usepackage{hyperref}
\hypersetup{
	breaklinks=true,   % splits links across lines
	colorlinks=true,   % displays links as colored text instead of blocks
	pdfusetitle=true,  % \title and \author values into pdf metadata
}
\usepackage{enumerate}
\usepackage{tikz}
\usetikzlibrary{matrix,arrows,patterns,cd}

\usepackage{scalerel}
\usetikzlibrary{svg.path}

\definecolor{orcidlogocol}{HTML}{A6CE39}
\tikzset{
	orcidlogo/.pic={
		\fill[orcidlogocol] svg{M256,128c0,70.7-57.3,128-128,128C57.3,256,0,198.7,0,128C0,57.3,57.3,0,128,0C198.7,0,256,57.3,256,128z};
		\fill[white] svg{M86.3,186.2H70.9V79.1h15.4v48.4V186.2z}
		svg{M108.9,79.1h41.6c39.6,0,57,28.3,57,53.6c0,27.5-21.5,53.6-56.8,53.6h-41.8V79.1z M124.3,172.4h24.5c34.9,0,42.9-26.5,42.9-39.7c0-21.5-13.7-39.7-43.7-39.7h-23.7V172.4z}
		svg{M88.7,56.8c0,5.5-4.5,10.1-10.1,10.1c-5.6,0-10.1-4.6-10.1-10.1c0-5.6,4.5-10.1,10.1-10.1C84.2,46.7,88.7,51.3,88.7,56.8z};
	}
}

\newcommand\orcidicon[1]{\href{https://orcid.org/#1}{\mbox{\scalerel*{
				\begin{tikzpicture}[yscale=-1,transform shape]
					\pic{orcidlogo};
				\end{tikzpicture}
			}{|}}}}
\usepackage{color}
\usepackage[mathscr]{euscript}

\usepackage{soul,todonotes}

\usepackage{leftidx}
\usepackage{mathrsfs}
\usepackage{hyperref}
\usepackage{datetime}

\numberwithin{equation}{section}
\newcommand{\id}{\mathrm{id}}

\newcommand{\ii}{{\rm i}}

\newcommand{\vol}{{\rm vol}}

\newcommand{\dd}{{\rm d}}

\DeclareMathOperator{\supp}{supp}

\DeclareMathOperator{\End}{End}
\DeclareMathOperator{\Hom}{Hom}

\newcommand{\ogth}{{\mathfrak o}}
\newcommand{\tgth}{{\mathfrak t}}

\newcommand{\Mb}{{\boldsymbol{M}}}
\newcommand{\Mbb}{{\boldsymbol{\EuScript{M}}}}

\newcommand{\1}{1\!\!1}
\newcommand{\beq}{\begin{equation}}
\newcommand{\ene}{\end{equation}}

\newcommand{\RR}{\mathbb{R}}
\newcommand{\CC}{\mathbb{C}}

\newcommand{\NN}{\mathbb{N}}

\newcommand{\BB}{\mathbb{B}}
\newcommand{\PP}{\mathbb{P}}

\DeclareMathOperator{\WF}{WF}

\let\Re\relax
\DeclareMathOperator{\Re}{Re}

\DeclareMathOperator{\ran}{ran}

\newcommand{\bstar}{{\bar{*}}}

\newcommand{\pc}{\textnormal{pc}}
\newcommand{\fc}{\textnormal{fc}}
\renewcommand{\sc}{\textnormal{sc}}

\newtheorem{thm}{Theorem}[section]
\newtheorem{lemma}[thm]{Lemma}

\newtheorem{prop}[thm]{Proposition}

\newtheorem{rem}[thm]{Remark}

\theoremstyle{definition}
\newtheorem{defn}[thm]{Definition}

\newdateformat{daymonthyear}{\THEDAY \, \monthname[\THEMONTH] \THEYEAR}
\newcommand{\Af}{\mathscr{A}}
\newcommand{\Bf}{\mathscr{B}}
\newcommand{\Cf}{\mathscr{C}}

\newcommand{\Sf}{\mathscr{S}}

\newcommand{\Uf}{\mathscr{U}}

\newcommand{\Wf}{\mathscr{W}}
\newcommand{\Zf}{\mathscr{Z}}

\newcommand{\Ac}{\mathcal{A}}
\newcommand{\Bc}{\mathcal{B}}
\newcommand{\Cc}{\mathcal{C}}
\newcommand{\Dc}{\mathcal{D}}
\newcommand{\Ec}{\mathcal{E}}
\newcommand{\Fc}{\mathcal{F}}
\newcommand{\Hc}{\mathcal{H}}

\newcommand{\Kc}{\mathcal{K}}
\newcommand{\LL}{\mathcal{L}}
\newcommand{\Mc}{\mathcal{M}}
\newcommand{\Nc}{\mathcal{N}}

\newcommand{\Rc}{\mathcal{R}}
\newcommand{\Sc}{\mathcal{S}}

\newcommand{\Uc}{\mathcal{U}}
\newcommand{\Vc}{\mathcal{V}}
\newcommand{\Wc}{\mathcal{W}}
\newcommand{\Zc}{\mathcal{Z}}

\def\fm{{\mathfrak m}}
\def\fh{{\mathfrak h}}

\newcommand{\dhadj}[1]{\leftidx{^{\bstar}}{#1}{}}
\newcommand{\hadj}[1]{\leftidx{^{\dagger}}{#1}{}}
\newcommand{\sadj}[1]{\leftidx{^{\star}}{#1}{}}

\newcommand{\abs}[1]{{\left\vert #1 \right\vert}}
\newcommand{\norm}[1]{{\left\Vert #1 \right\Vert}}
\newcommand{\IP}[1]{{\left\langle #1 \right\rangle}}
\newcommand{\cv}[1]{{\underline{ #1 }}}

\newcommand{\rfhgho}{RFHGHO}

\newcommand{\aux}{\textnormal{aux}}

\DeclareMathOperator{\Diag}{Diag}

\NiceMatrixOptions{
cell-space-limits=1.5pt,
custom-line={letter=I, tikz={dashed,thick},total-width=\pgflinewidth},
custom-line={letter=o, tikz={dotted,thick},total-width=\pgflinewidth},
custom-line={command=hdashedline,ccommand=cdashedline,tikz={dashed,thick},total-width=\pgflinewidth},
custom-line={command=hddottedline,ccommand=cddottedline,tikz={dotted,thick},total-width=\pgflinewidth}
}
\newenvironment{pnmatrix}[1]{\left(\begin{NiceArray}{#1}}{\end{NiceArray}\right)}
\begin{document}

\title{Coupled Proca theories\\ \Large Green-hyperbolicity, quantization and applications to polarization measurement}
%The charged Proca field in an external electromagnetic field: Green-hyperbolicity, quantization and Hadamard states}
\author{Christopher J. Fewster{\orcidicon{0000-0001-8915-5321}}${}^{1,2}$\thanks{chris.fewster@york.ac.uk}\hspace{0.25em} and Christiane K. M. Klein{\orcidicon{0000-0003-2398-7226}}${}^{1}$\thanks{christiane.klein@york.ac.uk}
\\[6pt]  
	\small ${}^{1}$Department of Mathematics, Ian Wand Building, Deramore Lane,
\break University of York, \\ \small York YO10 5GH, United Kingdom.\\[4pt]
	\small ${}^{2}$York Centre for Quantum Technologies, University of York, Heslington, York YO10 5DD, United Kingdom.}
    
%\affil{\small Department of Mathematics, University of York, Heslington, York YO10 5DD, United Kingdom.}
%\affil{\small York Centre for Quantum Technologies, University of York, Heslington, York YO10 5DD, United Kingdom.}
 
\date{14 November 2025}

\maketitle
 
\begin{abstract}
The Proca field describes a massive relativistic spin-$1$ particle  and was originally formulated in Minkowski spacetime. Here we consider a variety of generalizations in globally hyperbolic spacetimes, including couplings between a number of Proca fields via a mass-matrix, the charged Proca field with arbitrary magnetic moment in an arbitrary external electromagnetic field, and a Proca--Klein--Gordon theory with a spacetime-dependent bilinear coupling. The equations are analysed using a general auxiliary field method, introduced here, which provides practical criteria for showing that a given operator is (semi)-Green-hyperbolic. The method goes beyond what is achieved in existing analyses of deformed equations, which for example place restrictions on the magnetic moment and electromagnetic potential that can be coupled to a Proca field.
The theories considered can be quantized following a common pattern and all the examples treated in this work admit Hadamard states on any globally hyperbolic spacetime. 

As an application, the Proca--Klein--Gordon system is used to develop a measurement scheme sensitive to the Proca polarization, using a Klein--Gordon field as the probe. For a suitable family of $n$-particle Proca states, the leading-order probe response accords with Malus' law,  confirming that this system acts as a polarization-sensitive detector. 
\end{abstract}

\section{Introduction}
\label{sec:Intro}

Proca introduced the field equation that now bears his name in 1936~\cite{Proca:1936}  in the search for a relativistic  wave equation that could describe charged massive particles with spin and nonzero magnetic moment, while avoiding the perceived defect of the negative energy states  present in Dirac's theory.  Although the hope that this theory would describe electrons and positrons was futile, nonetheless Proca's equation has taken an important place in fundamental physics. Specifically, electroweak symmetry breaking produces charged massive vector bosons $W^\pm$ and a neutral massive vector boson $Z^0$, plus massless photons. The neutral Proca equation provides an effective description of the $Z^0$ boson and composite particles such as the $\omega$, $K^0$, and $\rho^0$-mesons, and models of photons with small nonzero mass, while the charged Proca equation describes the $W^\pm$ bosons, and  $K^\pm$ and $\rho^\pm$-mesons~\cite{ParticleDataGroup:2024}.

 In this paper, we discuss the solution theory and quantization of a number of variants of the Proca equation, typically coupled to other fields on general globally hyperbolic spacetimes. Our main goals are to give full treatments of the charged Proca field with arbitrary magnetic moment coupled to an external electromagnetic field, and also the charged Proca field coupled to a dynamical scalar field. These models have applications in the measurement theory of quantum fields~\cite{FewVer_QFLM:2018} that will be discussed here and elsewhere.

To understand why the task at hand is nontrivial, it is useful to recall how the standard uncharged Proca equation is solved (see e.g.~\cite{Few&Pfen03}), working in Minkowski space for simplicity and using signature convention $+---$. There, the inhomogeneous Proca equation is
\begin{equation}
    \nabla^\mu (\nabla_\mu W_\nu - \nabla_\nu W_\mu)+ m^2 W_\nu = J_\nu,
\end{equation}
where $m>0$ is the mass. This equation is not hyperbolic, but can be solved using a trick:
taking the divergence, the first term vanishes and one obtains the constraint 
\begin{equation}
    m^2 \nabla^\nu W_\nu = \nabla^\nu J_\nu,
\end{equation}
which is satisfied by all solutions. Substituting back and commuting derivatives, the Proca equation can be rearranged into the normally hyperbolic equation
\begin{equation}\label{eq:trick}
    (\nabla^\mu \nabla_\mu +m^2)W_\nu = J_\nu + m^{-2}\nabla_\nu \nabla^\mu J_\mu 
\end{equation}
In this way, advanced and retarded Green operators may be constructed for the Proca equation using those of the covector Klein--Gordon equation, acting on the modified source term in~\eqref{eq:trick}.

The same trick continues to work for the Proca equation in general globally hyperbolic spacetimes, in which the inhomogeneous equation is conveniently written using differential forms as
\begin{equation}
	-\delta \dd W+ m^2 W = J,
\end{equation}
where $\dd$ is the exterior derivative and $\delta$ the codifferential, making use of the fact that $\delta^2W=0$. 
However, the trick is fragile against additional couplings to the field. For example, the charged Proca field in an external electromagnetic vector potential $A_\mu$ obeys 
\begin{equation}\label{eq:chargedProca}
	-\delta_A \dd_A W + \ii q \kappa F^{\bullet} W + m^2 W = J
\end{equation}
where $q$ is the charge, $\dd_A$ and $\delta_A$ are minimally coupled generalizations of the exterior derivative and codifferential (see Sect.~\ref{sec:charged_Proca}),
$F=\dd A$ is the background field strength and the parameter $\kappa$ controls the magnetic moment $(\kappa+1)q/(2m)$, with $\kappa=0$ corresponding to minimal coupling. We have also written $F^{ \bullet} W=F_{\mu}^{\phantom{\mu}\nu}W_\nu$. Applying $\delta_A$ leads to the equation
\begin{equation}\label{eq:prenonhyp}
	m^2\delta_A W = \ii q(\kappa-1)F^{\bullet\bullet}\dd_A W + \ii q \kappa j^{ \bullet} W + \delta_A J ,
\end{equation}
where $j=-\delta F$ is the background current, $j^{ \bullet} W=j_\mu W^\mu$, 
and we use the notation $F^{\bullet\bullet}H=\tfrac{1}{2}F^{\mu\nu}H_{\mu\nu}$ (note that $\delta_A^2\neq 0$ when $F\neq 0$). Substituting back into~\eqref{eq:chargedProca}, one obtains
\begin{equation}\label{eq:nonhyp}
	K_A W + \ii q(\kappa-1)\dd_A(F^{\bullet\bullet}\dd_A W)+
	\textnormal{lower order} = (1-m^{-2}\dd_A \delta_A)J,
\end{equation}
where $K_A = -(\dd_A\delta_A+\delta_A\dd_A)+m^2$ is the $1$-form
Klein--Gordon operator  minimally coupled to $A$. Unless $\kappa=1$ or $F\equiv 0$, Eq.~\eqref{eq:nonhyp} is not normally hyperbolic, owing to the additional second order derivatives of $W$. Therefore the strategy used for the uncharged Proca field generally fails unless $\kappa=1$.

In fact, the value $\kappa=1$ is singled out for other reasons. For example, physical $W^\pm$-bosons have $\kappa=1$ at the tree-level value in the standard model after symmetry breaking. Corben and Schwinger~\cite{CorbenSchwinger:1940} noted that $\kappa=1$ is the only value for which the Noether current associated with global $U(1)$ invariance falls into two separately conserved terms with good physical interpretations. Moreover, there are systematic analyses of deformations of physical equations which only accommodate the $\kappa=1$, $j=0$ theory~\cite{KaparulinLyakhovichSharapov, CorteseRahmanSivakumar, Rahman:2020}. 

However, it would be puzzling if the free Proca theory could not be formulated consistently for other values of $\kappa$. Presently available measurements of the W-boson magnetic moment indicate a value $\kappa+1= 2.22^{+0.20}_{-0.19}$~\cite{ParticleDataGroup:2024}, owing to radiative corrections. It would be strange indeed if one could not formulate a free model with the physical value  of $\kappa$ as the basis of perturbative calculations. In particular, it would contradict the (generalized) principle of perturbative agreement \cite{HollandsWald:2005,DragoHackPinamonti:2017} which asserts the freedom to move quadratic terms between the `free' and `interaction' parts of the action.

In this paper, we  introduce an auxiliary field method that allows us to treat a variety of Proca and coupled Proca systems, including the charged Proca field for any $\kappa$, and which we expect to have wider applications. The idea is to augment the system of equations to one which is at least semi-Green-hyperbolic. Here, a differential operator $P$ acting on sections of a vector bundle over spacetime will be called semi-Green-hyperbolic if it has advanced and retarded Green operators.\footnote{Note that some authors use `Green-hyperbolic' to mean what we call `semi-Green-hyperbolic'.} Normally hyperbolic operators are semi-Green-hyperbolic, but there are numerous examples that are not normally hyperbolic, including the  neutral Proca equation. If the formal dual of $P$ is also semi-Green-hyperbolic then $P$ is called Green-hyperbolic, which has numerous consequences~\cite{Baer:2015}. Among other things one may construct a linear bosonic quantum field theory (QFT) governed by $P$ --  usually this is done under additional mild conditions  such as formal hermiticity with respect to a sesquilinear form induced by a bundle metric but here we show that only Green-hyperbolicity is required.  This construction, described in Sect.~\ref{subsec:General quantization}, is new in this generality but has been applied in \cite{KleinHafner:2025}.

%Returning to the issue of determining semi-Green-hyperbolicity, we will develop an auxiliary field method, 
In our auxiliary field method, the operator of interest is written in the form $P=\pi Q D$, where 
$D$ and $\pi$ are differential operators to and from an enlarged vector bundle acted on by a semi-Green-hyperbolic operator $Q$. We will give conditions under which $P$ can be shown to be semi-Green-hyperbolic with retarded/advanced Green operators given by
$E_P^\pm = \pi E_Q^\pm D$. 
The operators $\pi$, $Q$, and $D$ are not uniquely specified, but may be chosen in different ways.

The auxiliary field framework turns out to be versatile and general. Indeed every semi-Green-hyperbolic operator can be  expressed trivially in this way, but we will also give a number of interesting examples whose (semi)-Green-hyperbolicity is not immediately evident. The standard Proca equation is one, as are systems of Proca fields with varying mass matrix, and the complex Proca field in an external electromagnetic field for arbitrary $\kappa$ and $j$. 
%We will discuss \hl{later} how our method avoids the restrictions found in~\hl{refs}. 
The difference between our approach and that of \cite{KaparulinLyakhovichSharapov, CorteseRahmanSivakumar, Rahman:2020},  which require $\kappa=1$ and $j=0$, appears to lie in the fact that \cite{KaparulinLyakhovichSharapov, CorteseRahmanSivakumar, Rahman:2020} require consistency at all orders in the coupling constant, $q$ in this case, while we solve the full Proca equation -- which is quadratic in $q$ -- from the start.

Finally, we will apply the method to a system consisting of a Proca field bilinearly coupled to a Klein--Gordon field of possibly different mass, and with a spacetime-variable coupling strength given by a $1$-form field. Turning to the quantized theories, we investigate conditions under which they admit Hadamard states, using a recent generalization of the Hadamard condition to `decomposable' Green-hyperbolic operators~\cite{Fewster:2025a}. All the examples mentioned admit Hadamard states in general globally hyperbolic spacetimes.

The motivation for this work arises in the  measurement theory for quantum fields introduced in~\cite{FewVer_QFLM:2018}, in which one measures a target `system' QFT by coupling it to a `probe' QFT that is prepared independently from the system at early times. Measurements of probe observables at late times provide indirect measurements of system observables, and one may also derive state update results appropriate to selective or nonselective measurements. These updates are consistent with causality in various ways; in particular the framework automatically excludes `impossible measurements' of the type first discussed by Sorkin~\cite{sorkin1993impossible}, as shown in~\cite{BostelmannFewsterRuep:2020} -- see~\cite{PapageorgiouFraser:2023b} for more discussion and references. Further developments appear in~\cite{FewsterJubbRuep:2022} and a recent paper~\cite{MandryschNavascues:2024}, which shows how Gaussian-modulated measurements of smeared fields may be implemented. The framework has also proved useful as a starting-point for a discussion of von Neumann algebra `type change' in~\cite{FewsterJanssenLoveridgeRejznerWaldron:2024}, building on and generalising~\cite{chandrasekaranAlgebraObservablesSitter2023}.

The Proca field has a number of features that make it a good example for  measurement models. In particular, it has polarization states due to its spin. Real quantum optics experiments are performed with photons rather than Proca particles, of course, but the Proca theory avoids the technical complications associated with gauge freedom in electromagnetism, and charged Proca fields can be manipulated by external electromagnetic fields. By comparison, a device such as a Faraday rotator for photon polarization \cite[Section 8.11.2]{Hecht} relies on the properties of electromagnetism in media, which would add a further layer of complication to any model. In short, coupled Proca systems provide a good testing ground for ideas that could be implemented for electromagnetism at a later stage. 

As a case study, we show how the Proca--Klein--Gordon system can be used to model a polarization-sensitive detector that responds, at leading order, according to Malus' law~\cite{Malus:1809}: for suitable transversely polarized $n$-particle single mode states the probe response is proportional to the squared cosine of the angle between the polarization direction and the direction picked out by the coupling. 
 In more detail, given a qubit state $\sigma\in\CC^2$ obeying the normalisation condition $\sigma^\dagger\sigma =1$ and a square-integrable function $s$ of $3$-momentum peaked near some $\cv{k}_0$ (without loss of generality pointing in the $z$-direction), we define $n$-particle single mode states with transverse polarization determined by $\sigma$. We define $\alpha_{\rm max}$ to be the maximum angle between the $3$-momentum in the support of $s$ and $\cv{k}_0$, which therefore measures the extent to which these states are collimated along the axis.  The interaction coupling the scalar field $\phi$ to the Proca field $W$ is of the form $v^\mu(\phi \overline{W}_\mu + \overline{\phi}W_\mu)$,
 where $v$ is a compactly supported real vector field of a constant direction in the $xy$-plane at angle $\theta$ to the $x$-axis. We define a probe observable which is, up to a calibration constant, a simplified model for a smearing of the scalar condensate $ {:}|\phi|^2{:}$. Its corresponding induced observable (using the Minkowski vacuum as the probe preparation state) has an expectation in the $n$-particle state given by
\begin{equation}
    n |\sigma^\dagger e_\theta|^2 CF \lambda^2 +O(\lambda^{4})
\end{equation}
as the coupling strength $\lambda\to 0$. Here 
$e_\theta=(\cos\theta,\sin\theta)\in \CC^2$, while
 $\cos^2\alpha_{\rm max}\le C\le 1$ is a collimation factor accounting for the fact that a particle described by a wave packet cannot be perfectly collimated, and $F$ is a form factor independent of $\sigma$ and $\theta$. At leading order, and when $\alpha_{\rm max}\ll \pi/2$, the probe response is largely determined by the factor $|\sigma^\dagger e_\theta|^2$, which may be recognised as the idealised quantum mechanical probability for polarization state $\sigma$ to pass a perfect linear polarizer oriented  along  $4$-vector $v=(0, e_\theta,0)$. In the case that $\sigma$ is real, and therefore represents a linear polarization state, the probability reduces to the squared cosine of the difference in the polarization angles, which is the original form of Malus' law. The form factor $F$ accounts for the extent to which the particles hit the detector. It turns out that $F$ is equal to the leading order response of the same probe linearly coupled to a scalar field in a suitable state.

\paragraph{Further literature and remarks}
  
In the context of the algebraic approach to quantum field theory, the quantized Proca theory was first discussed by Furlani~\cite{Furlani:1999}, who quantized the theory in the style of Dimock's works~\cite{Dimock1980, Dimock92}, obtaining Weyl algebras; see also~\cite{Furlani:1997} for the more specific case of ultrastatic spacetimes with compact Cauchy surfaces. Later,  Strohmaier proved the Reeh--Schlieder property for ground and KMS states of the Proca (and other) fields in stationary spacetimes~\cite{Stroh:2000}, while Dappiaggi \cite{Dappiaggi:2011} formulated the Proca field as a locally covariant theory in the sense of~\cite{BrFrVe03}, making it possible to apply Sanders' general results on the Reeh--Schlieder property~\cite{Sanders_ReehSchlieder}. Likewise, the locally covariant formulation would allow one to obtain the split property following~\cite{Few_split:2015}. More recently, Schambach and Sanders studied the zero-mass limit~\cite{SchambachSanders:2018}, also including external sources. A general presentation of the Proca equation and its quantization can be found in~\cite{BeniniDappiaggi:2015}.  

Turning to the states of the theory, Hadamard states for the Proca field were defined by Fewster and Pfenning (FP)~\cite{Few&Pfen03} using a slightly indirect method, which expressed the Hadamard condition in terms of a Hadamard bisolution for the $1$-form Klein--Gordon equation. A more direct approach was suggested by Moretti, Murro and Volpe (MMV) in \cite{MorettiMurroVolpe:2023}; unfortunately their paper has a gap, which was identified and fixed in~\cite{Fewster:2025b} using polarization set methods. 
The MMV Hadamard condition constrains the wavefront set of the Proca two-point function, which is preferable to the less direct FP approach. MMV were able to prove a partial equivalence between the two definitions, which has  now been established in full~\cite{Fewster:2025a}. In fact~\cite{Fewster:2025a} also extends ideas from MMV to provide a unified and general approach to the Hadamard condition for a class of Green-hyperbolic operators well beyond the Proca equation. We will expand on this summary below in Sect.~\ref{sec:states}.

(Semi)-Green-hyperbolic operators play a central role in this work and others, as they are a cornerstone of the quantization procedure of linear theories in algebraic quantum field theory. Extensions of this concept, which may for example be useful for the quantization of non-local theories, can for instance be found in~\cite{LechnerVerch:2015, Fewster:2023}.

Lastly, an engaging account of Malus' life, particularly with respect to his 1809 paper on polarization, can be found in~\cite{KahrClaborn:2008}.

\paragraph{Structure}
 In Sect.~\ref{sec:charged_Proca}, we introduce the charged Proca equation and demonstrate in more detail that it cannot be solved with the trick for the neutral Proca equation presented above. 
In Sect.~\ref{sec:aux field}, we introduce the notion of semi-Green-hyperbolic operators in more detail and develop the auxiliary field method, taking the neutral Proca field as a starting point. We briefly discuss how the auxiliary field method allows one to draw conclusions about the $\Vc^\pm$-decomposability of the operator. 
The auxiliary field method is used in  Sect.~\ref{sec:appli} to show the semi-Green-hyperbolicity of a system of Proca fields with variable mass matrix, the charged Proca field in a background electromagnetic field for any $\kappa\in \RR$, and for the linearly coupled Proca-scalar system. We give an alternative treatment of the neutral Proca field to show the non-uniqueness of the construction in the auxiliary field framework. 
 Section~\ref{sec:quantum} treats the quantization of theories governed by a Green-hyperbolic operator. Apart from a presentation of the usual framework for formally hermitian operators, we develop a scheme relying solely on Green hyperbolicity in Sect.~\ref{subsec:General quantization}. In Sect.~\ref{sec:algobs}, we discuss the quantization of the examples discussed in Section~\ref{sec:appli}. Next, in Sect.~\ref{sec:states}, we briefly review how Hadamard states can be identified for the quantum theories corresponding to real and formally hermitian Green-hyperbolic operators, and conclude the existence of Hadamard states for all the example cases  of Sect.~\ref{sec:appli}.
After recalling the basic principles of the measurement scheme by Fewster \& Verch \cite{FewVer_QFLM:2018}, the particular example of the coupled Proca-scalar system is used in Sect.~\ref{sec:Malus} to set up a  toy model of a polarization-sensitive detector. To demonstrate the viability of the model, we verify that it follows (the generalised) Malus' law as described above. 
Concluding remarks are given in Sect.~\ref{sec:Conc}.

\paragraph{Notation} The symbol $\subset$ includes the possibility of equality; when alternatives are labelled 
by stacked symbols, e.g., $\pm$, they are labelled in order from top to bottom. For example, in a time-oriented Lorentzian spacetime, $J^\pm(S)$ denotes the set of points reachable from a spacetime subset $S$ by future/past-directed causal curves. If $B$ is a complex vector bundle, $\Gamma^\infty_\Sc(B)$ represents the smooth sections of $B$ with support system $\Sc$~\cite{Baer:2015}, where $\Sc=0/\fc/\pc/\sc$ stands for compact/future-compact/past-compact/spatially compact supports, respectively. Here, a subset $S$ is spatially compact if $S$ is closed and $S\subset J(K)$ for compact $K$; $S$ is future/past-compact if $S\cap J^\pm(x)$ is compact for every point $x$. 
For a bundle $B\xrightarrow{\pi}M$, $\dot B$ denotes the bundle with the zero-section removed. 
 Physical units in which $\hbar=c=1$ are employed.

\section{The charged Proca equation}\label{sec:charged_Proca}

Let $(M,g)$ be a $4$-dimensional Lorentzian spacetime, with signature ${+}{-}{-}{-}$, that is oriented and time-oriented. For the moment we suppress notation for the orientations. The Proca field, of mass $m$ and charge $q$ in an external electromagnetic field $A_\mu$, is described by a covector field $W_\mu$ with dynamics given by Lagrangian $\LL=\LL_0+\kappa\LL_{\textnormal{mag}}$, where
\begin{align}
	\LL_0 &= -\frac{1}{2} \overline{H}_{\mu\nu} H^{\mu\nu} + m^2 \overline{W}_\mu W^\mu
	= -(\nabla_\mu \overline{W}_\nu -\ii q A_\mu  \overline{W}_\nu ) H^{\mu\nu}+ m^2 \overline{W}_\mu W^\mu, \nonumber\\
	\LL_{\textnormal{mag}}&=\ii \overline{W}_\mu W_\nu F^{\mu\nu} 
\end{align}
and $\kappa$ is a fixed real parameter determining the anomalous magnetic dipole moment (see, e.g., equation~(12) in~\cite{CorbenSchwinger:1940}, modulo differences in conventions).
Here, $H$ is given in terms of $W$ by $H_{\mu\nu} = D_{\mu}W_\nu - D_\nu W_\mu$, with $D_\mu=\nabla_\mu+\ii q A_\mu$ and $F_{\mu\nu} = \nabla_\mu A_\nu-\nabla_\nu A_\mu$.
The resulting field equation is  
\begin{equation}\label{eq:chargedProca_index}
	D^\mu H_{\mu\nu}  + m^2 W_\nu + \ii\kappa F_{\nu\mu}W^\mu = 0.
\end{equation} 

It is convenient to reformulate~\eqref{eq:chargedProca_index} using differential forms, writing  $\Lambda^pM$ ($0\le p\le 4$) for the bundle of complex $p$-forms over $M$.  One has the usual exterior derivative $\dd:\Gamma^\infty(\Lambda^{p}M)\to\Gamma^\infty(\Lambda^{p+1}M)$ and, as $M$ is oriented, the Hodge isomorphism $\star$ of $p$-forms with $(4-p)$-forms, obeying
\begin{equation}
	\omega\wedge\star \eta = \frac{1}{p!}\omega_{\alpha_1\cdots\alpha_p}\eta^{\alpha_1\cdots\alpha_p}\vol,
\end{equation}
with $\vol$ being the metric volume $4$-form. 
This is used to define the 
codifferential by 
\begin{equation}
	\delta\omega = (-1)^{\deg \omega} \star^{-1}\dd\star \omega,
\end{equation}
for any form field $\omega$. Writing a typical $p$-form $\omega$ as an antisymmetric tensor, one has
\begin{equation}
	(\dd \omega)_{\alpha_0\cdots\alpha_{p}}= \sum_{j=0}^{p} (-1)^j \nabla_{\alpha_j} \omega_{\alpha_0\ldots \alpha_{p}\setminus\alpha_j},\qquad
	(\delta\omega)_{\alpha_1\ldots\alpha_{p-1}} = -\nabla^\beta \omega_{\beta\alpha_1\ldots\alpha_{p-1}}, 
\end{equation}
where $\nabla$ is the Levi--Civita connection and the notation $\omega_{\alpha_0\ldots \alpha_{p}\setminus\alpha_j}$ indicates that the index $\alpha_j$ is omitted.  The wedge product of a $p$-form $\omega$ and $q$-form $\eta$ is
\begin{equation}
        (\omega\wedge \eta)_{\alpha_1\ldots\alpha_p\beta_1\ldots\beta_q}= \frac{(p+q)!}{p!q!} \omega_{[\alpha_1\ldots\alpha_p}\eta_{\beta_1\ldots\beta_q]}.
\end{equation}

We define the operator $\Box=-(\dd\delta+\delta\dd)$ on $p$-forms (i.e., minus the Laplace--de Rham operator) which is normally hyperbolic, because $\Box$ and $g^{\mu\nu}\nabla_\mu \nabla_\nu$ differ by zeroth-order curvature terms (and agree for $p=0$).

Given the external electromagnetic field $A\in \Gamma^\infty(T^*M)$, we may define a modified exterior derivative and codifferential by 
\begin{equation}
	\dd_A\omega = \dd\omega + \ii q A\wedge \omega, \qquad \delta_A\omega = (-1)^{\deg \omega} \star^{-1}\dd_A\star \omega,
\end{equation}
for any form field $\omega$.
 Noting that $\delta_A$ can be expressed as
\begin{equation}
    (\delta_A \omega)_{\alpha_1\ldots\alpha_{p-1}} = -(\nabla^{\beta}+\ii q A^\beta)\omega_{\beta\alpha_1\ldots\alpha_{p-1}},
\end{equation}
the inhomogeneous version of Equation~\eqref{eq:chargedProca_index} becomes
\begin{equation}\label{eq:fProca}
	-\delta_A \dd_A W + \ii q \kappa F^{\bullet} W + m^2 W = J,
\end{equation}
as stated in~\eqref{eq:chargedProca}, with $(F^{\bullet} W)_\alpha = F_\alpha^{\phantom{\alpha}\beta}W_\beta$. 

The charged Proca operator $P_A = 	-\delta_A \dd_A   + \ii q \kappa F^\bullet + m^2$ is formally hermitian with respect to the hermitian pairing
\begin{equation}\label{eq:hermpair}
	\langle W, Z \rangle_p =(-1)^p\int_M \overline{W}\wedge \star Z
\end{equation} 
of $p$-forms. To show this, note that $-\delta_A$ is the formal hermitian adjoint of $\dd_A$, i.e.,
for a $p$-form $W$ and $(p-1)$-form $Z$ having compactly intersecting supports, 
\begin{align}
	\langle Z, \delta_A W\rangle_{p-1} &=(-1)^{p-1} \int_M(\overline{Z}\wedge \star \delta W+ (-1)^p \ii q\overline{Z}\wedge A\wedge\star W) \nonumber\\
    &= (-1)^{p-1}\int_M(\overline{\dd Z}\wedge\star W-\dd(\overline{Z}\wedge \star W)+\overline{\ii qA\wedge Z}\wedge\star W)\\\nonumber
    &= (-1)^{p-1}\int_M \overline{\dd_A Z}\wedge\star W= -\langle \dd_A Z, W\rangle_p \, .
\end{align}
The calculation makes use of the identities
\begin{align}
\dd (\overline{Z}\wedge\star W) &= \overline{\dd Z}\wedge \star W + (-1)^{p-1} \overline{Z}\wedge \star \star^{-1}\dd\star W = \overline{\dd Z}\wedge \star W - 
\overline{Z}\wedge \star \delta W\\
Z\wedge A&=(-1)^{p-1}A\wedge Z
\end{align}
for $p$-form $W$, $(p-1)$-form $Z$, and $1$-form $A$ as well as Stokes' theorem  and compactness of $\supp Z\cap\supp W$.

Consequently,  $\langle \delta_A\dd_A W,Z\rangle_1=-\langle \dd_AW,\dd_A Z\rangle_2 = \langle W,\delta_A\dd_A Z\rangle_1$, for any $W$ whose support intersects that of $Z$ compactly, and
\begin{equation}
	\langle W, P_A Z\rangle_1 -\langle P_A W, Z\rangle_1 =
	-\ii q \kappa \int_M \left(\overline{W}_\mu F^{\mu\nu}Z_\nu + \overline{F^{\nu\mu}W_\mu}Z_\nu\right)=0
\end{equation}
because $F$ is real and antisymmetric.

As discussed in the introduction, the $A=0$  Proca equation can be solved by applying $\delta$ to obtain $\delta W=-m^{-2}\delta J$ and then substituting back to find a normally hyperbolic equation for $W$.
This trick depends on the fact that $\delta^2=0$. However, the  modified operators do not square to zero; instead, one has
\begin{align}\label{eq:dAsquared}
	\dd_A^2\omega 
%	 \dd(\dd\omega +\ii q A\wedge\omega) + \ii q A\wedge(\dd+\ii q A\wedge \omega) \\
%	&= \ii q (\dd (A\wedge \omega) + A\wedge \dd\omega)\\
	&= \ii q F\wedge \omega	
\end{align}
and consequently,
\begin{equation}\label{eq:deltaAsquared}
	\delta_A^2 \omega = -\ii q \star^{-1} (F\wedge \star \omega),
\end{equation}
which reduces to $\delta_A^2\omega=-\ii q F^{\bullet\bullet}\omega$ if $\omega$ is a $2$-form, where $F^{\bullet\bullet}\omega =\tfrac{1}{2} F^{\mu\nu}\omega_{\mu\nu}$. The identity 
$\delta_A \omega =-(\nabla^\mu + \ii q A^\mu)\omega_\mu$ for $1$-forms implies 
\begin{equation}\label{eq:usefulidentity}
	\delta_A(F^{ \bullet} W) = -j^{ \bullet} W - F^{\bullet\bullet}\dd_A W ,
\end{equation}
where $j=-\delta F$ and $j^\bullet W = j^\mu W_\mu$. Using these facts, the modified codifferential of~\eqref{eq:fProca} leads to~\eqref{eq:prenonhyp} and
the non-hyperbolic Equation~\eqref{eq:nonhyp} as stated in the introduction. Consequently the method used for $A=0$ is not available to us except in the special case $\kappa=1$. 
Instead we will use an alternative method based on introducing auxiliary fields and constraints, which will be developed in the next section, using the neutral Proca field to motivate an abstract existence result that can be applied in various situations. 

\section{Auxiliary field method} \label{sec:aux field}

\subsection{Green-hyperbolic operators}

A globally hyperbolic spacetime consists of a smooth paracompact $n$-dimensional manifold $M$ with at most finitely many components, equipped with a smooth Lorentzian metric $g$ on $M$ of signature ${+}{-}\cdots{-}$ and a time-orientation, so that there are no closed $g$-causal curves in $M$ and $K\mapsto J^+(K)\cap J^-(K)$ maps compact sets to compact sets.

Let $B_1$ and $B_2$ be finite-dimensional complex vector bundles over a globally hyperbolic spacetime, and let $P:\Gamma^\infty(B_1)\to\Gamma^\infty(B_2)$ be a differential operator. 
If they exist, linear maps $E_P^\pm:\Gamma_0^\infty(B_2)\to \Gamma^\infty(B_1)$ such that 
\begin{enumerate}[G1]
	\item $PE_P^\pm f=f$ for all $f\in \Gamma_0^\infty(B_2)$
	\item $E_P^\pm P f=f$ for all $f\in \Gamma_0^\infty(B_1)$
	\item $\supp E_P^\pm f\subset J^\pm (\supp f)$ for all $f\in \Gamma_0^\infty(B_2)$
\end{enumerate} 
are called retarded ($+$) and advanced ($-$) Green operators for $P$, and we will say that $P$ is \emph{semi-Green-hyperbolic}. If the formal dual of $P$ (see later) is also semi-Green-hyperbolic, then $P$ is said to be Green-hyperbolic~\cite{Baer:2015}; the same is true if the bundles $B_1$ and $B_2$ are equipped with hermitian or bilinear pairings and the formal hermitian conjugate or formal transpose of $P$ are also semi-Green-hyperbolic (as is the case for the Proca operator $P_A$). For the moment, however, we only require semi-Green-hyperbolicity.  

An important result (see Theorem~3.8 in~\cite{Baer:2015}) is that if $P$ has Green operators $E_P^\pm$, then $E_P^\pm$ admit unique linear extensions $\overline{E}_P^\pm:\Gamma^\infty_{\pc/\fc}(B_2)\to \Gamma^\infty_{\pc/\fc}(B_1)$ that are inverses of the restriction of $P$ to a map  $\Gamma^\infty_{\pc/\fc}(B_1)\to \Gamma^\infty_{\pc/\fc}(B_2)$. In particular, $P$ is injective on 
$\Gamma^\infty_{\pc/\fc}(B_1)$ (cf.\ also Corollary~3.9 in~\cite{Baer:2015}). Furthermore, the maps
$\overline{E}_P^\pm$ and consequently $E_P^\pm$ are automatically continuous in appropriate topologies, and indeed $E_P^\pm$ are uniquely determined by $P$ (Corollaries~3.11 and~3.12 of~\cite{Baer:2015}).
We note a useful consequence.
\begin{lemma}\label{lem:intertwine}
	Let $P:\Gamma^\infty(B_1)\to\Gamma^\infty(B_2)$, $Q$, $C$ and $D$ be differential operators so that $CP=QD$, and suppose that $P$ and $Q$ 
	are semi-Green-hyperbolic. Then $DE_P^\pm =   E_Q^\pm C$ on $\Gamma_0^\infty(B_2)$.
\end{lemma}
\begin{proof}
	We have $\overline{E}_Q^\pm CP = D$ on $\Gamma^\infty_{\pc/\fc}(B)$ and hence $DE_P^\pm = \overline{E}_Q^\pm C = E_Q^\pm C$ on $\Gamma_0^\infty(B_2)$.  
\end{proof}

We will also make use of the following key result, which is stated in~\cite{Baer:2015} for Green-hyperbolic operators, but in fact applies to semi-Green-hyperbolic operators.
\begin{thm}\label{thm:exactseq}
	If $P:\Gamma^\infty(B_1)\to\Gamma^\infty(B_2)$ is semi-Green-hyperbolic, with Green operators $E_P^\pm$, then 
	\begin{equation}
		\begin{tikzcd}
			0 \arrow[r] & \Gamma_0^\infty(B_1) \arrow[r, "P"]  &\Gamma_0^\infty(B_2) \arrow[r, "E_P"]  & \Gamma^\infty_{\sc}(B_1) \arrow[r,"P"]  & \Gamma^\infty_{\sc}(B_2) \arrow[r]  & 0 
		\end{tikzcd}
	\end{equation}
	is an exact sequence, where $E_P=E_P^--E_P^+$.
\end{thm} 
\begin{proof}
	The sequence is a complex simply because $E_P$ is a difference of Green operators for $P$. It was shown to be exact at $\Gamma_0^\infty(B_1)$ and $\Gamma_0^\infty(B_2)$ in Theorem~3.5 of~\cite{BaerGinoux:2012} using properties of $P$ and $E_P^\pm$ only (an inessential difference is that~\cite{BaerGinoux:2012} take $B_2=B_1$). Exactness at $\Gamma^\infty_{\sc}(B_1)$ was proved in~\cite{BaerGinoux:2012} using Green operators of the adjoint (in the Green-hyperbolic case), but this can be avoided as follows. Suppose $F\in\Gamma^\infty_{\sc}(B_1)$ obeys $PF=0$. 
	Decompose $F=F^++F^-$ where $F_\pm\in \Gamma^\infty_{\pc/\fc}(B_1)$. 
	Then $F^\pm = \overline{E}_P^\pm P F^\pm$ and $P F^+=-P F^-$ must have compact support as the support is both future- and past-compact, so 
	$F = (\overline{E}_P^- - \overline{E}_P^+) PF^- = (E_P^- - E_P^+) PF^-\in \ran E_P$. Finally, any $G\in\Gamma_{\sc}^\infty(B_2)$ may be decomposed as
	$G=G^++G^-$ with $G^\pm\in \Gamma^\infty_{\pc/\fc}(B_2)$, whereupon $P(\overline{E}_P^+ G^+ + \overline{E}_P^- G^-)= G$. Thus the second $P$ in the complex is a surjection and exactness is proved.
\end{proof}
 An alternative proof  (even allowing for spacetimes with timelike boundary) based on contracting homotopies is for example given in \cite[Thm. 5.2.3]{Lupo_thesis:2015}.

For later use, we note the following simple result. 
\begin{lemma}\label{lem:LU}
	Let $Q$ be a linear differential operator on $\Gamma^\infty(B_1\oplus B_2)$ with block decomposition
	\begin{equation}
		Q = \begin{pmatrix} P & R \\S & T \end{pmatrix},
	\end{equation}
	with respect to the decomposition $\Gamma^\infty(B_1\oplus B_2)\cong \Gamma^\infty(B_1)\oplus \Gamma^\infty(B_2)$.
	If $P$ is invertible, with $P^{-1}\Gamma_0^\infty(B_1)\subset\Gamma_0^\infty(B_1)$, and
	$W=T-SP^{-1}R$ has Green operators $E_W^\pm:\Gamma_0^\infty(B_2)\to \Gamma^\infty(B_2)$, then
	$Q$ has Green operators 
	\begin{equation}\label{eq:ET}
		E_Q^\pm = \begin{pmatrix}
			P^{-1} & -P^{-1}R E_W^\pm \\ 0 & E_W^\pm 
		\end{pmatrix}
		\begin{pmatrix}
			\id & 0 \\ -SP^{-1} & \id
		\end{pmatrix}.
	\end{equation}
	In particular, we note that 
	\begin{equation}\label{eq:ET2}
		\pi_2 E_Q^\pm = E_W^\pm \begin{pmatrix}
			-SP^{-1} & \id
		\end{pmatrix} = \begin{pmatrix}
			-E_W^\pm SP^{-1} & E_W^\pm 
		\end{pmatrix} ,
	\end{equation}
	 in block matrix form, where $\pi_2:\Gamma^\infty(B_1)\oplus \Gamma^\infty(B_2)\to \Gamma^\infty(B_2)$ is the canonical projection.
\end{lemma}
\begin{proof}
	As $P$ is invertible, we may perform an $LU$-decomposition $Q=LU$ where
	\begin{equation}
		L=\begin{pmatrix}
			\id & 0 \\ SP^{-1} & \id
		\end{pmatrix},\qquad
		U = \begin{pmatrix}
			P & R \\ 0 & W
		\end{pmatrix},
	\end{equation}	
	and $L$ is invertible on $\Gamma^\infty(B_1\oplus B_2)$ with inverse given by the second factor on the right-hand side of~\eqref{eq:ET}. Note that
	$L^{-1}\Gamma_0^\infty(B_1\oplus B_2)\subset\Gamma_0^\infty(B_1\oplus B_2)$. Define $E_U^\pm$ to be the first factor on the right-hand side of~\eqref{eq:ET}. One checks directly that $UE_U^\pm F=F=E_U^\pm UF$ and  $\supp E_U^\pm F\subset J^\pm(\supp F)$ for all $F\in \Gamma_0^\infty(B_1\oplus B_2)$, thus showing that $E_U^\pm:\Gamma_0^\infty(B_1\oplus B_2)\to \Gamma^\infty(B_1\oplus B_2)$ are retarded/advanced Green operators for $U$. It follows that $E_U^\pm L^{-1}Q F=F = Q E_U^\pm L^{-1} F$, $\supp E_U^\pm L^{-1}F\subset J^\pm(\supp F)$, so $E_Q^\pm$ are Green operators for $Q$.
\end{proof}

\begin{rem}
    Note that if $Q$ in Lemma~\ref{lem:LU} depends smoothly on a parameter $\lambda$ in such a way that the lemma applies for each $\lambda$, $E^\pm_{W_\lambda}$ depends smoothly on the parameter in the strong sense that $\lambda\mapsto E^\pm_{W_\lambda} f$ is smooth for each $f\in \Gamma_0^\infty(B_2)$, and $R$, $S$ and $P^{-1}$ are polynomial in $\lambda$, then it is apparent that $E^\pm_{Q_\lambda}$ depend smoothly on $\lambda$ in the same sense. The condition that $R$, $S$, and $P^{-1}$ depend polynomially on $\lambda$ is chosen for simplicity; one could likely extend this result to any case where $S$, $R$, and $P^{-1}$ depend smoothly (in the strong sense) on $\lambda$ by applying the theory of topological vector spaces, see for example \cite{Fewster:2023} for the sort of convergence and Leibniz rules needed for this extension.
\end{rem}

\subsection{Warm-up: the neutral Proca field}\label{sec:warmup}
 
As a warm-up exercise, let $(M,g)$ be an oriented $4$-dimensional globally hyperbolic spacetime, 
and consider the uncharged Proca field $W\in \Gamma^\infty(\Lambda^1M)$ obeying the   inhomogeneous equation
\begin{equation}\label{eq:P1}
	-\delta \dd W + m^2 W = J
\end{equation}
for $J\in \Gamma^\infty(\Lambda^1M)$. We seek
solutions $W^\pm$ that obey retarded/advanced boundary conditions, i.e., $W^\pm\in\Gamma^\infty_{\pc/\fc}(\Lambda^1M)$, for compactly supported source $J$.

The system~\eqref{eq:P1} is not normally hyperbolic,
but we can rewrite it as a partial differential equation supplemented by a side condition
\begin{equation}\label{eq:P2}
	(\Box+ m^2)W +\dd\phi = J, \qquad \phi = \delta W.
\end{equation}
in which $W$ is acted on by the normally hyperbolic $p$-form Klein--Gordon operator $K=\Box+m^2$.
The auxiliary field $\phi$ is a smooth $0$-form, i.e., $\phi\in \Gamma^\infty(\Lambda^0M)$. 
Taking the codifferential of~\eqref{eq:P1} gives 
\begin{equation}\label{eq:P3}
	m^2\phi = \delta J
\end{equation}
on reimposing the side condition.
Equations~\eqref{eq:P3} and~\eqref{eq:P2}, including side conditions, constitute the system
\begin{equation}
\label{eq:D}
	\begin{pmatrix}
		m^2 & 0 \\ \dd & K 
	\end{pmatrix}
	DW = DJ, \qquad D = \begin{pmatrix} \delta \\ \id \end{pmatrix}:\Gamma^\infty(\Lambda^1 M)\to \Gamma^\infty(\Lambda^0 M\oplus \Lambda^1 M),
\end{equation} 
and we note that the side condition can be written as
\begin{align}
\label{eq:const operator C}
    C\begin{pmatrix}
        \phi\\ W
    \end{pmatrix}=0\, ,\qquad 
	C = \begin{pmatrix} \id & -\delta\end{pmatrix} : \Gamma^\infty(\Lambda^0 M\oplus \Lambda^1 M)\to \Gamma^\infty(\Lambda^0 M)\, ,
\end{align}
and is now satisfied automatically as $CD=0$.

An important point is that
\begin{equation}\label{eq:Q}
	Q = \begin{pmatrix}
		m^2 & 0 \\ \dd & K 
	\end{pmatrix}
\end{equation}
on $\Gamma^\infty(\Lambda^0 M\oplus\Lambda^1 M)$ is semi-Green-hyperbolic by Lemma~\ref{lem:LU}, because
the inverse of $m^2$ exists and does not change supports, and $K$ is normally hyperbolic. Indeed, 
\begin{equation}
	E_Q^\pm = \begin{pmatrix}
		m^{-2} & 0 \\ -m^{-2} E_K^\pm \dd & E_K^\pm
	\end{pmatrix}.
\end{equation}
Hence, if they exist, the unique solutions $W^\pm$ to~\eqref{eq:P1} satisfying retarded/advanced boundary conditions obey
\begin{equation}\label{eq:DWpmfromDJ}
	DW^\pm = E_Q^\pm DJ
\end{equation}
and can be recovered as
\begin{equation}\label{eq:WpmfromDWpm}
	W^\pm = \pi E_Q^\pm DJ,
\end{equation}
where 
\begin{equation}
	\pi = \begin{pmatrix} 0 & \id \end{pmatrix}: \Gamma^\infty(\Lambda^0 M\oplus \Lambda^1 M)\to \Gamma^\infty(\Lambda^1 M)
\end{equation}
is a left-inverse to $D$, i.e., $\pi D = \id$. A simple calculation shows that the Proca operator
$P=-\delta\dd+m^2$ can be written in terms of the operators just introduced in the form
\begin{equation}\label{eq:piQD}
	P = \pi Q D.
\end{equation}
Equation~\eqref{eq:WpmfromDWpm} shows that the only possible candidates for advanced/retarded Green operators of $P$ are 
\begin{equation}
	E_P^\pm = \pi E_Q^\pm D,
\end{equation}
which indeed reproduces the standard expressions as a result of the calculation
\begin{equation}
	E_P^\pm = -m^{-2}E_K^\pm \dd \delta +E_K^\pm =E_K^\pm (\id-m^{-2}\dd\delta).
\end{equation}

Although this calculation has resulted in the correct answer, it is worth reflecting on why it works. While Eq.~\eqref{eq:DWpmfromDJ} implies~\eqref{eq:WpmfromDWpm}, the converse implication is not immediate because $\pi$ is not a right-inverse of $D$. Instead, one has 
\begin{equation}
    D\pi  = \id - \iota_\aux C, 
\end{equation}
where $\iota_\aux:\Gamma^\infty(\Lambda^0M)\to \Gamma^\infty(\Lambda^0M\oplus\Lambda^1M)$ is the canonical embedding. The consistency of our procedure therefore depends on showing that 
\begin{equation}\label{eq:constraint}
	CE_Q^\pm DJ =0,
\end{equation} 
for all $J\in\Gamma_0^\infty(\Lambda^1M)$.  
Due to the support of $E^\pm_QDJ$, the constraint Eq.~\eqref{eq:constraint} holds in the far past/future. Now
we may note that
\begin{equation}\label{eq:KCCQ}
	KC = \begin{pmatrix}
		K & -K \delta
	\end{pmatrix} = \begin{pmatrix} \id &-\delta\end{pmatrix} \begin{pmatrix}
		m^2 & 0 \\ \dd & K 
	\end{pmatrix}
	= CQ,	
\end{equation}
where we use the same symbol $K$ for the Klein--Gordon operator on $0$- and $1$-forms.
Thus
\begin{equation}
	KC E_Q^\pm DJ = CQ E_Q^\pm DJ = CDJ = 0.
\end{equation}
 As $K$ is normally hyperbolic and~\eqref{eq:constraint} holds in the far past/future,
it follows that~\eqref{eq:constraint} holds globally, as required.
Here, $K$ takes the role of propagating the constraints from the (trivial) initial condition to the whole solution $E^\pm_Q DJ$.

The purpose of taking this rather pedestrian route is that it illustrates an approach that can be used in other situations. Collecting together the main ingredients, we  start with an operator $P$ on $\Gamma^\infty(B)$.  We introduce an auxiliary bundle $B_\aux$ and  an extension map $D:\Gamma^\infty(B)\to \Gamma^\infty(B\oplus B_\aux)$. We then seek an extended operator $Q$ on $\Gamma^\infty(B\oplus B_\aux)$ that is semi-Green hyperbolic and satisfies $P=\pi Q D$, where $\pi$ is a left-inverse of $D$. An important role is played by the constraint operator $C$. It satisfies $CD=0$ and in fact $\ran D=\ker C$. Thus, solutions $f\in \Gamma^\infty(B\oplus B_\aux)$ of $Qf=0$ which respect the constraint $Cf=0$ are of the form $f=DW$ , where $W\in\Gamma^\infty(B)$ solves $PW=0$. Finally, we have an operator $N$ ($N=K$ in the above example) which propagates the constraints thanks to the identity $NC=CQ$.

\subsection{Abstract existence result}

The essential features of the foregoing discussion can be summarized (in a slightly generalized form)
as follows.
\begin{thm}\label{thm:abs}
	Let $B$ and $B_\aux$ be finite-dimensional complex vector bundles. Let
	$P$, $D$, $\pi$, $C$, $Q$ and $N$ be linear differential operators 
	with domains and codomains shown in the following diagram:
	\begin{equation}\label{eq:cd1}
		\begin{tikzcd}
			  \Gamma^\infty(B) \arrow[r, "D"]\arrow[d, "P"] &\Gamma^\infty(B\oplus B_\aux) \arrow[r, "C"]\arrow[d,"Q"] & \Gamma^\infty(B_\aux) \arrow[d,"N"]     \\
				\Gamma^\infty(B)  & \Gamma^\infty(B\oplus B_\aux)
			\arrow[l, "\pi"]\arrow[r,"C"]  & \Gamma^\infty(B_\aux) . 
		\end{tikzcd}
	\end{equation}
	Suppose that (a) $P=\pi QD$ and $CQ=NC$, i.e., the diagram~\eqref{eq:cd1} commutes;  
	(b) $\pi D=\id$, and $D\pi+\iota_\aux C =\id$, where 
	$\iota_\aux:\Gamma^\infty(B_\aux)\to \Gamma^\infty(B\oplus B_\aux)=\Gamma^\infty(B)\oplus \Gamma^\infty(B_\aux)$ is  
	the canonical embedding; 
%	(d)  the top line in~\eqref{eq:cd1} is a short exact sequence; 
	and (c) $Q$ is semi-Green-hyperbolic with Green operators $E_Q^\pm$, 
	while $\Gamma_{\pc/\fc}^\infty(B_\aux)\cap \ker N=\{0\}$ (which holds in particular if $N$ is semi-Green-hyperbolic). Then $P$ is semi-Green-hyperbolic, with Green operators 
	\begin{equation}
    \label{eq:Green op P abstr}
		E_P^\pm=\pi E_Q^\pm D.
	\end{equation}
    Moreover, $D$ is injective and $\ker C=\ran D$, i.e., the top line of~\eqref{eq:cd1} extends to an exact sequence
    $0\to\Gamma^\infty(B)\to \Gamma^\infty(B\oplus B_\aux)\to  \Gamma^\infty(B_\aux)$.
    In cases where $\pi\iota_\aux=0$, it also holds that $C$ is surjective and so the top line 
    ~\eqref{eq:cd1} extends to a short exact sequence. 
    If, additionally, $N$ is semi-Green-hyperbolic, then the diagram displayed in Figure~\ref{fig:chainPQ} commutes, and its rows and columns are all exact sequences, with the columns extending to short exact sequences if $\pi\iota_\aux=0$.
 
	\end{thm}  
	\begin{figure}
		\begin{center}
			\[
				\begin{tikzcd}
				& 0\arrow[d] &  0\arrow[d] & 0\arrow[d] & 0\arrow[d] & \\
				0 \arrow[r] & \Gamma_0^\infty(B) \arrow[r, "P"]\arrow[d,"D"] &\Gamma_0^\infty(B) \arrow[r, "E_P"]\arrow[d,"D"] & \Gamma^\infty_{\sc}(B) \arrow[r,"P"]\arrow[d,"D"] & \Gamma^\infty_{\sc}(B) \arrow[r]\arrow[d,"D"] & 0 \\
				0 \arrow[r] & \Gamma_0^\infty(B\oplus B_\aux)  \arrow[r, "Q"]\arrow[d,"C"] 
				& \Gamma_0^\infty(B\oplus B_\aux) \arrow[r, "E_Q"]\arrow[d,"C"] 
				& \Gamma_{\sc}^\infty(B\oplus B_\aux)  \arrow[r,"Q"]\arrow[d,"C"] 
				& \Gamma_{\sc}^\infty(B\oplus B_\aux) \arrow[r]\arrow[d,"C"] 
				& 0\\
                0 \arrow[r] & \Gamma_0^\infty(B_\aux) \arrow[r, "N"]  &\Gamma_0^\infty(B_\aux) \arrow[r, "E_N"]  & \Gamma^\infty_{\sc}(B_\aux) \arrow[r,"N"] & \Gamma^\infty_{\sc}(B_\aux) \arrow[r]  & 0
				%0 \arrow[r] & \Gamma_0^\infty(B_\aux) \arrow[r, "N"]\arrow[d] &\Gamma_0^\infty(B_\aux) \arrow[r, "E_N"]\arrow[d] & \Gamma^\infty_{\sc}(B_\aux) \arrow[r,"N"]\arrow[d]& \Gamma^\infty_{\sc}(B_\aux) \arrow[r]\arrow[d] & 0
                %\\
				%& 0  &  0  & 0  & 0  &		
			\end{tikzcd}
			\]
		\end{center}
		\caption{Commutative diagram for Theorem~\ref{thm:abs}.}\label{fig:chainPQ}
	\end{figure}
In the neutral Proca field example, $B=\Lambda^1M$, $B_\aux=\Lambda^0M$, while $N=K$ on $0$-forms. 
Equations~\eqref{eq:piQD} and~\eqref{eq:KCCQ} show that~(a) holds; $D\pi +\iota_\aux C=\id$ and
$\pi D=\id$ are simple calculations; and we have already noted that both $Q$ and $N$ are semi-Green-hyperbolic.
Thus Theorem~\ref{thm:abs} applies  with $\pi\iota_\aux=0$.  
\begin{proof}
We prove the statements about $C$ and $D$ first.
Injectivity of $D$ follows from the existence of a left inverse assumed in~(b). 
Next, the identity $D\pi D + \iota_\aux CD = D$ together with $\pi D=\id$ and injectivity of $\iota_\aux$ implies that $CD=0$, so $\ran D\subset \ker C$. 
Conversely, if $\psi\in\ker C$, then $D\pi\psi  = (D\pi+\iota_\aux C)\psi = \psi$, so $\psi\in\ran D$. 
Thus $\ker C=\ran D$, completing the proof that the top line of~\eqref{eq:cd1} extends to the exact sequence
\begin{equation}\label{eq:cd2}
		\begin{tikzcd}
			0\arrow[r] & \Gamma^\infty(B) \arrow[r, "D"]  &\Gamma^\infty(B\oplus B_\aux) \arrow[r, "C"]  & \Gamma^\infty(B_\aux) .% \arrow[r] &0.  
		\end{tikzcd}
\end{equation}
In cases where $\pi\iota_\aux=0$, surjectivity of $C$ holds because $\iota_\aux C\iota_\aux = \iota_\aux - D\pi\iota_\aux =\iota_\aux$ and injectivity of $\iota_\aux$, so~\eqref{eq:cd2} extends to a short exact sequence.
    
Next, let $\Sc$ stand for either compact/spacelike compact support (in fact, $\Sc$ could be any support system~\cite{Baer:2015}). As condition~(b) holds with $\Gamma^\infty(B)$ and $\Gamma^\infty(B\oplus B_\aux)$ replaced by $\Gamma^\infty_\Sc(B)$ and $\Gamma^\infty_\Sc(B\oplus B_\aux)$, it follows that we also have exact sequences when $\Gamma^\infty$ is replaced by $\Gamma^\infty_\Sc$ in~\eqref{eq:cd2}. Thus all the columns in Figure~\ref{fig:chainPQ} are exact and the same argument as above shows that they extend to short exact sequences if $\pi\iota_\aux=0$.
	 
The operators $E_P^\pm=\pi E_Q^\pm D$ are well-defined due to~(c) and the fact that $D\Gamma^\infty_0(B)\subset \Gamma_0^\infty(B\oplus B_\aux)$, and inherit the support property G3 from $E_Q^\pm$. 
Next, we check that $E_P^\pm$ obey G1 and G2. 
For G2, we first note that
\begin{equation}\label{eq:DPeqQD}
		DP = D\pi QD = QD - \iota_\aux CQD = QD - \iota_\aux NCD=QD
\end{equation}
on $\Gamma^\infty(B)$, using assumptions~(a), (b), (a) and the identity $CD=0$. Then we compute
\begin{equation}
		E_P^\pm P = \pi E_Q^\pm DP= \pi E_Q^\pm QD=\pi D =\id
\end{equation}
on $\Gamma_0^\infty(B)$, using~\eqref{eq:DPeqQD} together with assumptions~(c) and then~(b).
For G1, let $J\in\Gamma_0^\infty(B)$ and set $\Psi=E_Q^\pm DJ\in\Gamma^\infty_{\pc/\fc}(B\oplus B_\aux)$. 
Then we have $Q\Psi=DJ$, which implies the two equations $D\pi Q\Psi + \iota_\aux NC\Psi = DJ$, by~(b) and~(a), and 	$D\pi Q\Psi = D\pi D J=DJ$ by~(b). Subtracting, we find $\iota_\aux NC\Psi=0$ and therefore $C\Psi=0$ using injectivity of $\iota_\aux$ and injectivity of $N$ on $\Gamma^\infty_{\pc/\fc}(B_\aux)$. As $J$ is arbitrary, one has $CE_Q^\pm D=0$ on $\Gamma_0^\infty(B)$, allowing us to calculate
\begin{equation}
		PE_P^\pm = \pi QD\pi E_Q^\pm D=  
		\pi QE_Q^\pm D - \pi Q\iota_\aux C E_Q^\pm D = \pi D = \id  
\end{equation}
on $\Gamma_0^\infty(B)$, also using~(a),~(c) and~(b), thus completing the proof that $E^\pm_P$ are Green operators for $P$.

Finally, all the rows in Figure~\ref{fig:chainPQ} are	exact by semi-Green-hyperbolicity of $P$, $Q$ and $N$, and Theorem~\ref{thm:exactseq}; the exactness of the columns  (or their extensions if $\pi\iota_\aux=0$) has already been shown. 
The left-hand and right-hand squares commute by restricting the identities from~(a) and~\eqref{eq:DPeqQD}. 
Lemma~\ref{lem:intertwine} then shows that $E_Q^\pm D=DE_P^\pm$ and $E_N^\pm C=CE_Q^\pm$ and hence that the central squares commute. 
Therefore the diagram in Figure~\ref{fig:chainPQ} commutes. 
\end{proof}

\begin{rem}
If $P$ depends on a parameter $\lambda$ in such a way that each $P_\lambda$ may be shown to be semi-Green-hyperbolic by Theorem~\ref{thm:abs} with $\pi_\lambda$, $D_\lambda$ depending polynomially on $\lambda$, and $E^\pm_{Q_\lambda}$ depending smoothly on $\lambda$ in the strong sense, then $E^\pm_{P_\lambda}$ will also be smooth in the strong sense. 
\end{rem}

The structure of the Green operators in \eqref{eq:Green op P abstr} allows one to draw conclusions about the $\Vc^\pm$-decomposability of the operator $P$, which is an important property for the definition of Hadamard states in the quantized theory discussed later on.

\begin{defn}{\cite[Def. 5.2]{Fewster:2025a}}\label{def:decomposable_sGHO}
    We say that a semi-Green-hyperbolic operator $P:\Gamma^\infty(B)\to \Gamma^\infty(B)$ is  $\Vc^\pm$-decomposable if the distributional kernel of the advanced-minus-retarded operator $E_P=E^-_P-E^+_P$ satisfies
    \begin{align}
        \WF(E_P)\subset (\Vc^+\times \Vc^-)\cup (\Vc^-\times\Vc^+)
    \end{align}
    for conic, relatively closed sets $\Vc^\pm$  in $\dot{T}^*M$ so that $\Vc^-=-\Vc^+$ and $\Vc^+\cap\Vc^-=\emptyset$.
\end{defn}

An important example for a pair $\Vc^\pm$ are the future and past light cones
\begin{equation}
    \Nc^\pm=\{(x,k)\in \dot{T}^*M: g^{-1}_x(k,k)=0\, , \, \, \pm k(u)>0\}\, ,
\end{equation}
where $u\in \Gamma^\infty(TM)$ is any future-directed vector field.

It follows as an immediate corollary of Theorem~\ref{thm:abs} that the operator $P$ is $\Vc^\pm$-decomposable if $Q$ is. 
The following result is useful if the semi-Green-hyperbolicity of $Q$ was obtained by using Lemma~~\ref{lem:LU}.
\begin{prop}
\label{prop:decomp}
    Let $P$, $\pi$, $D$, $C$, $Q$, $N$ satisfy the conditions of  Theorem~\ref{thm:abs} (but not yet assuming that $Q$ is semi-Green-hyperbolic). Assume that one can split $B_\aux=B_1\oplus B'_2$, so that $B_\aux\oplus B=B_1\oplus B'_2\oplus B=:B_1\oplus B_2$. Let $\pi: \Gamma^\infty(B\oplus B_\aux)\to\Gamma^\infty(B)$ be such that $\pi \iota_1=0$, where $\iota_1:\Gamma^\infty(B_1)\to \Gamma^\infty(B\oplus B_\aux)$ is the natural injection.
    Assume that $Q$ is of the form
     \begin{equation}
        Q=\begin{pmatrix}
            J & R\\ S & T
        \end{pmatrix}
    \end{equation}
   with respect to the decomposition $\Gamma^\infty(B\oplus B_\aux)\cong \Gamma^\infty(B_1)\oplus \Gamma^\infty(B_2)$.
    Let $J$ be invertible with inverse $J^{-1}:\Gamma_0^\infty(B_1)\to \Gamma_0^\infty(B_1)$ so that $SJ^{-1}$ is a differential operator. Then if
   \begin{equation}
       W=T-SJ^{-1}R
   \end{equation}
   is a $\Vc^\pm$-decomposable semi-Green-hyperbolic operator, $P$ is semi-Green-hyperbolic and $\Vc^\pm$-decomposable.
\end{prop}

\begin{proof}
    The semi-Green-hyperbolicity of $P$ follows immediately from Lemma~\ref{lem:LU} and Theorem~\ref{thm:abs}. To show the decomposability, we note that under the identification $\Gamma^\infty(B\oplus B_\aux)\cong \Gamma^\infty(B_1)\oplus \Gamma^\infty(B_2)$  we can decompose
    \begin{equation}
        \pi=\begin{pmatrix}
           0 & \pi_2
        \end{pmatrix}\, ,
    \end{equation}
    for some differential operator $\pi_2:\Gamma^\infty(B_2)\to \Gamma^\infty(B)$, and
    \begin{equation}
        D=\begin{pmatrix}
            D_1 \\ D_2
        \end{pmatrix}\, ,
    \end{equation}
    where $D_1:\Gamma^\infty(B)\to \Gamma^\infty(B_1)$ and $D_2:\Gamma^\infty(B)\to \Gamma^\infty(B_2)$ are differential operators. Combining Theorem~\ref{thm:abs} with Lemma~\ref{lem:LU}, one can then compute, using the splitting into $B_1$ and $B_2$,
    \begin{align}
        E^\pm_P=&\pi E^\pm_Q D= \begin{pmatrix}
            0 & \pi_2
        \end{pmatrix}\begin{pmatrix}
            J^{-1} & -J^{-1} R E^\pm_W \\ 0 & E^\pm_W
        \end{pmatrix}\begin{pmatrix}
            \id & 0\\ -SJ^{-1} & \id
        \end{pmatrix}\begin{pmatrix}
            D_1 \\ D_2
        \end{pmatrix}\nonumber\\
        =&\pi_2 E^\pm_W \left(D_2
            -SJ^{-1}D_1\right)\, ,
    \end{align}
   By assumption, $SJ^{-1}$ is a differential operator. Thus, denoting the distributional kernels of $E_P$ and $E_W$ by $E_P$ and $E_W$ as well, one has
    \begin{align}
        \WF(E_P)=\WF((\pi_2\otimes {}^\star(D_2-SJ^{-1}D_1)) E_W) \subset \WF(E_W) \subset (\Vc^+\times \Vc^-)\cup (\Vc^-\times \Vc^+)\, ,
    \end{align}
    where ${}^\star(D_2-SJ^{-1}D_1):\Gamma^\infty(B_2^\star)\to \Gamma^\infty(B^\star)$ is the formal dual of $(D_2-SJ^{-1}D_1)$.
\end{proof}

To conclude this section, we compare our approach with  other results in the literature. First, Khavkine has made a detailed study of hyperbolizable operators arising from variational principles~\cite{Khavkine:2014} whose aims partly overlap with ours. In the terms of~\cite{Khavkine:2014}, our hypotheses assert that the operator $P$ is equivalent to a \emph{constrained hyperbolic system} formed by $Q$ and $C$ (in the sense that their solution spaces may be put in one-to-one correspondence) and furthermore, that the constraints are \emph{hyperbolically integrable}. In the case where $N$ is semi-Green-hyperbolic, the conclusions of Theorem~\ref{thm:abs} imply that the constraints are also \emph{parametrizable}, i.e., the diagram
\begin{equation} \label{eq:cd3}
	\begin{tikzcd}
		\Gamma^\infty(B) \arrow[r, "D"]\arrow[d, "P"] &\Gamma^\infty(B\oplus B_\aux) \arrow[r, "C"]\arrow[d,"Q"] & \Gamma^\infty(B_\aux) \arrow[d,"N"]     \\
		\Gamma^\infty(B)\arrow[r, "D"]  & \Gamma^\infty(B\oplus B_\aux)
		\arrow[r,"C"]  & \Gamma^\infty(B_\aux) . 
	\end{tikzcd}
\end{equation}	
commutes, and $P$, $Q$ and $N$ are semi-Green-hyperbolic.\footnote{This remains true modulo slight generalization of definitions if $N$ obeys the weaker hypotheses in Theorem~\ref{thm:abs}.}
However, our method differs from that of~\cite{Khavkine:2014} in an important respect, namely that we typically expect $P$ and $Q$ to act on different bundles (i.e., $B_\aux$ is nontrivial)
whereas the definition of hyperbolizability and the notion of equivalence used from Section~3 onwards in~\cite{Khavkine:2014} require that the analogous operators would act on the same bundle. This is a key aspect of our approach and one which is particularly important in analysing systems such as the charged Proca and Proca-scalar systems below. In particular, one cannot directly use~\cite{Khavkine:2014} to deduce semi-Green-hyperbolicity of these systems. Note also that~\cite{Khavkine:2014} allows for gauge symmetries, which we do not address here. 
  
 On the subject of gauge theories,  however, we mention a general framework for the quantization of linear gauge theories introduced in~\cite{HackSchenkel:2013} (in close analogy to specific models such as~\cite{FewsterHunt:2013}). The framework is based on connecting the differential operator $P$ governing the dynamics of the gauge theory to a Green-hyperbolic operator $\tilde{P}$ via gauge fixing. This connection is used to establish properties which can replace Green-hyperbolicity in the quantization procedure. In this regard, their framework has similarities with ours. However, in contrast to our method, the operator $\tilde{P}$ acts on the same bundle as $P$, similar to \cite{Khavkine:2014}. Moreover, the treatment of gauge theories faces different difficulties than the ones we are addressing in this work. 

The work~\cite{BeniniMusanteSchenkel:2023} develops the concept of Green-hyperbolic complexes, a generalization of Green-hyperbolic operators which also encompasses gauge theories and is particularly adapted to describing their degeneracies. Simplifying somewhat, a chain complex of operators can be recognised as a Green-hyperbolic complex if it is equipped with a Green's witness, a concept introduced in~\cite{BeniniMusanteSchenkel:2023}. Semi-Green-hyperbolic operators determine particularly simple two-term Green-hyperbolic complexes with correspondingly simple Green's witnesses. Thus, each of $Q$ and $N$ has its own Green-hyperbolic complex, as does $P$ once it is known to be semi-Green-hyperbolic. In fact, diagram~\eqref{eq:cd3} can simply be extended by zero maps at the top and bottom of each column to produce the corresponding complexes, whereupon the horizontal arrows then provide a chain map between them. Beyond that rather simple observation, however, the formalism of~\cite{BeniniMusanteSchenkel:2023} does not seem particularly useful in our present context -- perhaps unsurprisingly, given that the operators we study do not have degeneracies.

\section{Applications}\label{sec:appli}

 In this section, we apply the abstract result Theorem~\ref{thm:abs} to a number of specific examples.

\subsection{Proca multiplets with variable mass matrix}

Let $B=(\Lambda^1M)\otimes \CC^k$ with $k\ge 1$ and consider a multiplet of $k$ Proca fields 
$W\in\Gamma^\infty(B)$ obeying 
\begin{equation}\label{eq:multiProca}
	(-\delta \dd \otimes 1_k + \rho )W=J,
\end{equation}
where  $J\in \Gamma_0^\infty(B)$, $\rho\in C^\infty(M)\otimes M_k(\CC)$ is a smooth field of strictly positive matrices -- the square of the mass matrix --
and $1_k$ denotes the $k\times k$ identity matrix. 
Let $B_\aux= (\Lambda^0M)\otimes\CC^k$, which embeds in $B\oplus B_\aux$ by $\iota_\aux\otimes 1_k$ where $\iota_\aux$ is as in Sect.~\ref{sec:warmup}. Introducing the auxiliary field $\phi=(\delta\otimes 1_k)W\in \Gamma^\infty(B_\aux)$, Eq.~\eqref{eq:multiProca} becomes
\begin{equation}\label{eq:multiProca2}
	(\Box\otimes 1_k+\rho)W - (\dd\otimes 1_k)\phi = 0,\qquad
	\phi=(\delta\otimes 1_k)W
\end{equation}
and applying $\delta\otimes 1_k$ and Leibniz' rule to~\eqref{eq:multiProca} we find that 
\begin{equation}\label{eq:multiProca3}
\rho \phi - (\dd\rho)^\bullet W = (\delta\otimes 1_k )  J,
\end{equation}  
where $((\dd\rho)^\bullet W)^i= (\dd\rho^i_{\phantom{i}j})^\mu W^j_\mu$, using an index notation for the component fields of $W$ and $\phi$. 
Then~\eqref{eq:multiProca2} and~\eqref{eq:multiProca3} can be written in the form $Q(D\otimes 1_k)W= (D\otimes 1_k)J$, where
\begin{equation}
\label{eq:Q multiProca}
		Q = \begin{pmatrix}
		\rho & -(\dd\rho)^\bullet \\ \dd\otimes 1_k & \Box\otimes 1_k+\rho
	\end{pmatrix}, 
\end{equation}
and $D$ is as in Sect.~\ref{sec:warmup}. Furthermore, 
 $P=-\delta \dd \otimes 1_k + \rho$ can be written as  
\begin{equation}
P = (\pi\otimes 1_k) Q (D\otimes 1_k),  
\end{equation}
where $\pi$ is as in Sect.~\ref{sec:warmup}. Note that $(\pi\otimes 1_k)\circ (\iota_\aux\otimes 1_k)=0$. A short calculation gives 
\begin{equation}
(C\otimes 1_k) Q =  N (C\otimes 1_k), \qquad N = \Box\otimes1_k + \rho =-\delta\dd\otimes 1_k + \rho,
\end{equation}
where $C$ is as in \eqref{eq:const operator C}. 
Using Lemma~\ref{lem:LU}, $Q$ is seen to be semi-Green-hyperbolic with Green operators obeying
\begin{equation}
\label{eq:R multiProca}
	(\pi\otimes 1_k) E_Q^\pm = E_{R}^\pm \begin{pmatrix}
		-(\dd\otimes 1_k)\rho^{-1} & \id
	\end{pmatrix}, \qquad \textnormal{where}\quad R =  \Box\otimes 1_k+\rho +(\dd\otimes 1_k) \rho^{-1} (\dd\rho)^\bullet
\end{equation}
is normally hyperbolic. We now apply Theorem~\ref{thm:abs} to the operators $P$, $\pi\otimes 1_k$, $D\otimes 1_k$, $C\otimes 1_k$, $Q$ and $N$. Conditions (a) and (c) have already been verified;
condition~(b) is inherited from the $k=1$ case, so we deduce that $P$ is semi-Green-hyperbolic with Green operators 
\begin{equation}
\label{eq:GreenOp multiProca}
	E_P^\pm = (\pi\otimes 1_k)E_Q^\pm(D\otimes 1_k)= E_R^\pm \left( \id - (\dd\otimes 1_k) \rho^{-1} (\delta\otimes 1_k)\right),
\end{equation}
with the neutral Proca field as an obvious special case where $k=1$, $\rho=m^2$. 
This expression can be verified directly, by first proving the identities (now suppressing the tensor product with the unit)
\begin{equation}
P -R  = \dd(\delta-\rho^{-1}(\dd\rho)^\bullet) 
= \dd\rho^{-1}\delta\rho,\qquad	\delta R = N\rho^{-1}\delta \rho,  
\end{equation}
the second of which implies $\rho^{-1}\delta \rho E_R^\pm = E_N^\pm \delta$  by Lemma~\ref{lem:intertwine}, and then computing
\begin{align}
	PE_P^\pm &= RE_R^\pm (\id-\dd\rho^{-1}\delta)  + (P-R)E_R^\pm(\id-\dd\rho^{-1}\delta)
	=\id - \dd\rho^{-1}\delta + \dd \rho^{-1}\delta\rho E_R^\pm(\id-\dd\rho^{-1}\delta)\nonumber\\
	&= 
	\id - \dd\rho^{-1}\delta + \dd E_N^\pm(\rho-\delta \dd)\rho^{-1}\delta  = \id - \dd\rho^{-1}\delta + \dd E_N^\pm N\rho^{-1}\delta
	= \id
\end{align}
as well as $E_P^\pm P = E_R^\pm (1-\dd\rho^{-1}\delta)(-\delta \dd + \rho)= E_R^\pm R = \id$.

Using the preceding results on the Proca operator and the fact that the mass matrix $\rho$ is strictly positive  with respect to the standard inner product on $\CC^k$, and hence in particular hermitian, it is straightforward to see that the operator $(-\delta \dd \otimes 1_k + \rho )$ is formally hermitian with respect to the pairing
 $(W,V)\mapsto \delta_{ij}\IP{W^i,V^j}_1$
and is therefore Green-hyperbolic. Moreover, the fact that $R$ is normally hyperbolic, and hence by \cite[Theorem 5.3]{Fewster:2025a} $\Nc^\pm$-decomposable, entails that $(-\delta \dd \otimes 1_k + \rho )$ is a $\Nc^\pm$-decomposable operator. Overall, we have shown
\begin{thm}
    The Proca multiplet operator with variable mass matrix, $P=(-\delta\dd\otimes 1_k+\rho)$ acting on $\Gamma^\infty(\Lambda^1M\otimes \CC^k)$, is formally hermitian, Green-hyperbolic and $\Nc^\pm$-decomposable on any globally hyperbolic spacetime $M$ for any $k\geq 1$ and any strictly positive $\rho \in C^\infty(M)\otimes M_k(\CC)$. Its Green operators $E^\pm_{P}$ are given by \eqref{eq:GreenOp multiProca}, where the operator $Q$ acting on sections of $(\Lambda^0M\otimes \CC^k)\oplus(\Lambda^1M\otimes\CC^k)$ is given by \eqref{eq:Q multiProca}, $D$ and $\pi$ are the same as for the single Proca field, and $R$ is given by \eqref{eq:R multiProca}.
\end{thm}

\subsection{Charged Proca field: general set-up and special case $\kappa=1$} 

We now return to the charged Proca field~\eqref{eq:fProca}, starting by establishing some notation and useful identities. First, we introduce various linear operators describing different ways in which $F$ can act on $p$-forms:
\begin{align}
	F^{\bullet}:\Gamma^\infty(\Lambda^1M) & \to \Gamma^\infty(\Lambda^1 M) \nonumber \\
	W&\mapsto F^{\bullet}W, \qquad (F^\bullet W)_\alpha= F_\alpha^{\phantom{\alpha}\beta}W_\beta \\
	F^\wedge:\Gamma^\infty(\Lambda^pM) &\to \Gamma^\infty(\Lambda^{p+2} M) \nonumber \\
	W &\mapsto F^\wedge W=F\wedge W,\\
	F^{\bullet\bullet}:\Gamma^\infty(\Lambda^2M) & \to \Gamma^\infty(\Lambda^0 M) \nonumber \\
	H&\mapsto F^{\bullet\bullet}H=\star^{-1}(F\wedge \star H)= \tfrac{1}{2}F^{\alpha\beta}H_{\alpha\beta}.
\end{align}
When $F^\wedge$ acts on $0$-forms, we will just write $F$. Next, the modified $\Box$-operator is defined by $\Box_A = -(\dd_A\delta_A+ \delta_A\dd_A)$, and is normally hyperbolic because
it has the same principal symbol as $\Box$. Note that
\begin{equation}
	\dd_A \Box_A = \Box_A\dd_A -\dd_A^2 \delta_A + \delta_A \dd_A^2, 
\end{equation}
so,  by~\eqref{eq:dAsquared},
\begin{equation}\label{eq:dABoxA}
	\dd_A \Box_A\omega = \Box_A\dd_A\omega + \ii q \left(\delta_A (F\wedge \omega) - F\wedge\delta_A\omega\right).
\end{equation}
Similarly,
\begin{equation}\label{eq:deltaABoxA}
	\delta_A \Box_A = \Box_A\delta_A -\delta_A^2 \dd_A + \dd_A \delta_A^2, 
\end{equation}
and so,  by~\eqref{eq:deltaAsquared},
\begin{equation}
	\delta_A \Box_A\omega = \Box_A\delta_A\omega + \ii q \left(\star^{-1}(F\wedge \star\dd_A \omega) - \dd_A \star^{-1} (F\wedge\star\omega)\right).
\end{equation}
The formulae~\eqref{eq:dABoxA} and~\eqref{eq:deltaABoxA} also apply when $\Box_A$ is replaced by the modified $p$-form Klein--Gordon operator $K_A=\Box_A+m^2$. In particular, we have
\begin{equation}\label{eq:deltaAKA1form}
	\delta_A K_A \omega = K_A \delta_A\omega +  \ii q F^{\bullet\bullet}\dd_A \omega
\end{equation}
for any $1$-form $\omega$ because $F\wedge\star\omega=0$.

Now we attempt to follow the method used for the uncharged Proca field, starting by  rewriting~\eqref{eq:fProca} as
\begin{equation}\label{eq:fProca2}
	(\Box_A +m^2)W + \ii q \kappa F^\bullet W  + \dd_A \phi= J
\end{equation}
with the side condition that $\phi=\delta_A W$. Equation~\eqref{eq:prenonhyp}, obtained by 
applying $\delta_A$ to~\eqref{eq:fProca}, may be written 
\begin{equation}
	m^2\phi - \ii q(\kappa-1) F^{\bullet\bullet}\dd_A W - \ii q \kappa j^\bullet W   = \delta_A J  \label{eq:Proca3c}
\end{equation}
where $j^\bullet W=j^\mu W_\mu$ defines $j^\bullet:\Gamma^\infty(\Lambda^1M)\to\Gamma^\infty(\Lambda^0M)$.
Equations~\eqref{eq:Proca3c} and~\eqref{eq:fProca2} with the side conditions can be written $Q_AD_AW = D_AJ$, where 
\begin{equation}
	Q_A = 	\begin{pmatrix}
		m^2 & - \ii q(\kappa-1) F^{\bullet\bullet}\circ\dd_A  - \ii q \kappa j^\bullet  \\ \dd_A & K_A + \ii q \kappa F^\bullet 
	\end{pmatrix}, \qquad
	D_A = \begin{pmatrix}
		\delta_A \\ \id
	\end{pmatrix}.
\end{equation}
In general, however, we cannot apply Lemma~\ref{lem:LU} to $Q_A$: although $K_A+\ii q\kappa F^{\bullet}$ is normally hyperbolic, the operator 
\begin{equation}\label{eq:KA1}
	K_A + \ii q \kappa F^\bullet + \ii q m^{-2}\dd_A \circ ((\kappa-1) F^{\bullet\bullet}\circ\dd_A  +\kappa j^\bullet)
\end{equation}
is not normally hyperbolic except in the special case $\kappa=1$ or for pure gauge $A$ (in which case one could also obtain Green-hyperbolicity directly from the uncharged case) and so it is not clear whether it possesses Green operators in general.

Before moving to the general situation, it is worth proceeding with the case $\kappa=1$ for a moment.
In this case, $Q_A$ is semi-Green-hyperbolic 
and the analysis can proceed much as for the uncharged field, replacing $\dd$ and $\delta$ in $C$ and $D$ by $\dd_A$ and $\delta_A$, to give $C_A$ and $D_A$, which still yields a short exact sequence, and satisfies $\pi D_A=\id$,  with $\pi$ as before.
The $\kappa=1$ Proca operator $P_A = -\delta_A\dd_A +\ii q F^\bullet + m^2$ is again obtained as
$P_A = \pi Q_A D_A$, and the analogue of $N$ is $N_A=K_A$ (acting on $0$-forms), which is normally hyperbolic. By Theorem~\ref{thm:abs}, $P_A$ is semi-Green-hyperbolic, with Green operators
\begin{equation}
	E_{P_A}^\pm =  E_{\widetilde{K}_A}^\pm \left(\id-m^{-2}\dd_A\delta_A\right),
\end{equation}
where the normally hyperbolic operator $\widetilde{K}_A = K_A + \ii q  F^\bullet + \ii q m^{-2}\dd_A \circ  j^\bullet$ is obtained by setting $\kappa=1$ in~\eqref{eq:KA1}. In other words, the charged Proca equation with $\kappa=1$ can be solved as a straightforward generalization of the usual trick for the neutral field.

Since $\tilde{K}_A$ is normally hyperbolic, and therefore $\Nc^\pm$-decomposable by \cite[Theorem 5.3]{Fewster:2025a}, it follows from Proposition~\ref{prop:decomp} that $P_A$ is $\Nc^\pm$-decomposable in the case $\kappa=1$. 
Moreover, the discussion in Sect.~\ref{sec:charged_Proca} shows that $P_A$ is formally hermitian with respect to the pairing \eqref{eq:hermpair}, and consequently Green-hyperbolic.

\subsection{The charged Proca field for general $\kappa\in\RR$} \label{sec:chargedProca_sGH}

The obstruction to constructing Green operators for $Q_A$ arose from the presence of a second order operator in the off-diagonal component, which resulted in the operator~\eqref{eq:KA1} failing to be normally hyperbolic. This can be circumvented by replacing $\dd_AW$ in~\eqref{eq:Proca3c} by a new $2$-form field $H$, and adding a side condition that $H=\dd_A W$. Equation~\eqref{eq:Proca3c} becomes
\begin{equation}
	m^2\phi - \ii q(\kappa-1) F^{\bullet\bullet}H- \ii q \kappa j^\bullet W   = \delta_A J  \label{eq:Proca4c}
\end{equation}
while we may leave~\eqref{eq:fProca2} unchanged. 
To complete the picture, we require some dynamics for $H$. 
Applying $\dd_A$ to~\eqref{eq:fProca2} and treating $\phi$ and $\delta_A W$ as independent, one finds
\begin{equation} \label{eq:Proca3b}
	(\Box_A+ m^2)\dd_A W + \ii q\left( \delta_A(F\wedge W)-(\delta_AW) F\right) +\ii q \kappa \dd_A(F^\bullet W)+\ii q \phi F
	= \dd_A J .
\end{equation}
Imposing $H=\dd_A W$, and $\phi=\delta_A W$, Eqs.~\eqref{eq:Proca4c},~\eqref{eq:fProca2} and~\eqref{eq:Proca3b} constitute the system
\begin{align}
	m^2\phi - \ii q(\kappa-1) F^{\bullet\bullet}H- \ii q \kappa j^\bullet W   &= \delta_A J \nonumber\\
	(\Box_A +m^2)W + \ii q \kappa F^{\bullet} W  + \dd_A \phi &= J \label{eq:Proca3a}\\
	(\Box_A+ m^2)H+ \ii q  \delta_A(F\wedge W)  + \ii q\kappa \dd_A(F^{\bullet} W) 
	&= \dd_A J   .   \nonumber
\end{align}
The appropriate auxiliary bundle is now $B_\aux=\Lambda^0M\oplus \Lambda^2M$. 
Presenting $B\oplus B_\aux$ as $\Lambda^0M\oplus \Lambda^1M\oplus\Lambda^2M$ for convenience,
the above system may be written
\begin{equation}\label{eq:ProcaQdef}
	Q_A D_A W = D_A J, 
\end{equation}
where $D_A:\Gamma^\infty(B)\to\Gamma^\infty(B\oplus B_\aux)$ and $Q_A:\Gamma^\infty(B\oplus B_\aux)\to \Gamma^\infty(B\oplus B_\aux)$ are given in the same block matrix form as
\begin{equation}\label{eq:DA}
	D_A = \begin{pmatrix}
	\delta_A \\ \id \\ \dd_A
\end{pmatrix},\qquad Q_A = \begin{pmatrix}
m^2 & -\ii q\kappa j^\bullet & \ii q (1-\kappa) F^{\bullet\bullet}  \\
\dd_A & \Box_A+m^2+ \ii q\kappa F^\bullet & 0 \\
0 &  \ii q \delta_A\circ F^\wedge+ \ii q\kappa \dd_A\circ F^\bullet & \Box_A+m^2 
\end{pmatrix} 
=: \begin{pmatrix}  Q_0\\ Q_1\\ Q_2 \end{pmatrix}.
\end{equation}
 Let $C_A:\Gamma^\infty(B\oplus B_\aux)\to \Gamma^\infty(B_\aux)$ and $\pi:\Gamma^\infty(B\oplus B_\aux)\to \Gamma^\infty(B)$ be given by
\begin{equation} \label{eq:CApi} 
	C_A = \begin{pmatrix}
		\id & -\delta_A & 0 \\ %0 & 0 & 0 \\
		0 & -\dd_A & \id
	\end{pmatrix},
	\qquad 
	\pi = \begin{pmatrix}
		0 & \id & 0
	\end{pmatrix},
\end{equation} 
 noting that $\pi\circ\iota_\aux=0$, and let 
\begin{equation}
	N_A = \begin{pmatrix}
		\Box_A+m^2  & \ii q (1-\kappa)F^{\bullet\bullet}  \\
		-\ii q F & \Box_A+m^2 
	\end{pmatrix}.
\end{equation}
Now we can begin to verify that the conditions of Theorem~\ref{thm:abs} are met  for $P_A$, $D_A$, $\pi$, $C_A$, $Q_A$ and $N_A$. Condition~(b) is checked by a simple calculation. Next, we calculate
\begin{equation}
	\pi Q_A D_A = \dd_A \delta_A + \Box_A+m^2+ \ii q\kappa F^\bullet = P_A.
\end{equation}
To  see that $	C_A Q_A =   N_AC_A$, we compute
\begin{equation}
	 N_AC_A = \begin{pmatrix}
		K_A & -K_A  \delta_A-\ii q(1-\kappa) F^{\bullet\bullet} \dd_A & \ii q (1-\kappa)F^{\bullet\bullet} \\
		-\ii q F & \ii q F\delta_A- K_A  \dd_A  & K_A	
	\end{pmatrix},
\end{equation}
where $K_A=\Box_A + m^2$. The first and third column match those of 
\begin{equation}
	C_A Q_A = \begin{pmatrix}
		K_A & -\ii q\kappa j^\bullet -\delta_A\circ (K_A + \ii q \kappa F^\bullet) & \ii q (1-\kappa)F^{\bullet\bullet} \\
		-\ii q F & \ii q \delta_A \circ F^\wedge -\dd_A K_A    & K_A	
	\end{pmatrix},
\end{equation}
recalling that $\dd_A^2=\ii qF$ on scalars. Thus only two $1$-form identities are to be verified, namely
\begin{align}
	-K_A  \delta_A-\ii q(1-\kappa) F^{\bullet\bullet} \dd_A &=   -\ii q\kappa j^\bullet -\delta_A\circ (K_A + \ii q \kappa F^\bullet)
	\nonumber \\
	\ii q F\delta_A- K_A\dd_A &=\ii q \delta_A \circ F^\wedge -\dd_A K_A,
\end{align}
the second of which follows immediately from~\eqref{eq:dABoxA} on $1$-forms,
while the first is a consequence of~\eqref{eq:deltaAKA1form} and the identity~\eqref{eq:usefulidentity}. Thus we have $C_A Q_A=N_A C_A$ and condition~(a) is verified.

Turning to condition~(c), $N_A$ is evidently normally hyperbolic, and Lemma~\ref{lem:LU} shows that $Q_A$ is semi-Green-hyperbolic, because the top right-hand entry in the block form of $Q_A$ is invertible (with inverse $m^{-2}$ preserving compact support) and the operator
\begin{equation}\label{eq:MA}
	M_A:=	\begin{pmatrix} K_A+\ii q\kappa F^\bullet & 0 \\ 
		\ii q \delta_A\circ F^\wedge+ \ii q\kappa \dd_A\circ F^\bullet & K_A
	\end{pmatrix}	+ \frac{\ii q}{m^2}\begin{pmatrix} \kappa \dd_A \circ j^\bullet & (\kappa-1)\dd_A\circ F^{\bullet\bullet} \\ 0 & 0
	\end{pmatrix}
\end{equation}
on $\Gamma^\infty(\Lambda^1M\oplus \Lambda^2 M)$ is normally hyperbolic and possesses Green operators. By \cite[Theorem 5.3]{Fewster:2025a}, this also implies that $M_A$ is $\Nc^\pm$-decomposable. In block form, 
on $\Lambda^0M \oplus (\Lambda^1M\oplus \Lambda^2M)$,
Lemma~\ref{lem:LU} gives 
\begin{equation}\label{eq:EpmQA}
	E_{Q_A}^\pm = \begin{pmatrix} m^{-2} & -m^{-2}R E_{M_A}^\pm \\ 0 & E_{M_A}^\pm \end{pmatrix}
	\begin{pmatrix}
		\id_{\Lambda^0M} & 0 \\ \begin{pmatrix} -m^{-2}\dd_A\\ 0 \end{pmatrix} & \id_{\Lambda^1M\oplus\Lambda^2M}
	\end{pmatrix},
\end{equation} 
where $R = \begin{pmatrix} -\ii q\kappa j^\bullet & \ii q (1-\kappa) F^{\bullet\bullet} \end{pmatrix}$.  Thus condition~(c) also holds, and $P_A$ is therefore semi-Green-hyperbolic by Theorem~\ref{thm:abs},  in the case with $\pi\iota_\aux=0$.

We now come to the main theorem on the charged Proca operator with general magnetic dipole moment.
\begin{thm}
	The charged Proca operator $P_A=-\delta_A \dd_A  + \ii q \kappa F^\bullet + m^2$ is Green-hyperbolic and $\Nc^\pm$-decomposable on any globally hyperbolic spacetime $M$, for any $\kappa\in\RR$, $m>0$, and with any background electromagnetic potential $A\in \Gamma^\infty(T^*M)$. The operator $P_A$ is also formally hermitian with respect to the pairing~\eqref{eq:hermpair}. The Green operators of $P_A$ are 
	\begin{equation}\label{eq:EpmPA}
		E_{P_A}^\pm = 
		\pi E_{Q_A}^\pm D_A = \pi_1 E_{M_A}^\pm \widetilde{D}_A, \qquad \widetilde{D}_A:=\begin{pmatrix}
			\id-m^{-2}\dd_A \delta_A \\ \dd_A
		\end{pmatrix}
	\end{equation} 
	where $E_{Q_A}^\pm$ and $E_{M_A}^\pm$ are Green operators for the operators $Q_A$ on $\Gamma^\infty(\Lambda^0M\oplus\Lambda^1M\oplus\Lambda^2M)$ and 
	$M_A$ on $\Gamma^\infty(\Lambda^1M\oplus \Lambda^2 M)$
	given by~\eqref{eq:ProcaQdef} and~\eqref{eq:MA}, and $\pi_1:\Lambda^1M\oplus \Lambda^2 M\to \Lambda^1M$ is the canonical projection.
\end{thm}
\begin{proof}
	We have shown that $P_A$ is semi-Green-hyperbolic, and it was shown in Sect.~\ref{sec:charged_Proca} that
	$P_A$ is formally hermitian with respect to the pairing~\eqref{eq:hermpair}. Thus $P_A$ is Green-hyperbolic. The $\Nc^\pm$-decomposability follows immediately from the corresponding decomposability of the operator $M_A$ above, the form of $Q_A$, and Proposition~\ref{prop:decomp}.
	For the final formula in~\eqref{eq:EpmPA}, first write $
	\pi_{1\oplus 2}:\Lambda^0M\oplus \Lambda^1M\oplus \Lambda^2M\to \Lambda^1M\oplus \Lambda^2M$ for the projection, whereupon~\eqref{eq:EpmQA} and the formula for $D_A$ in~\eqref{eq:DA} give
	\begin{equation}\label{eq:EpmQAMA}
		\pi_{1\oplus 2} E_{Q_A}^\pm D_A = E_{M_A}^\pm \widetilde{D}_A.
	\end{equation}
	The required formula follows because $\pi=\pi_1\circ \pi_{1\oplus 2}$. 
\end{proof} 

\subsection{An alternative treatment of the neutral Proca field}

Here, we give an alternative treatment of the Green operators for the neutral Proca field, using an auxiliary $2$-form field $H$ rather than a $0$-form, so $B_\aux = \Lambda^2M$. 
The purpose is simply to show that there can be more than one way of using Theorem~\ref{thm:abs} to conclude semi-Green-hyperbolicity of a given operator.

In fact, consider the $p$-form Proca operator $P:=-\delta\dd +m^2$ on $\Gamma^\infty(B)$, $B=\Lambda^p M$ with $p\ge 1$. Let $B_\aux=\Lambda^{p+1} M$ and define
\begin{equation}
	D = \begin{pmatrix} \id \\ \dd\end{pmatrix}, \qquad Q = \begin{pmatrix}
		m^2 & - \delta \\ 0 & K
	\end{pmatrix}, \qquad \pi =\begin{pmatrix} \id & 0 \end{pmatrix}, \qquad C = \begin{pmatrix} -\dd & \id\end{pmatrix},
\end{equation}
where $K$ is the $(p+1)$-form Klein--Gordon operator. It is straightforwardly seen that $P=\pi Q D$, 
and 
\begin{equation}
	CQ = \begin{pmatrix}
	  -m^2\dd & \dd\delta+K 
	\end{pmatrix}= NC
\end{equation}
where $N=-\delta \dd+m^2$ on $(p+1)$-forms. Furthermore, $\pi D=\id$ and $D\pi + \iota_\aux C = \id$, 
and $Q$ is semi-Green-hyperbolic by Lemma~\ref{lem:LU}. To satisfy the conditions of Theorem~\ref{thm:abs} it suffices to show that $N$ is semi-Green-hyperbolic. Thus the semi-Green-hyperbolicity of the $p$-form Proca operator follows from that of the $p+1$-form operator for $1\le p\le 3$. Iterating, the semi-Green-hyperbolicity of the $1$-form Proca operator follows from the semi-Green-hyperbolicity of the $4$-form Proca operator, 
which is simply multiplication by $m^2$ with Green operators $m^{-2}$. In this way we obtain an amusing, albeit slightly eccentric, alternative proof that the Proca operator in $p\ge 1$ is semi-Green-hyperbolic. 

\subsection{Coupled Proca--scalar system}
\label{sec:Proca--scalar}
Finally, we consider the coupled Proca--scalar system. We begin by considering a real Proca field $W\in \Gamma^\infty(\Lambda^1M)$ of mass $m$ coupled linearly to a real scalar field $\phi\in \Gamma^\infty(\Lambda^0M)$ of mass $\fm$.

First, keeping in line with the notation for the charged scalar,  $v\in \Gamma^\infty(\Lambda^1M)$ can act on $p$-forms as:
\begin{align}
    v^{\bullet}:\Gamma^\infty(\Lambda^pM)&\to \Gamma^\infty(\Lambda^{p-1}M)\\\nonumber
    W&\mapsto v^{\bullet}W=v^\alpha W_{\alpha\beta_1\dots\beta_{p-1}}\, ,\\
    v^\wedge:\Gamma^\infty(\Lambda^pM)&\to \Gamma^\infty(\Lambda^{p+1}M)\\\nonumber
    f&\mapsto v^\wedge f= v\wedge f\, ,
\end{align}
and for $v^\wedge$ acting on $0$-forms we will just write $v$.

The dynamics of the system are described by the Lagrangian $\mathcal{L}=\mathcal{L}_W+\mathcal{L}_\phi+\mathcal{L}_I$ with
\begin{align}
    \mathcal{L}_W&=-\frac{1}{4}H^{\mu\nu}H_{\mu\nu}+\frac{1}{2}m^2W_\mu W^\mu \\
    \mathcal{L}_\phi&=\frac{1}{2}(\nabla_\mu \phi\nabla^\mu \phi -\fm^2\phi^2)\\
    \mathcal{L}_I&= v^\mu \phi W_\mu\, .
\end{align}
Here, $H=\dd W$, and $m$ and $\fm$ are the masses of the Proca and scalar field, respectively.  $v\in \Gamma^\infty(\Lambda^1M)$ is a  real $1$-form that determines the polarization of $W$ the scalar field couples to, as well as the coupling strength. 

From this, we obtain the inhomogeneous differential equations for $W$ and $\phi$:
\begin{align}
\label{eq:diff eq compl}
    -\delta\dd W+m^2 W+v\phi&=J\\
    -\delta\dd \phi+\fm^2\phi- v^{\bullet} W&=K_\fm^{(0)}\phi- v^{\bullet} W=\rho\, ,
\end{align}
where $\rho\in \Gamma^\infty(\Lambda^0M)$ and $K_m^{(p)}:=-(\delta\dd+\dd \delta)+m^2$ denotes the $p$-form Klein--Gordon operator with mass $m$. 

In the notation of Theorem~\ref{thm:abs}, we have $B=\Lambda^1M\oplus\Lambda^0M$ and
\begin{align}
\label{eq:P real Proca scalar}
    P=\begin{pmatrix}
        -\delta\dd+m^2 &  v \\
        - v^{\bullet} & K_\fm^{(0)}
    \end{pmatrix}\, .
\end{align}

Equip $\Gamma^\infty(B)$ with the hermitian pairing
\begin{align}
    \IP{(W,\phi),(Z,\psi)}_B=\IP{\phi,\psi}_0+\IP{W,Z}_1= \int\limits_M(\bar{\phi}\psi-\bar{W}_\alpha Z^\alpha)\vol\, 
\end{align}
for sections with compactly overlapping supports, 
where $\IP{\cdot,\cdot}_p$ is the hermitian pairing for $p$-forms introduced earlier in Sect.~\ref{sec:charged_Proca}. Since $v$ is a real $1$-form, and the Klein--Gordon and Proca operators are formally hermitian with respect to $\IP{\cdot,\cdot}_0$ and $\IP{\cdot,\cdot}_1$, it follows immediately that the coupled Proca-scalar operator $P$ is formally hermitian. 

To show semi-Green-hyperbolicity, the naive choice for $Q$ and $B_\aux$, following the scheme for the real Proca field, would be to set $B_\aux=\Lambda^0M$ and 
\begin{align}
    Q=\begin{pmatrix}
       m^2 & 0 & \delta v \\
       \dd & K^{(1)}_m &  v\\
       0 & - v^{\bullet} & K_\fm^{(0)}
    \end{pmatrix}\, .
\end{align}
by rewriting $\delta W=\psi$, and adding the dynamics for $\psi$.
However, $Q$ does not satisfy the conditions of Lemma~\ref{lem:LU}, and therefore it is not immediately clear whether it is semi-Green-hyperbolic.

The problem is caused by the term $\delta v\phi$ in the equation for $\psi$. To remedy this, we observe that
\begin{align}
    \delta( v \phi)=(\delta v)\phi- v^{\bullet} \dd \phi\, .
\end{align}
We then introduce a new $1$-form field $V$, satisfying the side condition $\dd \phi=V$, and obtain dynamics for $V$ by applying $\dd$ to the equation for $\phi$. This enlarges the auxiliary space to $B_\aux=\Lambda^0M\oplus \Lambda^1M$.
 We will represent $B_\aux\oplus B$ as the fourfold direct sum
$\Lambda^0M\oplus \Lambda^1M\oplus \Lambda^1M\oplus \Lambda^0M$,
indicating the partition into $B_\aux\oplus B$ by dotted lines in block matrices. We will also repartition $B_\aux\oplus B$ as  $\Lambda^0M \oplus (\Lambda^1M\oplus \Lambda^1M\oplus \Lambda^0M)$, indicated by dashed lines in block matrices. This allows us to regard operators
on $B_\aux\oplus B$ and related spaces as $2\times 2$ block matrices in two ways, using dotted or dashed lines to indicate what partition is to be understood.
We set $D:\Gamma^\infty(B)\to\Gamma^\infty(B_\aux\oplus B)$, $\pi:\Gamma^\infty(B_\aux\oplus B)\to \Gamma^\infty(B)$, and $Q:\Gamma^\infty(B_\aux\oplus B)\to \Gamma^\infty(B_\aux\oplus B)$ to be
\begin{equation} 
		D= \begin{pnmatrix}{cc}
		\delta & 0\\ 
		\hdashedline
		0 & \dd\\ 
		\hddottedline
		\id & 0\\
		0 & \id\\
	\end{pnmatrix}\,,\quad 
	\pi=\begin{pnmatrix}{cIcocc}
	0 & 0 & \id & 0\\
	0 & 0 & 0 & \id
	\end{pnmatrix} \,,
	\quad Q=\begin{pnmatrix}{cIcocc}
		m^2 & - v^{\bullet} &  0 & (\delta v)\\
		\hdashedline
		0 & K_\fm^{(1)} & -\dd v^\bullet & 0\\
		\hddottedline
		\dd & 0 & K^{(1)}_m & v\\
		0 & 0 & - v^\bullet & K_\fm^{(0)}
	\end{pnmatrix}\, .
\end{equation}
Now, $m^2$ is invertible, and $T: \Gamma^\infty(\Lambda^1M\oplus B)\to \Gamma^\infty(\Lambda^1M\oplus B)$,
\begin{align}
	\label{eq:sP-norm hyp operator}
	T&=\begin{pnmatrix}{cocc}
		K_\fm^{(1)} & -\dd v^\bullet & 0\\ 
		\hddottedline
		0 & K_m^{(1)} & v \\
		0 & - v^\bullet  & K_\fm^{(0)}
	\end{pnmatrix}
	-m^{-2}\begin{pnmatrix}{c} 0\hphantom{,} \\ \hddottedline \dd \\ \,0\hphantom{,}
	\end{pnmatrix} 
	\begin{pnmatrix}{cocc}
		- v^\bullet & 0 & (\delta v)
	\end{pnmatrix}\nonumber\\
	&=\begin{pnmatrix}{cocc}
		K_\fm^{(1)} & -\dd v^\bullet & 0\\ \hddottedline
		m^{-2}\dd v^\bullet & K_m^{(1)} & v-m^{-2}\dd(\delta v) \\
		0 & - v^\bullet & K_\fm^{(0)} 
	\end{pnmatrix}\, ,
\end{align}
 is normally hyperbolic and hence semi-Green-hyperbolic. Moreover, by \cite[Theorem 5.3]{Fewster:2025a}, $T$ is also $\Nc^\pm$-decomposable. One can now use Lemma~\ref{lem:LU} to show that $Q$ is semi-Green-hyperbolic, with Green operators given by
  \begin{align}
 	\label{eq:sP-Q_Green op}
 	E_Q^\pm=\begin{pnmatrix}{cIc}
 		m^{-2} & -m^{-2} R E^\pm_T \\ 
 		\hdashedline
 		0 & E^\pm_T
 	\end{pnmatrix}
 	\begin{pnmatrix}{cIc}
 		\id & 0 \\
 		\hdashedline
 		-m^{-2}S & \id
 	\end{pnmatrix}\, .
 \end{align}
Here, $R:\Gamma^\infty(\Lambda^1M\oplus B)\to \Gamma^\infty(\Lambda^0M)$ and $ S:\Gamma^\infty(\Lambda^0M)\to \Gamma^\infty(\Lambda^1M\oplus B)$ are given by
\begin{align}
	R&= \begin{pnmatrix}{cocc}
		- v^{\bullet} &  0 & (\delta v)
	\end{pnmatrix}\, , \quad S=\begin{pnmatrix}{c}
		0\hphantom{,}\\ \hddottedline \dd \\ \,0\hphantom{,}\,
	\end{pnmatrix}\, .
\end{align}

To apply Theorem~\ref{thm:abs}, we further define $C:\Gamma^\infty(B_\aux\oplus B)\to \Gamma^\infty(B_\aux)$ as
\begin{equation}
	C=\begin{pnmatrix}{cIcocc}
		\id & 0 & -\delta & 0\\
		0 & \id & 0 & -\dd
	\end{pnmatrix} \, .
\end{equation}
It is then a straightforward calculation using the canonical injection $\iota_\aux: \Gamma^\infty(B_\aux)\to \Gamma^\infty(B_\aux\oplus B)$  written as
\begin{equation}
    \iota_\aux=\begin{pnmatrix}{cc}
        \id & 0\\ \hdashedline
        0 & \id\\ \hddottedline
        0 & 0\\
        0 & 0
    \end{pnmatrix}
\end{equation}    
to verify that $\pi D=\id$, $D\pi+\iota_\aux C=\id$, $CD=0$, and $\pi Q D=P$. Setting also 
\begin{align}
    &N=\begin{pmatrix}
       K^{(0)}_m & - v^\bullet\\
       0 & K^{(1)}_\fm
    \end{pmatrix}: \Gamma^\infty(B_\aux)\to \Gamma^\infty(B_\aux)\, 
\end{align}
one can check that $CQ=NC$, so conditions (a) and (b) of Theorem~\ref{thm:abs} are satisfied. Moreover, $N$ is normally hyperbolic and therefore semi-Green-hyperbolic, so that in particular $\Gamma^\infty_{pc/fc}(B_\aux)\cap \ker N=\{0\}$. With this, also requirement (c) of Theorem~\ref{thm:abs} is satisfied. Since $Q$ also satisfies the requirement of Proposition~\ref{prop:decomp}, and $T$ is $\Nc^\pm$-decomposable, we obtain

\begin{thm}
   The operator $P$ in \eqref{eq:P real Proca scalar} describing the dynamics of the coupled Proca-scalar system is Green-hyperbolic and $\Nc^\pm$-decomposable for any real coupling $1$-form $v\in \Gamma^\infty(\Lambda^1M)$, with Green operators
\begin{align}
\label{eq:sP- Green op formula}
    E^\pm_P=\pi E^\pm_QD\, .
\end{align} 
\end{thm}

Using the form and properties of the intermediate operators $R$, $S$, $Q$, and $T$, one can deduce that $E_P^\pm$ takes the form
\begin{align}
\label{eq:scalar-Proca Green op}
     E^\pm_P=\begin{pmatrix}
        E^\pm_{WW} & -E^\pm_{WW}v E^\pm_{K^{(0)}_\fm}\\
        E^\pm_{K^{(0)}_\fm} v^\bullet  E^\pm_{WW} & E^\pm_{K^{(0)}_\fm}-E^\pm_{K^{(0)}_\fm} v^\bullet E^\pm_{WW} v  E^\pm_{K^{(0)}_\fm}
    \end{pmatrix}\, ,
\end{align}
where 
\begin{align}
\label{eq:E^pm_WW}
    E^\pm_{WW}=\tilde{E}^\pm_{WW}(\id-m^{-2}\dd\delta)\, .
\end{align}
 and $\tilde{E}^\pm_{WW}:\Gamma_0^\infty(\Lambda^1M)\to \Gamma^\infty(\Lambda^1M)$ is a component of $E^\pm_T$ satisfying (G3) and 
 \begin{align}
 \label{eq:sP-coupled Green op-WW}
    \tilde{E}^\pm_{WW}O^\pm&=O^\pm \tilde{E}^\pm_{WW}=\id\, ,\\
    O^\pm&=K^{(1)}_m+(\id-m^{-2}\dd\delta) v E^\pm_{K^{(0)}_\fm} v^\bullet\, .
\end{align}
A more detailed account of the derivation can be found in Appendix~\ref{sec:App A}.

Instead, we could consider the charged Proca field coupled to a charged scalar field. In the absence of a background electromagnetic field, this theory is described by the  Lagrangian 
\begin{align}
    \mathcal{L}&=-\frac{1}{2}\overline{H}^{\mu\nu}H_{\mu\nu}+m^2\overline{W}_\mu W^\mu +\nabla_\mu \overline{\phi}\nabla^\mu \phi -\fm^2\phi\overline{\phi}+ v^\mu (\phi\overline{W}_\mu+ \overline{\phi}W_\mu )\, ,
\end{align}
where we still assume $v\in \Gamma^\infty(\Lambda^1M)$ to be real. We then find the inhomogeneous equations
\begin{align}
    -\delta\dd W+m^2 W+v \phi&=J\\
    -\delta\dd \overline{W}+m^2 \overline{W}+v \overline{\phi}&=\overline{J}\\
    K_\fm^{(0)}\phi- v^\bullet W=\rho\\
    K_\fm^{(0)}\overline{\phi}-v^\bullet \overline{W}=\overline{\rho}\, .
\end{align}
Hence, the equations for $(W,\phi)$ and $(\overline{W}, \overline{\phi})$ decouple, and the complex theory can be treated as two copies of the real theory.

\section{Quantized theory and Hadamard states}
\label{sec:quantum}
\subsection{Algebras of observables -- general theory}
\label{sec:gen framework} 

In this section, we will discuss the quantization of theories described by Green-hyperbolic operators, in particular the construction of the algebra of observables.  We largely follow~\cite{Fewster:2025a} but see also e.g.,~\cite{Baer:2015,BaerGinoux:2012}.

In many cases of interest, such as the real scalar field  obeying the Klein-Gordon equation or the neutral Proca field, the operator is not only Green-hyperbolic, but also real and formally hermitian (it is a RFHGHO in the nomenclature of \cite{Fewster:2025a}). In more detail, the bundle $B$ on whose sections the differential operator $P$ acts is equipped with a hermitian fibre metric $\IP{\cdot,\cdot}_x$ and a complex conjugation $\Cc$ satisfying $\IP{\Cc u,\Cc v}_x=\overline{\IP{u,v}_x}$ for any $u,v\in B_x$ and the operator $P$ satisfies $P=\Cc P \Cc$ and $\IP{f, Ph}_B=\IP{Pf, h}_B$ for all $f,h\in \Gamma^\infty(B)$ so that $\supp f\cap \supp h $ is compact. Here, 
\begin{equation}
   \IP{f,h}_B= \int \IP{f(x), h(x)}_x\vol \quad \forall f,h\in \Gamma^\infty(B): \, \supp f\cap \supp h \text{ compact}
\end{equation}
denotes the pairing of sections with compact overlap.  A bundle $B$ equipped with hermitian fibre metric and complex conjugation as above is called a \emph{hermitian vector bundle} in~\cite{Fewster:2025a}.

 We consider the theory with quadratic action functional
\begin{equation}
    S[\Ac] = -\frac{1}{2}\IP{\Cc \Ac,P\Ac},
\end{equation}
up to total divergences and ignoring cutoffs needed for convergence purposes, which has Euler--Lagrange equation $P\Ac=0$.
In this case, the algebra of observables of this theory can be taken to be the free unital $*$-algebra $\Af_R (P,B)$
generated by $\Ac(f)$, $f\in\Gamma_0^\infty(B)$, and subject to the relations
    \begin{itemize}
    \item[R1] $f\mapsto \Ac(f)$ is complex linear,
    \item[R2] $(\Ac(f))^*=\Ac(\Cc f)$ (hermiticity),
    \item[R3] $\Ac(P f)=0$ (field equations),
    \item[R4] $\left[\Ac(f), \Ac(h)\right]=\ii \IP{\Cc f, E_P h}_B$ (commutation relations),
\end{itemize}
for all $f, h\in \Gamma^\infty_0(B)$. Here $E_P=E^-_P-E^+_P$ is the advanced-minus-retarded Green operator.
The interpretation is that the generators $\Ac(f)$ correspond to the classical smeared fields $\Ac(f)=\IP{\Cc f, \Ac}_B$, and the commutation relations are those obtained from Peierls' method applied to the action stated above.  If $N\subset M$ is a \emph{region} of $M$, i.e., an open casually convex subset with at most finitely many connected components, then $N$ defines a globally hyperbolic spacetime $M|_N$ in its own right when equipped with the metric and causal structure inherited from $M$. We may identify $\Af_R(P|_N,B|_N)$ with the subalgebra of $\Af_R(P,B)$ generated by $\Ac(f)$ for $f\in \Gamma_0^\infty(B)$ with $\supp f\subset N$.

It is not difficult to find examples where this scheme does not apply. One example is the charged Proca field in a non-trivial background electromagnetic field discussed in  Sect.~\ref{sec:chargedProca_sGH}. In this example the operators are Green-hyperbolic and formally hermitian but not real; more generally, there may not necessarily be a complex conjugation $\Cc$ on the bundle $B$. Such operators, called FHGHOs in \cite{Fewster:2025a}, can nonetheless be quantized following a well-known scheme.

One possibility\footnote{Another is to generate the algebra using one linear and one antilinear field.} for the quantization in this case is to set the algebra of observables $\Af_F(P,B)$ to be the free unital $*$-algebra generated by $\Ac(\lambda)$, with $\lambda\in \Gamma_0^\infty(B^*)$, and $\Ac^\star(f)$, with $f\in \Gamma_0^\infty(B)$, and subject to the relations
\begin{itemize}
    \item[F1] $\lambda\mapsto \Ac(\lambda)$ and $f\mapsto \Ac^\star(f)$ are complex linear,
    \item[F2] $\Ac^\star(f)=\Ac(f^\star)^*$,
    \item[F3] $\Ac(\sadj{P}\lambda)=0=\Ac^\star(P f)$ (field equations),
    \item[F4] $\left[\Ac(\lambda),   \Ac^\star(f)\right]=\ii\lambda( E_{P} f) \1$, and 
    $\left[\Ac(\lambda), \Ac(\mu)\right]=\left[\Ac^\star(f), \Ac^\star(f)\right]=0$ (commutation relations).
\end{itemize}
Here, $\star: B\to B^*$, $B_x\ni f\mapsto \IP{f,\cdot}_x\ni B^*_x$ is an antilinear bundle isomorphism induced by the fibre metric and $\sadj{P}:\Gamma^\infty(B^*)\to \Gamma^\infty(B^*)$ is the formal dual of $P$ defined by the condition $\lambda(Pf)=(\sadj{P}\lambda)(f)$ for all $\lambda\in \Gamma^\infty(B^*)$, $f\in \Gamma^\infty(B)$ with compactly overlapping support.
The interpretation is that $\Ac(\lambda)$ and $\Ac^\star(f)$ correspond to the classical smeared fields $\lambda(\Ac)$  and  $\Ac^\star(f)$ (recall that $\lambda$ is valued in $B^*$, while $f$ is valued in $B$). Here, one has in mind a theory with action functional
\begin{equation}\label{eq:complex_action}
    S[\Ac] = -\Ac^\star(P\Ac),
\end{equation}
modulo cutoffs, with Euler--Lagrange equations $P\Ac=0$ and $\sadj{P}\Ac^\star=0$. Again the commutation relations can be obtained by Peierls' method.
Note that the nontrivial commutation relation in F4 has an equivalent form
\begin{equation}\label{eq:CCR_equiv} 
\left[ \Ac^\star(h)^*,   \Ac^\star(f)\right]=\ii
\IP{h,E_{ P}f}  \1
\end{equation}
for $f,h\in \Gamma_0^\infty(B)$.

If the bundle $B$ is equipped with a linear  bundle isomorphism  of $B$ to $B^*$ written $f\mapsto f^\sharp$,  and therefore 
with a complex conjugation $\Cc$ such that $f^\sharp=(\Cc f)^\star$, 
such as in the examples of the various Proca theories treated in this work,
one can replace the generators $\Ac(\lambda)$ by generators $ \widetilde{\Ac}(f):=\Ac( f^\sharp) =\Ac((\Cc f)^\star)$, also writing $ \overline{\Ac}(f)$ in place of $\Ac^\star(f)$.
In terms of these, the algebra $\Af_F(P, B)$ can be described as the free unital $*$-algebra generated by $\widetilde{\Ac}(f)$ and $\overline{\Ac}(f)$, $f\in \Gamma_0^\infty(B)$, and subject to the relations 
\begin{itemize}
    \item[F1a] $f\mapsto \widetilde{\Ac}(f)$ and $f\mapsto \overline{\Ac}(f)$ are complex linear,
    \item[F2a] $(\widetilde{\Ac}(f))^*=\overline{\Ac}(\Cc f)$, $(\overline{\Ac}(f))^*=\widetilde{\Ac}(\Cc f)$,
    \item[F3a] $\widetilde{\Ac}(\overline{P}f)=0=\overline{\Ac}(P f)$, where $\overline{P}=\Cc P \Cc :\Gamma^\infty(B)\to \Gamma^\infty(B)$ (field equations) ,
    \item[F4a] $\left[\widetilde{\Ac}(f), \overline{\Ac}(h)\right]=\ii\IP{ \Cc f, E_P h}_B\1$, and
    $\left[\widetilde{\Ac}(f), \widetilde{\Ac}(h)\right]=\left[\overline{\Ac}(f), \overline{\Ac}(h)\right]=0$ (commutation relations),
\end{itemize}
for all $f,h\in \Gamma_0^\infty(B)$.  As before, any region $N$ of $M$ determines a subalgebra $\Af_F(P|_N,B|_N)$ of $\Af_P(P,B)$.

A standard argument shows that $\Af(P,B)$ constructed in any one of the ways above satisfies the time-slice property:
If $N\subset M$ is a  region containing a Cauchy surface of $(M,g)$, then
\begin{align}
    \Af(P\vert_N, B\vert_N)=\Af(P,B)\, .
\end{align}

In particular, if $N\subset M$ is a causally convex set containing a Cauchy surface of $(M,g)$, and $\chi\in C^\infty(M;\RR)$ satisfies $\chi=1$ in a neighbourhood of $M\setminus J^+(N)$ and $\chi=0$ on a neighbourhood of $M\setminus J^-(N)$, then for any $F\in \Gamma^\infty_0(B)$, one has $P \chi E_P F\in \Gamma^\infty_0(B\vert_N)$ and $(\id-P\chi E_B)f\in P\Gamma_0^\infty(B)$, so that one can split $\Gamma^\infty_0(B)=\Gamma^\infty_0(B\vert_N)+ P\Gamma_0^\infty(B)$, and in the same way $\Gamma_0^\infty( B)=\Gamma_0^\infty( B\vert_N)+ \overline{P}\Gamma_0^\infty( B)$  or $\Gamma_0^\infty( B^*)=\Gamma_0^\infty( B^*\vert_N)+ \sadj{P}\Gamma_0^\infty( B^*)$, see e.g.~\cite{Baer:2015} or \cite[Theorem 3.1g]{Fewster:2025a}. By the relation F3(a) or R3, this implies the desired statement.

Moreover, by purely symbolic changes to the arguments presented in \cite[Appendix C]{FewVer_QFLM:2018}, and \cite[Lemma 3.1(i)]{FewVer:dynloc2}, it follows that if $B$  is a hermitian vector bundle, $\Af(P, B)$ constructed in one of the ways above satisfies the  \emph{Haag property}. Namely, let $K\subset M$ be compact, and let $A\in \Af(P, B)$ commute with all elements of $\Af(P\vert_N, B\vert_N)$ for all regions $N\subset M$ lying in the causal complement of $K$, i.e., $N\subset (M\setminus J(K))$. Then for any connected open causally convex set $L\subset M$ containing $K$, one has $A\in \Af(P\vert_L, B\vert_L)$.

\subsection{A general framework for Green-hyperbolic operators}
\label{subsec:General quantization}
 Most operators appearing in this paper will either be given as RFHGHOs or FHGHOs. However, one may wonder if there are cases where the FHGHO (or RFHGHO)-scheme does not apply. 
In fact, up to  the previous subsection, we had not made any additional assumptions on the vector bundles $B$, except that they have finite-dimensional fibres. In particular, we neither assumed the existence of a hermitian fibre metric, nor of a complex conjugation. Yet, there are in fact cases where the operator is Green-hyperbolic, but where the bundle it acts on lacks a canonical fibre metric and complex conjugation. One example of this case is the Teukolsky equation on the bundle of spin-weighted scalars \cite{Klein:2025, KleinHafner:2025}. 

In such a case, one can use a framework for the quantization that applies to all Green-hyperbolic operators, irrespective of the presence of additional structures in the bundle or additional properties satisfied by the operator (such as formal hermiticity). Since this framework does not seem to have been discussed in the literature before, we present it in some detail.  However, the reader who is interested specifically in the Proca theories and their applications can pass directly to Sect.~\ref{sec:algobs}.

In the following, let $B\xrightarrow{\pi}M$ be a complex finite-rank vector bundle over a globally hyperbolic spacetime $(M,g)$. Let $B^*$ be the dual bundle, i.e., $B^*_x$ consists of all linear maps from $B_x$ to $\CC$. Denote by $B^\bstar$ the anti-dual bundle whose fibres $B^\bstar_x$ consist of all antilinear maps from $B_x$ to $\CC$, and write the dual bundle of $B^\bstar$ as $\bar{B}$.
Then there is a natural linear isomorphism (in the case of finite rank) between $B$ and $B^{\bstar\bstar}$, the double anti-dual of $B$, given by the map
\begin{equation}
    \iota_{B,x}: B_x\ni v\mapsto \left(B^{\bstar}_x\ni \bar\lambda \mapsto \overline{\bar{\lambda}(v)}\right)\, 
\end{equation}
for any $x\in M$ (here, $\bar\lambda$ is used as notation for a typical element of the antidual).  We also have an antiduality pairing
\begin{equation}
    \bar\lambda(v)= \int_M \bar\lambda_x(v_x)\,\vol(x)
\end{equation}
for $v\in\Gamma^\infty(B)$, $\bar\lambda\in\Gamma^\infty(B^\bstar)$ with compactly intersecting supports. Every linear partial differential operator $P$ acting on $\Gamma^\infty(B)$ has a 
{\it formal anti-dual} $\dhadj{P}$, which is a linear partial differential operator on $\Gamma^\infty(B^\bstar)$ obeying
\begin{equation}
    (\dhadj{P} \bar\lambda)(v) = \bar\lambda(Pv)
\end{equation}
for all $\bar\lambda \in \Gamma^\infty(B^\bstar)$ and $v\in\Gamma^\infty(B)$ with compactly intersecting supports. If $P$ is a general linear operator on $\Gamma^\infty(B)$, the formal anti-dual is defined in the same way, if it exists.

One can then show
\begin{prop}\label{prop:extend}
    Let $\BB=B\oplus B^\bstar$ and consider fibre metrics on $\BB$ such that $\PP=P\oplus\dhadj{P}$ is formally hermitian for every linear partial differential operator $P$. Every such fibre metric takes the form
    \begin{equation}\label{eq:unique fibre metric}
         \IP{(v,\bar\lambda),(w,\bar\mu)}=\alpha\bar{\mu}(v)+\iota(w)(\alpha\bar{\lambda})=\alpha\bar{\mu}(v)+\overline{\alpha\bar{\lambda}(w)}
    \end{equation}
    for some constant $\alpha\in\CC\setminus \{0\}$. Conversely, Eq.~\eqref{eq:unique fibre metric} defines
    such a fibre metric for every constant $\alpha\in\CC\setminus \{0\}$.
\end{prop}

\begin{proof}
    Let $v,w\in \Gamma^\infty(B)$ and $\bar\lambda,\, \bar\mu\in \Gamma^\infty(B^\bstar)$ so that the supports of any pair of them have compact intersection. Then any fibre metric on $\BB$, acting on these sections, can be written as
\begin{equation}
 \IP{(v,\bar\lambda),(w,\bar\mu)}=
 (A_1w)(v)+(A_2\bar{\mu})(v)+\overline{\bar{\lambda}(A_3w)}+\overline{\bar{\lambda}(A_4\bar\mu)}
\end{equation}
where $A_1\in\Gamma^\infty(\Hom(B,B^\bstar))$, 
$A_2\in \Gamma^\infty(\End(B^\bstar))$, $A_3\in\Gamma^\infty(\End(B))$, and $A_4\in\Gamma^\infty(\Hom(B^\bstar,B))$.
The hermiticity of the fibre metric and $\PP$  implies the condition
\begin{align}
\label{eq: fib metric cond}
    &(A_1 P w)(v)+(A_2\dhadj{P}\bar\mu)(v)+\overline{\bar\lambda(A_3 Pw)}+\overline{\bar\lambda(A_4\dhadj{P}\bar\mu)}\nonumber\\
    =&\overline{(A_1 Pv)(w)}+\overline{(A_2\dhadj{P} \bar\lambda)(w)}+\bar\mu(A_3Pv)+\bar\mu(A_4\dhadj{P}\bar\lambda)\, 
\end{align}
  which has to hold for any $v,w\in \Gamma^\infty(B)$, $\bar\lambda,\bar\mu\in \Gamma^\infty(B^\bstar)$ as described above and any linear differential operator $P:\Gamma^\infty(B)\to\Gamma^\infty(B)$.

  In particular, setting $P=\id_\BB$, one obtains
    \begin{align}
        (A_2\bar\lambda)(w)&=\bar\lambda(A_3w)\quad \forall w\in \Gamma^\infty(B), \, \bar\lambda\in \Gamma^\infty(B^\bstar):\supp w\cap\supp \bar\lambda \text{ compact}        
    \end{align}
    which asserts that $A_2=\dhadj{A_3}$,  
    and
    \begin{align}
         (A_1v)(w)&=\overline{(A_1w)(v)} \quad \forall v,\, w\in \Gamma^\infty(B): \supp w\cap\supp v \text{ compact}\\
         \bar\mu(A_4\bar\lambda)&=\overline{\bar\lambda(A_4\bar\mu)} \quad \forall \bar\lambda,\, \bar\mu\in \Gamma^\infty(B^\bstar):\supp \bar\mu\cap\supp \bar\lambda \text{ compact}.
    \end{align}
    Combining the last two with the analogous conditions arising from $P= \ii \, \id$, one obtains
     \begin{equation}
        \ii(A_1w)(v)=(A_1\ii\,\id w)(v)=\overline{(A_1\ii \, \id v)( w)}=-\ii\overline{(A_1v)(w)}=-\ii (A_1w)(v) 
    \end{equation}
    for all $v,w\in \Gamma^\infty(B)$ with compactly overlapping support, implying $A_1=0$. In the same way, one asserts that $A_4=0$.
    In the next step, $\bar\lambda=0$ and $w=0$ entails the condition
    \begin{equation}
        \bar\mu(A_3Pv)=\bar\mu(PA_3v)\quad \forall v\in \Gamma^\infty(B) , \, \, \bar\mu\in \Gamma^\infty(B^\bstar):\supp v\cap\supp \bar\mu \text{ compact} 
    \end{equation}
     or equivalently
\begin{equation}
\label{eq:fib metric comm}
    [A_3, P]=0 
\end{equation}
for all linear differential operators $P:\Gamma^\infty(B)\to \Gamma^\infty(B)$.
Choosing $P$ of $0$-th order, i.e., $P\in \Gamma^\infty(\End(B))$, this necessitates $A_3=\alpha \id_B$, where $\alpha\in C^\infty(M)$. Plugging the result into \eqref{eq:fib metric comm} and assuming $P=P_0\id_B$, where $P_0\in {\rm PDO}(M)$ is a linear differential operator of arbitrary order acting on scalar functions, \eqref{eq:fib metric comm} reduces to
\begin{equation}
    [\alpha, P_0]=0 \quad \forall P_0\in {\rm PDO}(M)\, ,
\end{equation}
asserting that $\alpha\in \CC$ is a constant, and since the fibre metric is weakly non-degenerate the case $\alpha=0$ is excluded. Thus~\eqref{eq:unique fibre metric} is proved. The converse is a straightforward computation.
\end{proof}

This construction allows us to obtain a formally hermitian operator even in cases where a priori we do not even have a fibre metric to define a notion of (formal) hermiticity. Moreover, if $P$ is Green-hyperbolic, then so is $\PP$, with retarded and advanced Green operators given by $E^\pm_\PP=E^\pm_P \oplus E^\pm_{\dhadj{P}}$, thus constituting a formally hermitian Green-hyperbolic operator (FHGHO) which can then be quantized  following the FHGHO scheme discussed above.

Proposition~\ref{prop:extend} demonstrates that any Green-hyperbolic operator on a finite-rank vector bundle can be enlarged to a FHGHO. Physically, the method presented here amounts to introducing additional fields valued in the anti-dual bundle. The constant $\alpha$ that remained as a freedom in the choice of the fibre metric of $\BB$ can be absorbed into the definition of the additional fields, so one can restrict to $\alpha=1$. The enlarged theory has the action of the form~\eqref{eq:complex_action} for $\PP$, namely 
\begin{equation}
    S[v, \bar\lambda]=
    -\IP{(v,\bar\lambda),\PP(v,\bar\lambda)}= - (\dhadj{P}\bar\lambda)(v)-\overline{\bar\lambda(Pv)}=-2\Re \bar\lambda(Pv) \, ,
    % \frac{1}{2}\left(\bar\lambda(Pv)+\overline{\bar\lambda(Pv)}\right) \, ,
\end{equation}
modulo cutoffs, see \cite{Iuliano} for this action (up to an overall factor) in the example of the Teukolsky scalar \cite{Toth}.
In the case $\alpha=1$, using the FHGHO quantisation scheme F1--F4, one has smeared fields $\Ac^\star(v,\bar\lambda)$ labelled by $v\in\Gamma_0^\infty(B)$ and $\bar\lambda\in \Gamma_0^\infty(B^\bstar)$, obeying the commutation relation 
\begin{equation}
    \left[ \Ac^\star(v,\bar\lambda)^*,  \Ac^\star(w,\bar\mu)\right]=\ii 
\IP{(v,\bar\lambda),E_{\PP} (w,\bar\mu)}\1 =\ii ( (E_{\dhadj{P}}\bar\mu)( v) + \overline{\bar{\lambda}(E_P w)})\1  
\end{equation}
(see~\eqref{eq:CCR_equiv}). In this sense we see that the additional fields are conjugate to those labelled by test functions in $B$.

 To conclude this subsection, we  briefly consider circumstances in which $\PP$ or related operators are \rfhgho{}s, leaving most details to the reader. We will employ several natural maps. Apart from the natural isomorphism $\iota_B: B\to B^{\bstar\bstar}$   already mentioned, and the usual natural isomorphism between $B$ and $B^{**}$, which will be denoted $\tilde\iota_B$, there are natural anti-linear, invertible maps between the dual- and anti-dual spaces given by 
    \begin{align}
        C_B:B^*\to B^\bstar\, , \quad B^*_x\ni\lambda\mapsto\left(B_x\ni v\mapsto \overline{\lambda(v)}\right)\, .
    \end{align} 
There are two additional structures that can turn $\PP$ or related operators into \rfhgho{}s. 

First, assume that there is a linear bundle isomorphism $I:B\to B^{\bstar *}$, so that $\Cc_B = \iota_B^{-1}(C_{B^\bstar}) I$ satisfies $\Cc_B^2=\id_B$. Then $\Cc_B$ is a complex conjugation on $B$, and
$\Cc_{B^\bstar}:=C_BI^*\tilde{\iota}_{B^\bstar}$ is a complex conjugation on $B^\bstar$  that satisfies $(\Cc_{B^\bstar}\bar\lambda)(v)=\overline{\bar\lambda(\Cc_B v)}$ for all $\bar\lambda\in\Gamma^\infty(B^\bstar)$ and $v\in\Gamma^\infty(B)$.
As a result, $\Cc_\BB = \Cc_B\oplus \Cc_{B^\bstar}$ is a complex conjugation on $\BB$, giving $\BB$ a hermitian vector bundle structure. Moreover, any Green-hyperbolic operator $P$ that is $\Cc_B$-real, i.e., satisfies $\Cc_BP\Cc_B=P$, determines a \rfhgho\ $\PP$.

In this setup, one can alternatively consider the enlarged bundle $B\oplus B^*$. It can be equipped with the bilinear symmetric form
\begin{equation}
    \left((v,\lambda),(w,\mu)\right)=\lambda(w)+\mu(v)\, ,
\end{equation}
which is uniquely fixed up to rescaling (in dimension $>1$) by requiring that $P\oplus \sadj{P}$ is symmetric for any linear differential operator $P:\Gamma^\infty(B)\to\Gamma^\infty(B)$. If one is given a bundle isomorphism $I$ as above, $\Cc_{B^*}=I^*\tilde{\iota}_{B^\bstar}C_B$ is a complex conjugation on $B^*$ and $\Cc_{B\oplus B^*}=\Cc_B\oplus \Cc_{B^*}$ is a complex conjugation on $B\oplus B^*$, equipping it with a Hermitian vector bundle structure with hermitian fibre metric $\IP{\cdot,\cdot}=(\Cc_{B\oplus B^*}\cdot,\cdot)$. Then any $\Cc_B$-real, Green-hyperbolic $P$ determines a \rfhgho\ $P\oplus \sadj{P}$.

In fact, in this scenario there is a linear isomorphism of hermitian vector bundles (i.e., the isomorphism is unitary and intertwines the complex conjugations) from $\BB$ to $B\oplus B^*$ given by $\id_B\oplus I^*\circ \tilde{\iota}_{B^\bstar}$ and we have $I^*\tilde{\iota}_{B^\bstar}\dhadj{P}=\sadj{P}I^*\tilde{\iota}_{B^\bstar}$ for any $\Cc_B$-real operator $P$.

Second, assume that there is a linear bundle isomorphism  $J:B\to B^*$ satisfying $J=\sadj{J}\tilde{\iota}_B$. Let $c:=C_BJ:B\to B^\bstar$. Then $\Cc_\BB:\BB\ni (f,\bar\lambda) \mapsto(c^{-1}\bar\lambda, c f)\in \BB $ is a complex conjugation on $\BB$ that equips $\BB$ with a hermitian vector bundle structure. Moreover, any  Green-hyperbolic linear differential operator $P:\Gamma^\infty(B)\to \Gamma^\infty(B)$ that satisfies $cP=\dhadj{P}c$ determines a \rfhgho\ $\PP$.

Alternatively, one may consider $B\oplus B^{\bstar*}=B\oplus \bar B$. On this enlarged bundle, there is a complex conjugation given by $\Cc(v,\bar w)=(\tilde{c}^{-1}\bar w, \tilde{c} v)$, where $\tilde{c}:=C_{B^\bstar}^{-1}\iota_B: B\to \bar B$. The bundle isomorphism $J$ induces a hermitian fibre metric $\IP{(v,\bar w), (x,\bar y)}=\bar y(cv)+\overline{\bar w(cx)}$ compatible with the complex conjugation, giving $B\oplus \bar B$ a hermitian vector bundle structure. Moreover, any  Green-hyperbolic $P$ that satisfies $cP=\dhadj{P}c$ determines a \rfhgho\ $P\oplus \overline{ P}$, where $\overline{P}=\tilde{c}P\tilde{c}^{-1}$.

As in the previous scenario, there is a linear isomorphism of hermitian vector bundles between $\BB$ and $B\oplus \bar B$ given by $\id_B\oplus \tilde{c}c^{-1}$ and we have $\overline{P} \tilde{c}c^{-1}=\tilde{c}c^{-1}\dhadj{P}$ for all $P$ satisfying $cP=\dhadj{P}c$.

\subsection{Algebras of observables for Proca theories}\label{sec:algobs}

In what follows we will be interested in the theories of the real, uncharged, Proca field and the charged Proca field with a general magnetic moment, on a general globally hyperbolic spacetime $(M,g)$ with background electromagnetic field $A\in\Gamma^\infty(T^*M)$, as well as in the coupled and uncoupled theories of the complex Proca--scalar system in a vanishing background electromagnetic field. 

$\Lambda^pM$ is a hermitian vector bundle with complex conjugation given by the usual complex conjugation, and the hermitian bundle metric corresponding to the pairing \eqref{eq:hermpair}. Further, all operators of interest are formally hermitian and Green-hyperbolic, so that no augmentation of the theory is necessary to construct the algebras of observables.

It is useful to adopt a more formal notation for the background spacetime and vector field,
representing an orientation on a Lorentzian spacetime by an equivalence class $\ogth$ of non-vanishing $4$-forms, and a time-orientation $\tgth$ by an equivalence class of non-vanishing timelike $1$-forms, where in both cases $\omega\sim \omega'$ if and only if $\omega=f\omega'$ for some strictly positive $f\in C^\infty(M)$. A tuple $(M,g,\ogth,\tgth)$ is an \emph{oriented globally hyperbolic spacetime} if $(M,g)$ is a ($4$-dimensional) globally hyperbolic spacetime equipped with time-orientation $\tgth$ and orientation $\ogth$. We typically denote such tuples by $\Mb$ and write $M$ for the underlying manifold of $\Mb$ when there is no ambiguity. If an oriented globally hyperbolic spacetime $\Mb$ is equipped with a background
electromagnetic field $A\in \Gamma^\infty(T^*M)$, then the pair $(\Mb,A)$ constitutes a background, typically denoted $\Mbb$. Again, we write $M$ for the underlying manifold of $\Mbb$ when this does not create ambiguity.
 
Throughout, we keep the real parameters $m>0$, $q\neq 0$ and $\kappa$ fixed but arbitrary and write
\begin{equation}
\label{eq:charged Proca on M operator}
	P_{\Mbb}=-\delta_A \dd_A  + \ii q \kappa F^\bullet + m^2
\end{equation}	
for the charged Proca operator on background $\Mbb$. The uncharged Proca operator can be defined on any oriented globally hyperbolic spacetime $\Mb$ by $P_\Mb=P_{(\Mb,0)}$.
The complex conjugate operator $\overline{P}_{\Mbb}$ is given by $\overline{P}_{\Mbb}f = \overline{P_{ (\Mb,A)}\overline{f}}$, so $\overline{P}_{(\Mb,A)}=P_{(\Mb,-A)}$ and in particular,
$P_{\Mb}=\overline{P}_{\Mb}$.
The Green operators and advanced-minus-retarded solution operator for $P_{\Mbb}$ (resp., $P_\Mb$) will be written $E^\pm_{\Mbb}$ and $E_{\Mbb}$ (resp., $E^\pm_\Mb$ and $E_\Mb$).
Recalling the hermitian pairing $\IP{\cdot,\cdot}_p$ for $p$-form fields in \eqref{eq:hermpair}, note that 
\begin{align}
\IP{P_\Mbb f, E_{\Mbb}h}_1 &= 0 = \IP{f,E_\Mbb P_{\Mbb}h}_1\nonumber\\
\IP{P_\Mb f, E_{\Mb}h}_1 &= 0 = \IP{f,E_\Mb P_{\Mb}h}_1
\end{align}
for all $f,h\in \Gamma_0^\infty(\Lambda^1M)$.

Starting with the uncharged theory (representing $Z^0$-bosons, for example): on oriented globally hyperbolic spacetime $\Mb$, the unital $*$-algebra of observables $\Zf(\Mb)$ is generated by symbols $\Zc_\Mb(f)$ labelled by $f\in\Gamma_0^\infty(\Lambda^1M)$ and subject to the relations R1-R4, where $\Zc_\Mb(f)$ plays the role of $\Ac(f)$.
One may check that the sign in R4 agrees with Peierls' method applied to the Lagrangian $-\tfrac{1}{2}(\dd Z)\wedge\star \dd Z+\tfrac{1}{2}m^2 Z\wedge\star Z$ and the Dirac quantization prescription $[\hat{X},\hat{Y}]=\ii\widehat{\{X,Y\}}$. 

Meanwhile, the theory of the charged field (representing $W^\pm$-bosons,  for example) on background $\Mbb$ is described by the unital $*$-algebra $\Wf(\Mbb)$ generated by symbols $\Wc_{\Mbb}(f)$ and $\overline{\Wc}_{\Mbb}(f)$ labelled by $f\in\Gamma_0^\infty(\Lambda^1M)$ and subject to the relations F1a-F4a, with $\Wc_\Mbb(f)$ replacing $\widetilde{\Ac}(f)$ and $\overline{\Wc}_\Mbb(f)$ replacing $\overline{\Ac}(f)$.

An important fact is that $\Wf((\Mb,0))\cong \Zf(\Mb)\otimes\Zf(\Mb)$, i.e., the charged field on a vanishing background is isomorphic to two copies of the uncharged field. The correspondence is given by
\begin{align}
\label{eq:0 bckgnd corresp}
 \Wc_{(\Mb,0)}(f) &\mapsto \frac{1}{\sqrt{2}}\left( \Zc_{\Mb}(f)\otimes\1 + \ii \1\otimes  \Zc_{\Mb}(f) \right) \nonumber\\
  \overline{\Wc}_{(\Mb,0)}(f) &\mapsto \frac{1}{\sqrt{2}}\left( \Zc_{\Mb}(f)\otimes\1 - \ii \1\otimes  \Zc_{\Mb}(f) \right) 
\end{align}
for $f\in \Gamma_0^\infty(\Lambda^1 M)$. 

Before moving on to the Proca--scalar system, we briefly consider the theory of the complex scalar field governed by the operator 
\begin{align}
    K_\Mbb= -\delta_A \dd_A+\fm^2:\,\Gamma^\infty(\Lambda^0M)\to\Gamma^\infty(\Lambda^0M)\, ,
\end{align}
for some fixed real parameter $\fm>0$. There is a strong analogy to the theory of the Proca field. In particular, the uncharged Klein--Gordon operator is given by $K_\Mb=K_{(\Mb,0)}$ and the complex conjugate operator satisfies $\bar{K}_{(\Mb,A)}=K_{(\Mb, -A)}$, so $K_\Mb$ is real. The algebra of observables for the uncharged theory, $\Sf(\Mb)$, is the free unital $*$-algebra generated by $\phi_\Mb(f)$, $f\in \Gamma_0^\infty(\Lambda^0M)$, satisfying relations R1-R4 with $\phi_\Mb(f)$ in place of $\Ac(f)$, while the algebra of observables  $\Cf(\Mbb)$ of the charged scalar is the free unital $*$-algebra generated by $\Phi_\Mbb(f)$ and $\overline{\Phi}_\Mbb(f)$, $f\in \Gamma_0^\infty(\Lambda^0M)$, satisfying relations F1a-F4a with $\Phi_\Mbb(f)$ and $\overline{\Phi}_\Mbb(f)$ replacing $\widetilde{\Ac}(f)$ and $\overline{\Ac}(f)$, correspondingly. On a vanishing background, one has $\Cf((\Mb,0))\cong \Sf(\Mb)\otimes\Sf(\Mb)$, and the correspondence is given by the analogue of \eqref{eq:0 bckgnd corresp}.

Lastly, we discuss the quantization of the complex Proca--scalar theory in the absence of a background electromagnetic field. Let us denote $(\Mb, v)=\Mc$, where $v\in \Gamma_0^\infty(\Lambda^1M)$ is a real-valued coupling $1$-form. Then we write
\begin{align}
    P_\Mc=\begin{pmatrix}
        P_\Mb & v \\ -v^\bullet & K_\Mb
    \end{pmatrix}
\end{align}
for the coupled Proca-scalar operator acting on $\Gamma^\infty(\Lambda^1M\oplus \Lambda^0M)$ in the background $\Mb$ with coupling $v$. The complex conjugate operator is given by $\overline{P}_\Mc=P_\Mc$, since $v$ is real. This theory is described by the $*$-algebra $\Uf(\Mc)$ generated by symbols $\Uc_\Mc(f\oplus h)$ and $\overline{\Uc}_\Mc(f\oplus h)$, with $f\in \Gamma_0^\infty(\Lambda^1M)$ and $h\in \Gamma^\infty(\Lambda^0M)$, and subject to the relations F1a-F4a with $\widetilde{\Ac}(f\oplus h)$ given by $\Uc_\Mc(f\oplus h)$ and $\overline{\Ac}(f\oplus h)$ by $\overline{\Uc}_\Mc(f\oplus h)$.

For $v=0$, $P_\Mc=P_\Mb\oplus K_\fm$ is the uncoupled Proca--scalar theory. 
Similar to the identification of $\Wf((\Mb,0))$ with two copies of $\Zf(\Mb)$, $\Uf((\Mb,0))\cong \Wf((\Mb,0))\otimes \Cf((\Mb,0))$, i.e., $\Uf((\Mb,0))$ is isomorphic to the tensor product of the algebras $\Wf((\Mb,0))$ of the complex Proca field and $\Cf((\Mb,0))$ of the complex scalar field in vanishing background, so that the algebra of the uncoupled system is generated by the generators of the free Proca and scalar fields. In detail, the correspondence is given by 
\begin{align}
    \Uc_{(\Mb,0)}(f\oplus h)&\mapsto \Wc_{(\Mb,0)}(f)\otimes \1_\Cf + \1_\Wf\otimes \Phi_{(\Mb,0)}(h)\, , \\
    \overline{\Uc}_{(\Mb,0)}(f\oplus h)&\mapsto \overline{\Wc}_{(\Mb,0)}(f)\otimes \1_\Cf + \1_\Wf\otimes \overline{\Phi}_{(\Mb,0)}(h)\, . 
\end{align}

\subsection{Hadamard states}
\label{sec:states}
We briefly comment on the existence of Hadamard states for the QFTs discussed in this work,  following the definition of Hadamard states for general $\Vc^\pm$-decomposable RFHGHOs
%\footnote{ See Definition~\ref{def:decomposable_sGHO} and Remark~\ref{rem:RFHGHO case}.} 
which was developed recently in \cite{Fewster:2025a}. Given a $\Vc^\pm$-decomposable RFHGHO $P$ acting on sections of a hermitian vector bundle $B$ over some spacetime $(M,g)$, let $\Af(P,B)$ be the associated algebra as in Sect.~\ref{sec:gen framework}. Then a state $\omega:\Af(P,B)\to \CC$ is  $\Vc^+$-Hadamard if the two-point function $W\in \Dc'(B\boxtimes B)$ defined by
\begin{align}
    W(\rho f^\#, \rho h^\#)=\omega(\Ac(f) \Ac(h))\, , \quad f,h\in \Gamma_0^\infty(B)
\end{align}
satisfies
\begin{align}
    \WF(W)\subset \Vc^+\times \Vc^-\, .
\end{align}
Here, $\rho$ is the metric volume density of the spacetime, and $\#:B\to B^*$ is the linear isomorphism induced by the fibre metric and complex conjugation of $B$. See \cite{Fewster:2025a} for more details. 

In this paper, we will only be interested in the case $\Vc^\pm=\Nc^\pm$, the bundles of nonzero future/past-directed null covectors, and refer simply to `Hadamard' rather than `$\Nc^+$-Hadamard'.
For scalar fields satisfying the Klein--Gordon equation, this is equivalent to earlier formulations of the Hadamard property \cite{Radzikowski_ulocal1996, SahlmannVerch:2000RMP}, and their existence on general globally hyperbolic spacetimes is well-established \cite{FullingSweenyWald:1978, FullingNarcowichWald}. This extends to normally hyperbolic formally hermitian operators on bundles with positive definite fibre metrics \cite{IslamStrohmaier:2020,FewsterStrohmaier:2025}.

For the real Proca field, it was shown in \cite[Theorem 7.2]{Fewster:2025a} that the theory has Hadamard states on any globally hyperbolic spacetime. A close inspection of the state constructed in \cite[Theorem 3]{MorettiMurroVolpe:2023} for Proca fields on ultrastatic spacetimes reveals that the Minkowski vacuum state for real Proca fields is Hadamard, as one would expect. 

In the absence of a background electromagnetic field, the equation for the charged Proca field is real, and the theory of the charged Proca field can be identified with two copies of the theory of the real Proca field. By \cite[Theorem 5.23]{Fewster:2025a}, this implies the existence of Hadamard states for the charged Proca theory in the case of vanishing electromagnetic background fields. The same result implies the existence of Hadamard states for the uncoupled Proca-scalar system and Proca multiplets with constant diagonal mass matrix.

To show the existence of Hadamard states in the case of non-vanishing electromagnetic background fields, variable mass matrix of the multiplet, or non-vanishing coupling of the Proca-scalar system, we note that all of these theories are governed by $\Nc^\pm$-decomposable operators. We focus on the charged Proca theory, as the proofs for the other two cases work in the same way. Let $\Mbb=(\Mb,A)$, let $P_{\Mbb}$ denote the charged Proca operator as in \eqref{eq:charged Proca on M operator}, and let $P_{\Mb}=P_{(\Mb,0)}$ be the operator of the charged Proca with vanishing background field.
 Let $\omega_0$ be any Hadamard state on $\Wf(\Mb,0)$ (such states exist, as noted above).
 Choose Cauchy surfaces $\Sigma_\pm\subset M$ and $\chi\in C^\infty(M;\RR)$ with $\chi=0$ on a neighbourhood of $M\setminus J^+(\Sigma_+)$ and $\chi=1$ on a neighbourhood of $M\setminus J^-(\Sigma_-)$.
We then consider the intermediate operator $P_{(\Mb, \chi A)}$. By construction, we can find a causally convex open set $N\subset J^-(\Sigma_-)$ containing a Cauchy surface of $M$, on which $P_{(\Mb,\chi A)}$ agrees with $P_{\Mb}$. 
 This agreement, in combination with Theorem 3.5e), Lemma 5.15c), and Theorem 5.16a) of \cite{Fewster:2025a} implies that we can  pull $\omega_0$ back to a Hadamard state of $\Wf(\Mb|_N,0)$, the theory on $N$, and then push this state forward to a Hadamard state $\omega_\chi$ on $\Wf(\Mb,\chi A)$. We can now choose a causally convex open set $N'\subset J^+(\Sigma_+)$. On this set, $P_\Mbb$, agrees with $P_{(\Mb, \chi A)}$, and repeating the above argument, one can obtain a Hadamard state of the fully interacting theory $\Wf(\Mbb)$ governed by $P_\Mbb$. Thus we have established the existence of Hadamard states for the theories with non-vanishing background electromagnetic field on any globally hyperbolic spacetime.

\section{Application:  a polarization-sensitive detector}\label{sec:Malus}

As an application of, and motivation for, our analysis of the Proca--scalar model, we show how it can be used to model a polarization-sensitive detector for measuring a complex Proca field. This can form a building block in a broader discussion of measurements involving polarization. The test of our detector will be that it replicates Malus' law  at leading order.

\subsection{Recall of measurement theory}
    The measurement framework of Fewster and Verch \cite{FewVer_QFLM:2018} describes how observables of a probe  field theory described by an algebra $\Bf$ can be used to extract information about a system theory with algebra $\Af$, given that the two interact in a (time-)compact coupling zone $K$ of spacetime.
    In the out-/in-regions $M^\pm =M\setminus J^\mp(K)$, one can then find unit-preserving, invertible $*$-homomorphisms $\tau^\pm$ between the uncoupled probe-system algebra $\Af\otimes\Bf$ and the algebra $\Cf$ of the coupled theory, given that both theories satisfy the time-slice axiom.
    For a measurement, system and probe are prepared in $M^-$ in state $((\tau^-)^{-1})^*(\omega\otimes \sigma)$, where $\omega$ and $\sigma$ are states for $\Af$ and $\Bf$. Then, in $M^+$,  the probe observable $B\in \Bf$, corresponding to the observable $\tau^+(\1\otimes B)$ in $\Cf$, is measured. 
    The resulting expected outcome can be expressed in terms of the system only, as the expectation in the original system state $\omega$, 
    \begin{align}
      (((\tau^-)^{-1})^*(\omega\otimes \sigma))( \tau^+(\1\otimes B))=\omega(\Ec_\sigma(B))\, ,
    \end{align}
     where $\Ec_\sigma(B)=\eta_\sigma(\Theta(\1\otimes B))$ is the induced system observable, with $\Theta=(\tau^-)^{-1}\circ\tau^+$ the scattering map and $\eta_\sigma(A\otimes B)= \sigma (B)A$ a partial trace over $\Bf$.

     Let $\Af$ and $\Bf$ be constructed from RFHGHOs (or FHGHOs) $P_a: \Gamma^\infty(E)\to \Gamma^\infty(E)$ and $P_b: \Gamma^\infty(F)\to \Gamma^\infty(F)$ as discussed in Sect.~\ref{sec:gen framework}, so that $\Af\otimes \Bf$ is the algebra for $P_0:=P_a\oplus P_b$,  with generators
     $\Psi(e\oplus f)$ ($e\in \Gamma_0^\infty(E)$, $f\in \Gamma_0^\infty(F)$). 
    Assume $\Cf$ is  also the algebra corresponding to a RFHGHO (or FHGHO) $P: \Gamma^\infty(E\oplus F)\to \Gamma^\infty(E\oplus F)$ so that $P=P_0$ on $M\setminus K$, and 
    \begin{align}
        P=\begin{pmatrix}
            P_a & R\\  \hadj{R} & P_b
        \end{pmatrix} =
        P_0 + \Rc, \qquad \Rc = \begin{pmatrix}
        	0 & R \\ \hadj{R}  &0
        \end{pmatrix}
    \end{align}
    for some differential operator $R: \Gamma^\infty(F)\to \Gamma^\infty(E)$ and its formal adjoint $\hadj{R}:\Gamma^{\infty}(E)\to\Gamma^\infty(F)$. 
     In this case, the scattering morphism takes the form \cite[Eq. (4.11)]{FewVer_QFLM:2018}
    \begin{align}
        \Theta(\Psi(e\oplus f))=\Psi(e\oplus f-(P-P_0) E_P^-(e\oplus f))\, .
    \end{align} 
% \todo[inline,caption='']{Apparently not needed\\
% Replacing $R$ by $\lambda R$, for $\lambda\in\RR$, one has a Born expansion
% \begin{equation}\label{eq:Born_exp}
% 	\Theta(\Psi(e\oplus f))=\Psi(e\oplus f-\lambda \Rc E_{P_0}^- (e\oplus f) + \lambda^2 \Rc E_{P_0}^-\Rc E_{P_0}^-(e\oplus f)+\mathcal{O}(\lambda^3)).  
% \end{equation}
% { Here, the convergence of the Born approximation of the test function should be understood in light of the holomorphicity of the functions $\lambda\mapsto E^\pm_{P_0+\lambda \Rc}$ in the topology of bounded linear operators between $H^s_0$ and $H^{s+1}_\loc$ for any $s\in \RR$, see \cite[Corollary 3.2]{Fewster:2023} for more details. That is to say, by $\mathcal{O}(\lambda^3)$ we mean that $\lambda^{-3}\norm{\left[(P-P_0)E_P^--\lambda \Rc E^-_{P_0}+\lambda^2 \Rc E^-_{P_0}\Rc E^-_{P_0}\right](e\oplus f)}_{H^{s+1}_\loc}$ has a finite limit as $\lambda \to 0$.}
% }

\subsection{ Minkowski vacuum representation}

We begin with some notation. If $(\fh, \IP{\cdot,\cdot})$ is a complex Hilbert space, then $\Fc_s(\fh)$ will denote the symmetric Fock space over $\fh$, 
with annihilation and creation operators $c_\fh(u)$, $c_\fh^*(u)$ 
labelled by $u\in\fh$ and satisfying the commutation relations
\begin{equation}
    [c_\fh(u),c_\fh(v)]=0=[c^*_\fh(u),c^*_\fh(v)] , \qquad 
    [c_\fh(u),c^*_\fh(v)]=
    \IP{u,v}_\fh \1_{\Bc(\Fc_s(\fh))}
\end{equation}
on suitable domains in $\Fc_s(\fh)$, together with $c^*_\fh(u)=c_\fh(u)^*$. In addition, the map $u\mapsto c^*_\fh(u)$ is linear, while $u\mapsto c_\fh(u)$ antilinear. The Fock vacuum vector is denoted $\Omega_\fh\in \Fc_s(\fh)$.

Now let $\Mb$ denote Minkowski spacetime with metric $\eta = \Diag(+1,-1,-1,-1)$ in inertial Cartesian coordinates. The Minkowski vacuum representations of the two algebras $ \Wf((\Mb,0))$ and $ \Cf((\Mb,0))$ of the charged Proca and scalar field are defined on the Hilbert spaces
\begin{align}
    \Hc_W&=\Fc_s(\fh_W)\otimes \Fc_s(\fh_W)\\
    \Hc_\Phi&=\Fc_s(\fh_\Phi)\otimes \Fc_s(\fh_\Phi)\, ,
\end{align}
where the one-particle Hilbert spaces are given by
\begin{align}
    \fh_W&=\left\{v\in L^2(\RR^3, (2\pi)^{-3}\dd^3\cv{k})\otimes \CC^4:\, \IP{v(\cv{k}), k}=0,\, k=(\omega(\cv{k}),\cv{k})\right\}\\
    \fh_\Phi&=L^2(\RR^3, (2\pi)^{-3}\dd^3 \cv{k})\, ,
\end{align}
respectively, with $\omega(\cv{k})=\sqrt{m^2+\abs{\cv{k}}^2}$. The inner product on $\CC^4$  used in defining $\fh_W$ is given by 
\begin{align}\label{eq:C4IP}
    \IP{v,w}=-\eta^{-1}(\overline{v},w)\, ,
\end{align}
where $\eta$ is the Minkowski metric.\footnote{ Eq.~\eqref{eq:C4IP} is consistent with choosing the hermitian form on the bundle $\Lambda^pM$ as $ \IP{w,v}=(-1)^p\star^{-1}(\overline{w}\wedge \star v)$.} It is positive definite on $\{v\in\CC^4: \IP{v,k}=0\}$ for any timelike $k\in\RR^4$.

The representations $\pi_W$ and $\pi_{\Phi}$ act on the smeared field operators generating the algebras of observables $\Wf(\Mb,0)$ and $\Cf(\Mb,0)$ by
\begin{align}
    \pi_W (\Wc(f))&=a(\Kc\overline{f})+b^*(\Kc f)\, ,\quad f\in \Gamma_0^\infty(\Lambda^1M)\\
    \pi_\Phi (\Phi(h))&=a_\Phi(\Kc_\Phi\overline{ h})+b_\Phi^*(\Kc_\Phi h) \, , \quad h\in \Gamma_0^\infty(\Lambda^0M)=C_0^\infty(M)\, ,
\end{align}
where 
\begin{align}
a(u)& =c_{\fh_W}(u)\otimes\1_{\Bc(\Fc_s(\fh_W))},  \qquad  
b(u) =\1_{\Bc(\Fc_s(\fh_W))}\otimes c_{\fh_W}(u),\\
 a_\Phi(v) &=c_{\fh_\Phi}(v)\otimes\1_{\Bc(\Fc_s(\fh_\Phi))} \qquad  
b_\Phi(v) =\1_{\Bc(\Fc_s(\fh_\Phi))}\otimes c_{\fh_\Phi}(v), 
\end{align} 
for $u\in\fh_W$, $v\in\fh_\Phi$,
and
\begin{align}
    (\Kc f)_\mu(\cv{k})&=(2\omega(\cv{k}))^{-1/2}\Pi^{\phantom{\mu}\nu}_\mu\widehat{f}_\nu(k)\vert_{k_0=\omega(\cv{k})}\, ,\\
    \Pi^{\phantom{\mu}\nu}_\mu(k)&=\delta_\mu^{\phantom{\mu}\nu}-m^{-2}k_\mu k^\nu\, ,\\
    (\Kc_\Phi h)(\cv{k})&=(2\omega_\Phi(\cv{k}))^{-1/2}\widehat{h}(k)\vert_{k_0=\omega_\Phi(\cv{k})}\, ,
    \end{align}
with $\omega_\Phi(\cv{k})=\sqrt{\fm^2+\abs{\cv{k}}^2}$.
Here, the hat denotes the Fourier transform in $\RR^4$, defined using the nonstandard convention
\begin{align}
    \widehat{h}(k)=\int\limits_{\RR^4} \dd^4 x\, e^{\ii k\cdot x} h(x)\, ,
\end{align}
 where $k$ is a covector, and the same definition is used for elements of $\Gamma_0^\infty(\Lambda^1 M)$ by applying it component-wise in Cartesian coordinates.

The only non-trivial commutation relations are
\begin{align}
    \left[a(F), a^*(F')\right]&=\left[b(F), b^*(F')\right]=\IP{F, F'}_{\fh_W}\1_{\Bc(\Hc_W)}\, , \quad F,\, F'\in \fh_W\\
    \left[a_\Phi(F), a_\Phi^*(F')\right]&=\left[b_\Phi(F), b_\Phi^*(F')\right]=\IP{F,F'}_{\fh_\Phi}\1_{\Bc(\Hc_\Phi)}\, ,\quad F,\, F'\in \fh_\Phi
\end{align}

 Finally, the Minkowski vacuum states for the Proca and scalar theories are given by
\begin{align}
   \Omega_W = \Omega_{\fh_W}\otimes \Omega_{\fh_W}\in \Hc_W
\,,\qquad 
   \Omega_\Phi  = \Omega_{\fh_\Phi}\otimes \Omega_{\fh_\Phi}\in \Hc_\Phi.
\end{align}
Let $S\in \fh_W$ be chosen such that $\norm{S}_{\fh_W}=1$. Then an $n$-particle single-mode state for the Proca field is given by 
\begin{align}
\label{eq:one particle state}
   \omega_{S,n}(A)=\frac{1}{n!}\IP{(a^*(S))^n \Omega_W, \pi_W(A) (a^*(S))^n\Omega_W}_{\Hc_W}\, .
\end{align}

To model a mode in terms of a wave packet with sharply-peaked momentum and purely transversal polarization in the laboratory frame, we start by defining $\tilde{S}_\mu\in \fh_W$ with the form 
\begin{align}
\tilde{S}_\mu(\cv{k})=s(\cv{k}) \epsilon_\mu(\cv{k}).
\end{align}
 Here $s\in C^\infty(\RR^3;\CC)$  is supported in a small neighbourhood of some $\cv{k_0}\neq 0$, and satisfies 
\begin{equation}
\norm{s}_{L^2(\RR^3, (2\pi)^{-3} \dd^3\cv{k})}=1\, ,
\end{equation}
while $\epsilon_\mu(\cv{k})$ is a real covector field satisfying 
\begin{equation}\label{eq:k_perp_epsilon}
    \IP{\epsilon(\cv{k}),k}=0\, 
\end{equation}
on $\supp(s)$. Typically one chooses $\epsilon(\cv{k})$ to be (at least) approximately normalised $\IP{\epsilon(\cv{k}), \epsilon(\cv{k})}\approx 1$ on $\supp(s)$.

To specify $\epsilon$ further, one can choose Cartesian coordinates so that $\cv{k_0}=(0,0,k_0)$. Then in a small neighbourhood thereof, we can parametrize $k_\mu$ as 
\begin{align}
\label{eq:k param}
    k_\mu=(\omega(k), k\sin\alpha\cos\beta, k\sin\alpha\sin\beta, k\cos\alpha)\, ,
\end{align}
where $k=\abs{\cv{k}}$, $\alpha$ is the angle between $\cv{k}$ and the $z$-axis, and $\beta$ is the angle in the $xy$-plane (see Fig.~\ref{fig:coordinates}). The assumption that $s(\cv{k})$ is supported in a small neighbourhood of $\cv{k_0}$ can then be translated to the assumption that $s(k,\alpha,\beta)$ is supported in $I_k\times[0,\alpha_{\rm max})\times[0,2\pi)$, where $I_k\subset \RR$ is a small interval containing $k_0$, and $\alpha_{\rm max}<\tfrac{\pi}{2}$.
\begin{figure}
    \centering
    \begin{tikzpicture}[scale=1]
\draw(0,0) circle(3pt) node[above]{z};
\draw[fill=black] (0,0) circle(0.5pt);
\draw[->](0.1,0)--(1.45,0) node[midway, below]{$k\sin\alpha$};
\draw[->](0,1.5) arc(90: 440:1.5) node[midway, below]{$\beta$};
\end{tikzpicture}
    \caption{The parametrization of $k_\mu$ viewed from above the $xy$-plane}
    \label{fig:coordinates}
\end{figure}
The polarization covector representing linear polarization in the $x$-direction, as measured in the laboratory frame, should satisfy~\eqref{eq:k_perp_epsilon} and have vanishing $t$ and $y$ components, and is therefore uniquely fixed up to scale by
\begin{align}
\label{eq:Proca x-polarization}
     \epsilon_\mu^x(\cv{k})=\left(0, 1, 0, -\tan\alpha \cos\beta\right);
\end{align}
similarly, the polarization covector representing linear polarization in the $y$-direction is 
\begin{align}
\label{eq:Proca y-polarization}
    \epsilon_\mu^y(\cv{k})=\left(0,0, 1, -\tan\alpha \sin\beta\right).
\end{align} 

These two covectors satisfy
\begin{align}
\label{eq:non-ONB}
    \IP{\epsilon^x, \epsilon^x}&=1+\tan^2\alpha\cos^2\beta=1+\mathcal{O}(\alpha^2)\, ,\nonumber\\
    \IP{\epsilon^y, \epsilon^y}&=1+\tan^2\alpha\sin^2\beta=1+\mathcal{O}(\alpha^2)\, ,\\
    \IP{\epsilon^x,\epsilon^y}&=\tan^2\alpha \cos\beta\sin\beta=\mathcal{O}(\alpha^2)\, ,\nonumber
\end{align}
hence they are orthogonal and normalized on the axis $\alpha=0$, and they form a basis of the two-dimensional space of transversal polarization in the laboratory frame in a small neighbourhood of $\cv{k_0}$.
We then choose
\begin{equation}
\label{eq:pol covector}
    \epsilon(\cv{k}) = \sigma_x \epsilon^x(\cv{k}) + \sigma_y \epsilon^y(\cv{k})
\end{equation}
where $\sigma=(\sigma_x,\sigma_y)\in \CC^2$ is now a `polarization qubit' state which we assume to satisfy $\abs{\sigma_{x}}^2+\abs{\sigma_{y}}^2=1$. 

It follows from \eqref{eq:non-ONB} that $\tilde{S}(\cv{k})=s(\cv{k})\epsilon(\cv{k})$ is not exactly normalized to one, but instead satisfies
\begin{equation}
    \norm{\tilde{S}}_{\fh_W}^2=\int\limits_{\RR^3}\frac{ \dd^3\cv{k}}{(2\pi)^3} \abs{s(\cv{k})}^2 \left(\abs{\sigma_x}^2+\abs{\sigma_y}^2+\tan^2\alpha\abs{\sigma_x\cos\beta+\sigma_y\sin\beta}^2\right)\, .
\end{equation}
One can use the Cauchy-Schwarz inequality to estimate 
\begin{equation}
  \abs{\sigma_x\cos\beta+\sigma_y\sin\beta}^2 \leq \abs{\sigma_x}^2+\abs{\sigma_y}^2 \, , 
\end{equation}
and combine this result with the properties of $\sigma$ and $s$ to control $\norm{\tilde{S}}_{\fh_W}$ by
\begin{equation}
\label{eq:bound norm correction}
    1\leq \norm{\tilde{S}}_{\fh_W}^2\leq 1+\tan^2 \alpha_{\rm max}= \sec^2  \alpha_{\rm max}\, .
\end{equation}
The state describing $n$ Proca particles with sharply peaked momentum carrying the polarization qubit $\sigma$ is then given by the state $\omega_{S,n}$ as defined in \eqref{eq:one particle state}, with $S=\norm{\tilde{S}}^{-1}_{\fh_W}\tilde{S}$.

\subsection{Malus' law} 

Malus' law describes the change in intensity of a polarized light beam after it passes through a linear polarizer  as a function of the orientation of the polarizer. 
Specifically, if the linear polarizer is oriented at an angle $\theta$ to the polarization of the beam then the transmitted intensity is equal to the
intensity of the incident beam multiplied by $\cos^2\theta$. We therefore expect the same form of modulation in the measurement of linearly polarized $n$-particle states with a polarization-sensitive detector, i.e., a polarization measurement.

In our model of a polarization measurement in quantum field theory, we replace the electromagnetic field by a complex Proca field. As a probe theory, we select a charged scalar field, linearly coupled to the Proca field as described in Sect.~\ref{sec:Proca--scalar}. The coupling zone is determined by the support of the coupling $1$-form $v$. We assume that in an appropriate inertial Cartesian coordinate system, which we will call the laboratory frame, the coupling 1-form $v$ takes the form $v=\lambda \rho u_\theta$, where 
\begin{equation}
    u_\theta=(0, \cos\theta, \sin\theta, 0)\, 
\end{equation}
is constant and transversal to the $z$-axis in the laboratory frame and $\lambda\in\RR$ is the coupling constant. 
 Meanwhile, $\rho\in C_0^\infty(M;[0,1])$ is a test-function which is equal to one in the centre of the coupling zone and vanishes outside of it.

This setup is motivated by wire-grid polarizers, a type of linear polarizer consisting of a number of parallel, thin wires. The free electrons in the wires are effectively restricted to move along the wires, allowing the polarizer to absorb or reflect, and hence effectively block, the component of light polarized parallel, but not transversal, to the wires \cite[Chapter 8.3.1]{Hecht}.

We assume that the Proca field is prepared in a state $\omega_{S,n}$ as in \eqref{eq:one particle state} that represents $n$ Proca particles in the same transversally polarized mode in the laboratory frame, so the particles travel approximately in the $z$-direction. The scalar field is prepared in its Minkowski vacuum state. The probe observable that will be measured is 
\begin{equation}
    B_f=\Phi(f)^*\Phi(f)-b_f\1_\Cf\, ,
\end{equation}
where $f\in C_0^\infty(M)$ is a test function supported in $M^+$, and $b_f\in \RR$ is a constant.
The observable $\Phi(f)^*\Phi(f)$ is chosen here as a proxy for a smearing of the renormalized Wick square, or vacuum polarization, ${:}\Phi^2{:}$. Much as in a classic paper of Fredenhagen \& Haag~\cite{FredenhagenHaag:1990}, this is done for simplicity, since it avoids the issue of renormalization.
The constant $b_f$ is a `calibration constant' that will be used to eliminate noise terms which are independent of the function $S\in L^2(\RR^3)\otimes \CC^4$ and particle number $n$ defining the state $\omega_{S,n}$ of the system, but  which can depend on the choice of $f$, i.e., the probe observable, and $\theta$, $\rho$, and $\lambda$, i.e., the detector setup.

Hence, the induced system observable is
\begin{align}
    \Ec_{\Omega_\Phi}( B_f)=\eta_{\Omega_\Phi}(\Theta(\1_\Wf\otimes B_f))\, ,
\end{align}
in terms of the scattering map $\Theta$, whose action we now determine. Since the theories are linearly coupled, one has
\begin{equation}
    \Theta (\1_\Wf\otimes \Phi(f)) = \Wc(h^-)\otimes \1_{\Cf}+ \1_{\Wf}\otimes \Phi(f^-),
\end{equation}
where 
\begin{align}
    \begin{pmatrix}  h^-\\ f^- \end{pmatrix}&=\begin{pmatrix} 0 \\ f
    \end{pmatrix}-\lambda \rho\begin{pmatrix}
    0 &  u_\theta \\ -u_\theta^\bullet & 0  \end{pmatrix} E^-_P\begin{pmatrix}
        0 \\ f
    \end{pmatrix}=\begin{pmatrix}
        -\lambda \rho u_\theta E^-_{K^{(0)}_\fm}(\id-\lambda^2\rho u_\theta^\bullet E^-_{WW}\rho u_\theta E^-_{K^{(0)}_\fm}) f\\
        (\id-\lambda^2 \rho u_\theta^\bullet E^-_{WW}\rho u_\theta E^-_{K^{(0)}_\fm}) f
    \end{pmatrix}\, 
\end{align}
by the result in \eqref{eq:scalar-Proca Green op}. In the following, we will drop the $(p)$-superscript on $K_\fm$, since only the version acting $0$-forms appears in the rest of this section.
 As $\Theta$ is a homomorphism, it follows that
\begin{align}
    \Theta (\1_\Wf\otimes B_f)=& \Wc(h^-)^*\Wc(h^-)\otimes \1_\Cf + \Wc(h^-)^*\otimes \Phi(f^-)\\\nonumber
    &+\Wc(h^-)\otimes \Phi(f^-)^*+ \1_\Wf\otimes \Phi(f^-)^*\Phi(f^-)-b_f\1_\Wf\otimes \1_\Cf\, .
\end{align}
 
Taking into account that $\Omega_\Phi$ has vanishing one-point functions, we obtain
\begin{align}
    \mathcal{E}_{\Omega_\Phi}(B_f)=\Wc(h^-)^*\Wc(h^-)+\left[\Omega_\Phi(\Phi(f^-)^*\Phi(f^-))-b_f\right]\1_\Wf\, .
\end{align}
Thus 
 \begin{align}
     \omega_{S,n}( \mathcal{E}_{\Omega_\Phi}(B_f)) 
     &= \Omega_\Phi( \Phi(f^-)^*\Phi(f^-))-b_f +\frac{1}{n!}\IP{a^*(S)^n\Omega_W,\pi_W(\Wc(h^-)^*\Wc(h^-))a^*(S)^n\Omega_W}_{\Hc_W}\, .
 \end{align} 

Since the first two terms are independent of the state of the Proca field, our focus will be on the last term. Using the representation of the smeared field operators, one obtains 
\begin{align}
&\frac{1}{n!}\IP{a^*(S)^n\Omega_W,\pi_W(\Wc(h^-)^*\Wc(h^-))a^*(S)^n\Omega_W}_{\Hc_W}\\\nonumber
&\qquad=\frac{1}{n!}\IP{\Omega_W, a(S)^n\left(a^*(\Kc\overline{h^-})+b(\Kc h^-)\right)\left( a(\Kc\overline{h^-})+b^*(\Kc h^-)\right)a^*(S)^n\Omega_W}_{\Hc_W} \, .
\end{align}
The only terms giving non-zero contributions are
\begin{align}
\label{eq:term 1}
    \frac{1}{n!} \IP{\Omega_W, a(S)^nb(\Kc h^-)b^*(\Kc h^-)a^*(S)^n\Omega_W}_{\Hc_W} 
\end{align}
and
\begin{align}
\label{eq:term 2}
    \frac{1}{n!} \IP{\Omega_W, a(S)^na^*(\Kc\overline{h^-})a(\Kc\overline{h^-})a^*(S)^n\Omega_W}_{\Hc_W} \, .
\end{align}
Let us first consider \eqref{eq:term 1}. Making use of the commutation relations, one finds by induction that
\begin{align}
    \frac{1}{n!}\IP{\Omega_W, a(S)^nb(\Kc h^-)b^*(\Kc h^-)a^*(S)^n\Omega_W}_{\Hc_W}=\norm{S}_{\fh_W}^{2n}\norm{\Kc h^-}_{\fh_W}^2=\norm{\Kc h^-}_{\fh_W}^2
\end{align}
is independent of $S$ and $n$. 

For \eqref{eq:term 2}, the commutation relations lead to the result
\begin{align}
   \frac{1}{n!}\IP{\Omega_W, a(S)^na^*(\Kc \overline{h^-})a(\Kc \overline{h^-})a^*(S)^n\Omega_W}_{\Hc_W}&=n\abs{\IP{\Kc \overline{h^-},S}_{\fh_W}}^2\norm{S}_{\fh_W}^{2(n-1)}\, .
\end{align}

Altogether, recalling the form of $v$ and $E^-_P$ given in \eqref{eq:scalar-Proca Green op}, we find
\begin{align}
\label{eq:determine bf}
     \omega_{S,n}(\mathcal{E}_{\Omega_0}(B_f))=& n\lambda^2\abs{\IP{\Kc\overline{\rho u_\theta E^-_{K_\fm}(\id-\lambda^2 \rho u_\theta^\bullet E^-_{WW}\rho u_\theta E^-_{K_\fm})f}, S}_{\fh_W}}^2+ \norm{\Kc h^-}_{\fh_W}^2\\\nonumber
    & +\norm{\Kc_\Phi f^-}_{\fh_\Phi}^2
    -b_f\, .
\end{align}

We observe that only the first term on the right-hand side of \eqref{eq:determine bf} depends on $S$ and $n$. Therefore, by setting
\begin{equation}
    b_f=\norm{\Kc h^-}^2_{\fh_W}+\norm{\Kc_\Phi f^-}^2_{\fh_\Phi}
\end{equation}
we can eliminate the terms independent of the polarization and particle number. Note that, as expected, $b_f$  does indeed depend on the choice of observable and detector setup through $f$ and $v=\lambda\rho u_\theta$.  This choice gives the exact result
\begin{equation} 
\label{eq:exp value}
 \omega_{S,n}(\mathcal{E}_{\Omega_0}(B_f))= n\lambda^2\abs{\IP{\Kc\overline{\rho u_\theta E^-_{K_\fm}(\id-\lambda^2 \rho u_\theta^\bullet E^-_{WW} \rho u_\theta E^-_{K_\fm})f}, S}_{\fh_W}}^2 . 
\end{equation}
Now, $\Kc \overline{\rho  u_\theta E^-_{K_\fm} \rho u_\theta^\bullet E^-_{WW}\rho u_\theta E^-_{K_\fm}f}$  converges to a  limit in $\fh_W$ as $\lambda\to 0$, as shown in Appendix~\ref{app:convergence}. Thus
\begin{equation}
\omega_{S,n}(\Ec_{\Omega_0}(B_k)) = R(n,S,\rho, u_\theta,f, \lambda) + O(\lambda^{4})
\end{equation} 
where
\begin{align}
    R(n,S,\rho, u_\theta,f, \lambda):=n\lambda^2\abs{\IP{\Kc \overline{ \rho u_\theta E^-_{K_\fm}f}, S}_{\fh_W}}^2\, .
\end{align} 

Recall that $S$ is of the form  $S=\norm{s\epsilon}_{\fh_W}^{-1}s\epsilon$, where $\epsilon$ is given by \eqref{eq:pol covector} in the detector frame and determined by a choice of normalized qubit state $\sigma\in \CC^2$.
 Then we find
\begin{align}
    -\eta^{-1}(\Pi(\cv{k})u_\theta, \epsilon(\cv{k}))=-\eta^{-1}(u_\theta,\epsilon(\cv{k})) =\sigma_x\cos\theta+\sigma_y\sin\theta\, .
    \end{align}
    This allows us to write
\begin{align}
\label{eq:Malus law}
R(n,S,\rho, u_\theta,f, \lambda)
    =&n\lambda^2\abs{\sigma_x\cos\theta+\sigma_y\sin\theta}^2\norm{s\epsilon}^{-2}_{\fh_W}\abs{A(s,\rho, f)}^2\, ,
\end{align}
where we have defined
\begin{align}
\label{eq:form factor}
    A(s,\rho, f):&=\int_{\RR^3}\frac{ \dd^3\cv{k}}{(2\pi)^3\sqrt{2\omega(\cv{k})}}\widehat{(\rho E^-_{K_\fm}f)}(-\omega(\cv{k}), -\cv{k})s(\cv{k})\,  .
\end{align} 
We see that the particle number $n$ acts like a luminosity, contributing an overall multiplicative factor.
Note further that $A(s, \rho,f)$ is independent of $\epsilon$, $u_\theta$, and $n$. In fact, if one replaces the Proca field by a second complex scalar field $\Psi$ of mass $m$, so that the two fields are coupled by the bilinear coupling term $2\lambda \rho\Re(\bar{\Psi}\Phi)$, and considers the same induced observable in the scalar version of $\omega_{S,n}$, one obtains exactly $\lambda^2 n\abs{A(s,\rho,f)}^2$ as the expectation value at leading order. See Appendix~\ref{sec:scalar case} for more details. 

To illustrate the physical meaning of $A(s,\rho,f)$, we consider displacing the detector and observable by replacing $(\rho E^-_{K_\fm} f)(x)$ by $ (\rho E^-_{K_\fm} f)(x-\delta e)$, where $0\neq e=(e^0, \Vec{e})\in \RR^4$ is a tangent vector chosen such that $\eta(e,e)\in \{\pm 1, 0\}$, and $\delta>0$ is a constant. In what follows we write $3$-vectors as $\Vec{e}$ and $3$-covectors as $\cv{k}$ and so on,  using minus the (inverse) Euclidean metric to lower (raise) indices.
Under this displacement, we find that 
\begin{align}
    A(s,\rho,f)\to A_\delta(s,\rho,f):=\int_{\RR^3}\frac{ \dd^3\cv{k}}{(2\pi)^3\sqrt{2\omega(\cv{k})}}e^{-\ii \delta k\cdot e}\widehat{(\rho E^-_{K_\fm}f)}(-\omega(\cv{k}), -\cv{k})s(\cv{k})\,  .
\end{align} 
We note that the integrand is a smooth, compactly supported function on $\RR^3$, and that  the phase function $\phi_e(\cv{k}) := -k\cdot e=-(\omega(\cv{k})e^0+\cv{k} \Vec{e})$ is a real-valued, smooth function. Moreover, $\phi_e(\cv{k})$ satisfies
\begin{align}
   \dd\phi_e(\cv{k})=\frac{e^0 \Vec{k}}{\omega(\cv{k})}-\Vec{e}\\
    \det(H_{\phi_e}(\cv{k}))=-\frac{(e^0)^3m^2}{\omega(\cv{k})^5}\, ,
\end{align}
where $H_{ \phi_e}(x)$ is the Hessian matrix of second derivatives. We can see that $\dd\phi_e$ can only vanish if $e^0\neq 0$, in which case the determinant of the Hessian is non-zero. To be more precise, $\dd\phi_e(\cv{k})=0$ iff  $e^0\neq 0$ and $\Vec{k}/\omega(\cv{k})=\Vec{e}/e^0$,  which in particular requires $\Vec{k}$ and $\Vec{e}$ to be collinear  and
consequently $e = \pm m^{-1}k^\sharp$, so in particular $e$ is timelike.

We consider now the behaviour of $A_\delta(s,\rho,f)$ as $\delta \to \infty$  using a stationary phase estimate. First, assume that $e$ is such that $\dd\phi_e\neq 0$ on $\supp s$.
Then, by \cite[Theorem 7.7.1]{Hoermander1} for any $N\in \NN$ there is a constant $C_N$ so that
\begin{equation}
    \abs{A_\delta(s,\rho,f)}\leq C_N \delta^{-N}\, 
\end{equation}
as $\delta\to \infty$.
Next, if  $e=\pm m^{-1}k_0^\sharp$ for some $\cv{k_0}\in \supp s$, then by \cite[Theorem 7.7.5]{Hoermander1}, one has
\begin{equation}
   A_\delta(s, \rho,f)= e^{\mp \ii(\delta m -\pi/4)}\frac{m^2\sqrt{\omega(\cv{k_0})}}{4(\pi\delta m)^{3/2}}  s(\cv{k_0})(\widehat{\rho E^-_{K_\fm}f})(-k_0)
  +\mathcal{O}((\delta m)^{-5/2})\, .
\end{equation}

This result can be understood as follows: If we move the detector and observable along the direction in which the particles travel, we obtain the loss expected from the spreading of the particle beam due to the imperfect collimation. In contrast, if we displace the detector and observable in any other direction, $\abs{A_\delta(s,\rho, f)}$ decays faster than any polynomial. In this sense, $A(s,\rho,f)$ can be understood as a measure of the extent to which the particle beam meets the detector.

 Even though the \emph{collimation factor} $\|s\epsilon\|_{\fh_W}$ appearing in $R(n,S,\rho,u_\theta,f, \lambda)$ depends on the choice of $\sigma$, it obeys uniform bounds in $\sigma$, see \eqref{eq:bound norm correction}. Moreover, these bounds can be tightened by concentrating the support of $s$ around the $z$-axis. Its deviation from unity can be regarded as a small correction that appears due to the non-orthonormality of the basis $(\epsilon^x,\epsilon^y)$ for momenta $\cv{k}$ lying off the $z$-axis.

Apart from this small correction, the dependence of $R(n,S,\rho,u_\theta,f,\lambda)$ on the polarization $\epsilon$ of the Proca particle and $u_\theta$ of the coupling covector is entirely contained in the first factor in \eqref{eq:Malus law}, which is the quantum mechanical probability that the qubit state $\sigma$ passes a test for polarization direction $(\cos\theta,\sin\theta)$.

 If $\sigma\in \RR^2$, it is of the form $\sigma=(\cos\eta, \sin\eta)$. In this case the state $\omega_{S,n}$ describes  linearly polarized Proca particles and one has
\begin{equation}
    \sigma_x\cos\theta+\sigma_y\sin\theta=\cos\eta\cos\theta+\sin\eta\sin\theta=\cos(\theta-\eta)\, .
\end{equation}
We therefore find that in a linearly polarized $n$-particle state of the Proca field the expectation value of the induced observable $\Ec_{\Omega_\Phi}( B_f)$ depends on the cosine squared of the angle $\theta-\eta$ between the polarization covector $\epsilon$ projected to the $xy$-plane of the laboratory frame and the coupling covector $u_\theta$.
This is the direct analogue of Malus' law for the intensity of a polarized light beam going through a linear polarizer.

 Another illuminating example is to choose $\sigma=2^{-1/2}(1,\pm i)$, corresponding to circular or anticircular polarization of the Proca particles. In this case, one quickly obtains 
\begin{equation}
    \abs{\sigma_x\cos\theta+\sigma_y\sin\theta}^2=\frac{1}{2}\, ,
\end{equation}
as one would have expected \cite{Hecht}.

We conclude that our measurement scheme provides a model for a polarization-sensitive detector. 

\section{Conclusion}
\label{sec:Conc}
In this work, we have developed a method of auxiliary fields, which provides practical criteria for testing the (semi)-Green-hyperbolicity and the decomposability of differential operators. An application of this method revealed that the equations of motion for  (a) the charged Proca in the presence of a background electromagnetic field and with arbitrary anomalous magnetic moment,  (b) a Proca multiplet with variable mass matrix, and  (c) a Proca field linearly coupled to a scalar field are all governed by Green-hyperbolic, $\Nc^\pm$-decomposable differential operators. 
We have also discussed how theories based on Green-hyperbolic operators can be quantized, even in the absence of the mild additional assumptions that are usually made.

Making use of these results, we demonstrated that the coupled Proca--scalar system can be used to model a polarization-sensitive detector. To be precise, we considered the scalar field as a probe for the Proca field, and observed that the leading-order response in the measurement of polarized $n$-particle states of the Proca field is in accordance with Malus' law up to a collimation factor.
This computation serves as a proof of concept, that can be extended in various ways, for example by considering charged Proca in background electromagnetic fields to model other ways of manipulating the field, or considering configurations of multiple detectors.

\section*{Acknowledgments}
CJF's work is partly supported by EPSRC Grant EP/Y000099/1 to the University of York, and additional support from the Institut Henri Poincaré (UAR 839 CNRS-Sorbonne Université), and LabEx CARMIN (ANR-10-LABX-59-01) is gratefully acknowledged. CK is funded by the Deutsche Forschungsgemeinschaft (DFG, German
Research Foundation) – Projektnummer 531357976. It is a pleasure to thank Eli Hawkins, Onirban Islam, Daan Janssen, Benito Ju\'arez Aubry, Igor Khavkine, Valter Moretti, Kasia Rejzner, Alexander Schenkel, Alexander Strohmaier, Rainer Verch and Berend Visser for insightful comments at various stages of this work.

For the purpose of open access, the authors have applied a creative commons attribution (CC BY) licence to any author accepted manuscript version arising. 

\section*{Declarations}
\subsection*{Funding}
CJF's work is partly supported by EPSRC Grant EP/Y000099/1 to the University of York, and additional support from the Institut Henri Poincaré (UAR 839 CNRS-Sorbonne Université), and LabEx CARMIN (ANR-10-LABX-59-01) is gratefully acknowledged. CK is funded by the Deutsche Forschungsgemeinschaft (DFG, German
Research Foundation) – Projektnummer 531357976.
\subsection*{Conflict of Interest}
The authors have no competing interests to declare that
are relevant to the content of this article.
\subsection*{Ethical statement, Informed consent, Data availability}
This research involved neither humans nor animals, nor was any data produced beyond what is written in the article. Therefore, the authors have nothing to declare regarding ethical statements, informed consent, or data availability that is relevant to the content of this article.

\appendix
\section{Derivation of \eqref{eq:scalar-Proca Green op}}
\label{sec:App A}
In this appendix, we will give some details on the derivation of \eqref{eq:scalar-Proca Green op}. 

We start by considering the retarded and advanced Green operators for the operator $T$ given in \eqref{eq:sP-norm hyp operator}. Decomposing $\Gamma^\infty(\Lambda^1M\oplus B)\cong \Gamma^\infty(\Lambda^1M) \oplus \Gamma^\infty(\Lambda^1M)\oplus \Gamma^\infty(\Lambda^0M)$, we can write
\begin{align}
    E^\pm_{T}&=\begin{pmatrix}
        \tilde{E}_{VV}^\pm  & \tilde{E}_{VW}^\pm  & \tilde{E}_{ V\phi}^\pm \\
        \tilde{E}_{WV}^\pm & \tilde{E}_{WW}^\pm & \tilde{E}_{ W\phi}^\pm\\
        \tilde{E}_{\phi V}^\pm & \tilde{E}_{\phi W}^\pm & \tilde{E}_{\phi \phi}^\pm
    \end{pmatrix}\, .
\end{align}
 Plugging this decomposition into \eqref{eq:sP-Q_Green op} and in turn into \eqref{eq:sP- Green op formula}, we find that the advanced and retarded Green operators for the coupled Proca--scalar operator $P$ \eqref{eq:P real Proca scalar} can be written as
 \begin{align}
 \label{eq:sP-Green op decomp}
     E^\pm_P=\begin{pmatrix}
        \tilde{E}^\pm_{WW}(1-m^{-2}\dd\delta) & \tilde{E}_{W\phi }^\pm+\tilde{E}^\pm_{WV} \dd\\
        \tilde{E}^\pm_{\phi W}(1-m^{-2}\dd\delta) & \tilde{E}_{\phi \phi}^\pm+\tilde{E}^\pm_{\phi V} \dd
        \end{pmatrix}\, .
 \end{align}
 To analyse this further, we turn to $E^\pm_T$. Using the same decomposition as before, we can use the properties G1 and G2 of $E^\pm_T$ together with the explicit form of $T$ to derive the following equations:

 From G1, we obtain the following identities on compactly supported sections of $\Lambda^1M$ or $\Lambda^0M$
 \begin{align}
     v^\bullet \tilde{E}^\pm_{WV}&=K^{(0)}_{\fm}\tilde{E}^\pm_{\phi V}\, ,\\
     v^\bullet \tilde{E}^\pm_{WW}&=K^{(0)}_{\fm}\tilde{E}^\pm_{\phi W}\, ,\\
     \dd v^\bullet \tilde{E}^\pm_{WW}&=K^{(1)}_{\fm}\tilde{E}^\pm_{VW}\, ,\\
     \dd v^\bullet \tilde{E}^\pm_{W\phi}&=K^{(1)}_{\fm}\tilde{E}^\pm_{V\phi}
 \end{align}
 and acting from the left on the first two with $E^\pm_{K^{(0)}_{\fm}}$ and the second two with $E^\pm_{K^{(1)}_{\fm}}$, we can express $\tilde{E}^\pm_{\phi V}$, $\tilde{E}^\pm_{\phi W}$, $\tilde{E}^\pm_{VW}$, and $\tilde{E}^\pm_{V\phi}$ in terms of $\tilde{E}^\pm_{WW}$,
 $\tilde{E}^\pm_{WV}$, $\tilde{E}^\pm_{W\phi}$,
  $v$, differential operators, and Green operators for the Klein--Gordon operator (we also use the compact support of $v$ and the support property G3). From G2, we obtain the additional identities 
 \begin{align}
     -m^{-2}\tilde{E}^\pm_{WW}\dd v^\bullet&=\tilde{E}^\pm_{WV}K^{(1)}_{\fm}\, , \\
     -\tilde{E}^\pm_{WW}(v-m^{-2} \dd(\delta v))&=\tilde{E}^\pm_{W\phi}K^{(0)}_{\fm}
 \end{align}
 on sections with compact support.
 They can be composed with $E^\pm_{K^{(1)}_{\fm}}$ or $E^\pm_{K^{(0)}_{\fm}}$ from the right respectively, to obtain expressions for $\tilde{E}^\pm_{WV}$ and $\tilde{E}^\pm_{W\phi}$. 
Furthermore, we take from G1 the identities
\begin{align}
    K^{(1)}_\fm \tilde{E}^\pm_{VV}-\dd v^\bullet \tilde{E}^\pm_{WV}&=\id\, ,\\
    K^{(0)}_\fm \tilde{E}^\pm_{\phi\phi}-v^\bullet \tilde{E}^\pm_{W\phi}&=\id\, ,\\
    m^{-2}\dd v^\bullet \tilde{E}^\pm_{VW}+K^{(1)}_m\tilde{E}^\pm_{WW}+(v-m^{-2}\dd(\delta v))\tilde{E}^\pm_{\phi W}&=\id\, ,
\end{align}
and from G2 the identity
\begin{align}
    \tilde{E}^\pm_{WW}K^{(1)}_{m}-\tilde{E}^\pm_{WV}\dd v^\bullet-\tilde{E}^\pm_{W\phi}v^\bullet=\id\, .
\end{align}
Plugging the results for the off-diagonal components of $E^\pm_T$ into the first two and acting with $E^\pm_{K^{(1)}_\fm}$ or $E^\pm_{K^{(0)}_\fm}$ from the left, respectively, we obtain expressions for $\tilde{E}^\pm_{VV}$ and $\tilde{E}^\pm_{\phi\phi}$. The third equation from G1 and the equation for G2, combined with the previous result, gives \eqref{eq:sP-coupled Green op-WW}. 

In the last step, we use the fact that $E^\pm_{K^{(1)}} \dd=\dd E^\pm_{K^{(0)}}$ and $\dd v^\bullet \dd -\dd(\delta v)=-\dd \delta v$ as an operator acting on $0$-form fields. The same identities can be used to compute
\begin{align}
    \tilde{E}^\pm_{W\phi}+\tilde{E}^\pm_{WV}\dd&=-\tilde{E}^\pm_{WW}(\id-m^{-2}\dd \delta) v E^\pm_{K^{(0)}_\fm}\\
   \tilde{E}^\pm_{\phi\phi}+\tilde{E}^\pm_{\phi V} d&= E^\pm_{K^{(0)}_{\fm}}-E^\pm_{K^{(0)}_\fm}v^\bullet \tilde{E}^\pm_{WW}(\id-m^{-2}\dd \delta)v E^\pm_{K^{(0)}_\fm}\\
   \tilde{E}^\pm_{\phi W}(1-m^{-2}\dd \delta)&=E^\pm_{K^{(0)}_\fm}v^\bullet \tilde{E}^\pm_{WW}(\id-m^{-2}\dd\delta)\, .
\end{align}
Plugging these results into \eqref{eq:sP-Green op decomp}, one obtains \eqref{eq:scalar-Proca Green op}.

\section{Bounds for higher-order terms in Malus' law}
\label{app:convergence}
In this section, we will make precise in which sense the higher-order terms in \eqref{eq:exp value} are of order $\lambda^4$. As mentioned there, it is enough to show that
$\Kc \overline{\rho u_\theta  E^-_{K_\fm^{(0)}} \rho u_\theta^\bullet E^-_{WW}\rho u_\theta E^-_{K^{(0)}_\fm}f}$ converges in $\fh_W$ as $\lambda\to 0$.
As in  Sect.~\ref{sec:Malus}, we will drop the superscripts on $K_\fm$ and $K_m$, since only $K_\fm^{(0)}$ and $K_m^{(1)}$ play a role in this derivation.

We note that the only $\lambda$-dependence is contained in the operator $E^-_{WW}$.  We will make use of various facts that will be proved below. First, we claim that 
\begin{equation}\label{eq:claim1}
\|\Kc u_\theta h\|_{\fh_W}\le C \| h\|_{H^3_{\tilde{K}}(M)}
\end{equation}
for all compact $\tilde{K}$ and $h\in C^\infty_{\tilde{K}}(M)$. 
Here and below, $H^s_{\tilde{K}}(M)$ is the Sobolev space of order $s$ of functions supported in $\tilde{K}$, while
$H^s_{\tilde{K}}(M,\Lambda^1M)$ is the analogous Sobolev space of $1$-forms. (Our definitions follow those in the Appendix of~\cite{Fewster:2023}.)

Second, if $\tilde{K}$ contains $\supp \rho$ then $\rho E^-_{K_\fm}$ defines a bounded operator between $H^2_{\tilde{K}}(M)$ and $H^3_{\tilde{K}}(M)$. 
Third, we will show that $\rho u_\theta^\bullet E^-_{WW}$ defines a bounded operator $\mathcal{B}(H^3_{\tilde K}(M;\Lambda^1M),H^{2}_{\tilde{K}}(M))$, converging to $\rho u_\theta^\bullet E^-_{P_\Mb}$ in operator norm as $\lambda\to 0$.
Here $E^-_{P_\Mb} =E^-_{K_m} (1-m^{-2}\dd\delta)$ is the advanced Green operator for the real, uncoupled Proca field. Fourth, for each $f\in C_0^\infty(M)$ one has
$\rho u_\theta E^-_{K_\fm}f\in \Gamma_{\tilde{K}}^\infty (\Lambda^1M) \subset H^3_{\tilde K}(M;\Lambda^1M)$.
Given these facts, it follows that 
\begin{equation}
  \Kc \overline{\rho u_\theta  E^-_{K_\fm} \rho u_\theta^\bullet E^-_{WW}\rho u_\theta E^-_{K_\fm}f} \to
  \Kc \overline{\rho u_\theta  E^-_{K_\fm} \rho u_\theta^\bullet E^-_{P_{\Mb}}\rho u_\theta E^-_{K_\fm}f}
\end{equation}
in $\fh_W$ as $\lambda\to 0$ for each fixed $f\in C_0^\infty(M)$, which is our required overall result.

It remains to establish the four facts above. Indeed, the second is a standard fact about Green operators of normally hyperbolic operators, which gain one derivative
by the results of \cite[Theorem 6.5.3]{DuiHoer_FIOii:1972} (see also \cite{IslamStrohmaier:2020} for a bundle version). Meanwhile the fourth fact is trivial so only the first and third need be shown.

Starting with the first fact, let $h\in C^\infty_{\tilde{K}}(M)$ and consider $\Kc u_\theta h$. Then
\begin{equation}
\label{eq:C1}
       \norm{\Kc u_\theta h}_{\fh_W}^2=\int\limits_{\RR^3}\frac{\dd^3\cv{k}}{(2\pi)^3 2\omega(\cv{k})} \phi_\theta(\cv{k})\abs{\hat{h}(\omega(\cv{k}),\cv{k})}^2\, ,
\end{equation}
 where 
 \begin{equation}
     \phi_\theta(\cv{k})=-\eta^{-1}(\Pi u_\theta, \Pi u_\theta)=1+m^{-2}\abs{\cv{k}}^2\sin^2\alpha\cos^2(\beta-\theta) \le 1+m^{-2}\abs{\cv{k}}^2 = m^{-2}\omega(\cv{k})^2
 \end{equation}
can be computed using the parametrization \eqref{eq:k param} of the 4-covector $k_\mu$.
We can estimate \eqref{eq:C1} by
\begin{align}
       \norm{\Kc u_\theta h }_{\fh_W}^2&\le \frac{1}{2m^2}\int\limits_{\RR^3}\frac{\dd^3\cv{k}}{(2\pi)^3 \omega(\cv{k})^5} \abs{\omega(\cv{k})^3\widehat{h}(\omega(\cv{k}),\cv{k})}^2\,\\\nonumber
       &=  \frac{1}{2m^2}\int\limits_{\RR^3}\frac{\dd^3\cv{k}}{(2\pi)^3 \omega(\cv{k})^5}  \abs{\widehat{\partial_t^3h}(\omega(\cv{k}),\cv{k})}^2\,\\\nonumber
       &\leq C \norm{\partial_t^3 h}_{L^1(\RR^4)}^2\, ,
\end{align}
 for some constant $C>0$, where we have used the notation $t$ for the $x^0$-coordinate. In the last step, we have used the facts that $\omega^{-5}\in L^1(\RR^3,\dd^{3}\cv{k})$ and that the $L^\infty$-norm of a function's Fourier transform is bounded by its $L^1$-norm. Note that one could make do with a smaller number of derivatives, but this result is sufficient for our purposes.
 Finally, let $\chi\in L^2(\RR^4)$ so that $\chi(x)=1$ on $\tilde{K}$. Then we have 
 \begin{align}
      \norm{\partial_t^3h}_{L^1(\RR^4)}&=\norm{\chi \partial_t^3 h}_{L^1(\RR^4)}\leq \norm{\partial_t^3 h}_{L^2(\RR^4)}\norm{\chi}_{L^2(\RR^4)} 
      \leq C'   \norm{h}_{H^3_{\tilde{K}}(M)} , 
 \end{align} 
and combining with the previous estimate the first fact is established.

Lastly, recall that $E^\pm_{WW}=\tilde{E}^{\pm}_{WW}(1-m^{-2}\dd\delta)$ by \eqref{eq:E^pm_WW}, so the $\lambda$-dependence of $E^\pm_{WW}$ is contained in $\tilde{E}^{\pm}_{WW}$.
The operators $\tilde{E}^\pm_{WW}$ satisfy the equations \eqref{eq:sP-coupled Green op-WW}, i.e., they invert the non-local operator $O^\pm=K_m+\lambda^2(\id-m^{-2}\dd \delta )v_\theta E^\pm_{K_\fm}v_\theta^\bullet$ where $v_\theta=\rho u_\theta\in \Gamma_0^\infty(\Lambda^1M)$ is a fixed, real one-form.  Acting on the equation $\tilde{E}^\pm_{WW}O^\pm=\id$ with $E^\pm_{K_m}$ from the right leads to the equation
\begin{equation}
\label{eq:C2}
        \tilde{E}^\pm_{WW}(\id+\lambda^2 (\id-m^{-2}\dd\delta)v_\theta E^\pm_{K_\fm} v_\theta ^\bullet E^\pm_{K_m})=E^\pm_{K_m}\, ,
\end{equation}
which holds on $\Gamma_0^\infty(\Lambda^1M)$. 

Now the operator $(\id-m^{-2}\dd\delta)v_\theta E^\pm_{K_\fm} v_\theta^\bullet E^\pm_{K_m}$ extends to a bounded operator on the Hilbert space $H^s_{\tilde K} (M;\Lambda^1M)$ for any $s\in \RR$ and any compact set $\tilde K\subset M$ that contains $\supp \rho$. Consequently, the operator 
\begin{align}
    (\id+\lambda^2 (\id-m^{-2}\dd\delta)v_\theta E^\pm_{K_\fm} v_\theta^\bullet E^\pm_{K_m})^{-1}: H^s_{\tilde K}(M;\Lambda^1M)\to H^s_{\tilde K}(M;\Lambda^1M)
\end{align}
is the resolvent of a bounded operator, which is well-defined and bounded if $|\lambda|<\lambda_0$ for some $\lambda_0>0$ which may depend on $s$; furthermore the resolvent is holomorphic in $\lambda$ within this disk.

From now on, let us fix some arbitrary $s\in \RR$ and let us assume that $0<\lambda<\lambda_0$, so that the resolvent is a bounded operator on $H^s_{\tilde K}(M;\Lambda^1M)$. Then acting on \eqref{eq:C2} with the resolvent from the right, and by contraction with  the compactly supported 1-form $v_\theta$ from the left,
one obtains the formula 
 \begin{align}
      v_\theta^\bullet \tilde{E}^\pm_{WW}= v_\theta^\bullet E^\pm_{K_m} (\id+\lambda^2 (\id-m^{-2}\dd\delta)v_\theta E^\pm_{K_\fm} v_\theta^\bullet E^\pm_{K_m})^{-1}\, 
    \end{align}
 on $\Gamma^\infty_{\tilde{K}}(\Lambda^1M)\subset H^s_{\tilde K}(M;\Lambda^1M)$.
 By the preceding discussion, we can thus see that $v_\theta^\bullet \tilde{E}^\pm_{WW}$ extends to a bounded operator from $H^s_{\tilde K}(M;\Lambda^1M)$ to $H^{s+1}_{\tilde{K}}(M)$, and that the $\lambda$-dependence of this operator is determined by the $\lambda$-dependence of the resolvent. Consequently, as $\lambda\to 0$, $v_\theta^\bullet \tilde{E}^\pm_{WW}$ converges in norm to $v_\theta^\bullet E^\pm_{K_m} \in \mathcal{B}(H^s_{\tilde K}(M;\Lambda^1M),H^{s+1}_{\tilde{K}}(M))$.
 Recalling \eqref{eq:E^pm_WW}, we notice that no additional $\lambda$-dependence is introduced in going from $\tilde{E}^\pm_{WW}$ to $E^\pm_{WW} =\tilde{E}^\pm_{WW}(1-m^{-2}\dd\delta)$, but that the transition loses up to two derivatives. We conclude that,  for any $s\in\RR$, $v_\theta^\bullet E^-_{WW}:\,   H^s_{\tilde K}(M;\Lambda^1M)\to H^{s-1}_{\tilde{K}}(M)$ is a bounded operator, which converges in norm to $v_\theta^\bullet E^-_{P_\Mb}$ in $\mathcal{B}(H^s_{\tilde K}(M;\Lambda^1M),H^{s-1}_{\tilde{K}}(M))$. While one can make stronger statements concerning differentiability in $\lambda$, the convergence result is all that is needed for our purposes. The third fact above is the special case $s=3$.

\section{Two coupled scalars}
\label{sec:scalar case}
 In this section, we will  show that the factor $A(s,\rho, f)$  also arises in the response of our probe field  when it is linearly coupled to a scalar field. 

 For this, we consider as our system theory a complex scalar of mass $m$ (i.e., the  same mass as the Proca field in Section~\ref{sec:Malus}). The uncoupled theory is then described by $K_m\oplus K_\fm$ (we drop the $(0)$-superscript on the Klein-Gordon operator since we will only deal with scalar functions in this section), and we assume that the differential operator describing the coupled theory takes the form 
 \begin{equation}
     P=\begin{pmatrix}
         K_m & \lambda\rho \\ \lambda\rho & K_\fm
     \end{pmatrix}\, ,
 \end{equation}
 where $\lambda\in\RR$ and $\rho\in C^\infty_0(M;[0,1])$ are the same coupling constant and scalar test function as considered for the Proca-scalar system. 
 The probe observable we consider for this system is 
 \begin{align}
     O_f=\Phi(f)^*\Phi(f)-c_f\1_\Cf\, ,
 \end{align}
where $f\in C_0^\infty(M^+)$ is the same test function as in the Proca-scalar system, and $c_f\in \RR$ is a `calibration constant' which may depend on $f$, $\lambda$, and $\rho$.

A calculation along the same lines as in Section~\ref{sec:Malus} then shows that the induced observable in this case is given by
\begin{align}
    \Ec_{\Omega_\Phi}(O_f)=\Psi(h^-)^*\Psi(h^-)+\left(\Omega_\Phi(\Phi(f^-)^*\Phi(f^-))-c_f\right)\1_\Sf\, ,
\end{align}
where $\Sf$ is the algebra of observables of the system, and $\Psi(f)$, $\Psi(f)^*$ are its generators. We have also defined
\begin{align}
    \begin{pmatrix}
        h^-\\ f^-
    \end{pmatrix}=\begin{pmatrix}
        -\lambda\rho E^-_{K_\fm} f+\mathcal{O}(\lambda^3) \\ f+\lambda^2\rho E^-_{K_m}\rho E^-_{K_\fm}f+\mathcal{O}(\lambda^4)
    \end{pmatrix}\, ,
\end{align}
which follows from a Born expansion for the coupled propagator, see for example \cite{FewVer_QFLM:2018} or \cite{FewsterJubbRuep:2022}. 

As before, we consider system and probe in states in the folium of the Minkowski vacuum, so that the probe fields $\Phi(f)$ are represented as before, and the system fields $\Psi$ are represented in the same way as the probe fields with $\fm$ replaced by $m$.

We then select the system preparation state to be the state 
\begin{align}
    \omega_{s,n}(A)=\frac{1}{n!}\IP{a_\Psi^*(s)^n\Omega_\Psi, \pi_\Psi(A)a_\Psi^*(s)^n\Omega_\Psi}_{\Hc_\Psi}\, ,
\end{align}
where $\Hc_\Psi$, $\Omega_\Psi\in \Hc_\Psi$, $\pi_\Psi$, and $a^*_\Psi$ are the Hilbert space, Minkowski vacuum state, representation of $\Sf$ on $\Hc_\Psi$, and creation operator for the complex scalar field $\Psi$, respectively.
The function $s\in L^2(\RR^3, (2\pi)^{-3}\dd^3\cv{k})$ is chosen to be the same as in the scalar-Proca system. Thus this is an $n$-particle single-mode state with momentum concentrated around the $z$-axis.

Using the commutation relations of the creation and annihilation operators, we then obtain
\begin{align}
    \omega_{s,n}(\Ec_{\Omega_\Phi}(O_f))=n\abs{\IP{\Kc_\Psi\overline{h^-}, s}_{\fh_\Psi}}^2+\norm{\Kc_\Psi h^-}_{\fh_\Psi}^2+\norm{\Kc_\Phi f^-}_{\fh_\Phi}^2-c_f\, .
\end{align}
By setting $c_f=\norm{\Kc_\Psi h^-}_{\fh_\Psi}^2+\norm{\Kc_\Phi f^-}_{\fh_\Phi}^2$, which is independent of the choice of $s$ and $n$, and using the Born expansion for $h^-$, we obtain
\begin{align}
    \omega_{s,n}(\Ec_{\Omega_\Phi}(O_f))&=\lambda^2n\abs{\int \frac{\dd^3\cv{k}}{(2\pi)^3} s(\cv{k})\overline{\Kc_\Psi(\rho E^-_{K_\fm} f)}(\cv{k})}^2+\mathcal{O}(\lambda^3)\\
    &=\lambda^2n\abs{\int \frac{\dd^3\cv{k}}{(2\pi)^3(2\omega(\cv{k}))^{1/2}} s(\cv{k})\widehat{\rho E^-_{K_\fm} f}(-\omega(\cv{k}),-\cv{k})}^2+\mathcal{O}(\lambda^3)\,. 
\end{align}
Comparing with \eqref{eq:form factor}, one can see immediately that this is indeed equal to $\lambda^2\abs{A(s,\rho, f)}^2+\mathcal{O}(\lambda^3)$.

{\small
\bibliographystyle{spmpsci_mod_href}
\bibliography{covariant}
}
\end{document}